\documentclass[11pt, a4paper]{article}
\usepackage[utf8]{inputenc}
\usepackage{dcolumn,lscape}
\usepackage{amsmath,longtable,multicol,dcolumn,tabularx,graphicx,amssymb}
\usepackage{exscale,amsthm,multirow,rotating,booktabs}
\usepackage{natbib}
\usepackage{bbm,bm}
\usepackage{tikz,dsfont,float}
\usepackage{epstopdf}
\usepackage[T1]{fontenc}
\usepackage{ae,aecompl}
\usepackage{subfigure}
\usepackage{array}
\usepackage{mathrsfs}
\usepackage[singlespacing]{setspace}
\usepackage{rotating}
\usepackage{multirow}
\usepackage{threeparttable}
\usepackage{enumerate}
\usepackage{subfigure}
\usepackage{caption}
\usepackage{hyperref}

\usepackage{geometry}
\reversemarginpar

\hypersetup{
    colorlinks=true,
    linkcolor=blue,
    citecolor=blue,
    filecolor=magenta,      
    urlcolor=cyan,
    pdftitle={review paper},
    pdfpagemode=FullScreen,
}

\usepackage{subfigure}
\newcommand{\tr}{\text{tr}}

\newcommand{\elwz}{e^{\lambda_0 \mathbf{W}}}
\newcommand{\elwzp}{e^{\lambda_0 \mathbf{W}'}}
\newcommand{\elwzn}{e^{-\lambda_0 \mathbf{W}}}
\newcommand{\elwznp}{e^{-\lambda_0 \mathbf{W}'}}
\newcommand{\elw}{e^{\lambda \mathbf{W}}}

\newcommand{\ermn}{e^{-\rho \mathbf{M}}}
\newcommand{\ermz}{e^{\rho_0 \mathbf{M}}}
\newcommand{\ermzn}{e^{-\rho_0 \mathbf{M}}}
\newcommand{\ermzp}{e^{\rho_0 \mathbf{M}'}}
\newcommand{\ermzpn}{e^{-\rho_0 \mathbf{M}'}}
\newcommand{\nermz}{e^{-\rho_0 \mathbf{M}}}
\newcommand{\erm}{e^{\rho \mathbf{M}}}
\newcommand{\ermp}{e^{\rho \mathbf{M}'}}
\newcommand{\p}{\mathbf{P}}
\newcommand{\F}{\mathbf{F}}
\newcommand{\pphi}{\boldsymbol{\Phi}}
\newcommand{\D}{\mathbf{D}}

\newcommand{\eeta}{\boldsymbol{\eta}}
\newcommand{\var}{\mathrm{Var}}
\newcommand{\G}{\mathbf{G}(\zet)}
\newcommand{\GP}{\mathbf{G}^{'}(\zet)}
\newcommand{\GZPI}{\mathbf{G}^{-1'}(\zet_0)}
\newcommand{\GZ}{\mathbf{G}(\zet_0)}
\newcommand{\GZI}{\mathbf{G}^{-1}(\zet_0)}
\newcommand{\HH}{\mathbf{H}}
\renewcommand{\vec}{\text{vec}}
\newcommand{\DDDelta}{\boldsymbol{\Delta}}
\newcommand{\TT}{\mathbf{T}}
\newtheorem{thm}{Theorem}[section]
\newcommand{\OOmega}{\boldsymbol{\Omega}}

\newcommand{\XXi}{\boldsymbol{\Xi}}

\newtheorem{lemma}{Lemma}
\newtheorem{algorithm}{Algorithm}
\newtheorem{assumption}{Assumption}

\newcommand{\mf}{\mathbf}
\newcommand{\bs}{\boldsymbol}
\newcommand{\bet}{\boldsymbol{\beta}}
\newcommand{\thet}{\boldsymbol{\theta}}
\newcommand{\zet}{\boldsymbol{\zeta}}
\newcommand{\gam}{\boldsymbol{\gamma}}

\newcommand{\sigmaz}{\sigma_0}

\newcommand{\thetz}{\boldsymbol{\theta}_0}
\newcommand{\kap}{\boldsymbol{\kappa}}
\newcommand{\Ome}{\boldsymbol{\Omega}}
\newcommand{\elwp}{e^{\lambda \mathbf{W}'}}
\newcommand{\plim}{\mathrm{plim}}
\newcommand{\SSigma}{\boldsymbol{\Sigma}}
\newcommand{\vvec}{\mathrm{vec}}
\DeclareMathOperator{\argsolve}{argsolve}

\newcommand{\W}{\mathbf{W}}

\newcommand{\X}{\mathbf{X}}
\newcommand{\M}{\mathbf{M}}
\newcommand{\R}{\mathbf{R}}
\newcommand{\I}{\mathbf{I}}
\newcommand{\U}{\mathbf{U}}
\newcommand{\A}{\mathbf{A}}
\newcommand{\Ss}{\mathbf{S}}
\newcommand{\V}{\mathbf{V}}

\newcommand{\K}{\mathbf{K}}

\newcommand{\B}{\mathbf{B}}
\newcommand{\Y}{\mathbf{Y}}

\newcommand{\Z}{\mathbf{Z}}
\newcommand{\w}{\mathbf{w}}
\newcommand{\WW}{\mathbb{W}}
\newcommand{\oomega}{\boldsymbol{\omega}}
\newcommand{\Q}{\mathbf{Q}(\rho)}
\newcommand{\PP}{\mathbf{P}(\rho)}
\newcommand{\RR}{\mathbf{R}}

\numberwithin{equation}{section}

\DeclareMathOperator*{\argmax}{arg\,max}
\DeclareMathOperator*{\argmin}{arg\,min}
\DeclareMathOperator*{\E}{E}
\DeclareMathOperator*{\Diag}{Diag}

\begin{document}

\def\spacingset#1{\renewcommand{\baselinestretch}%
	{#1}\small\normalsize} \spacingset{1}


\title{\bf A Review of Cross-Sectional Matrix Exponential Spatial Models}
\author{ Ye Yang\footnote{School of Accounting, Capital University of Economics and Business, Beijing, China, Email:
yeyang557@hotmail.com.} \and Osman Do\u{g}an\footnote{Department of Economics, Istanbul Technical University, Istanbul, Turkey, Email: osmandogan@itu.edu.tr.} \and Süleyman Ta\c{s}p\i nar\footnote{ Department of Economics, Queens College, The City University of New York, New York, U.S.A., Email: staspinar@qc.cuny.edu.} \and Fei Jin\footnote{School of Economics, Fudan University, Shanghai, China, Email: jinfei@fudan.edu.cn.}}
\maketitle

\begin{abstract}
The matrix exponential spatial models exhibit similarities to the conventional spatial autoregressive model in spatial econometrics but offer analytical, computational, and interpretive advantages. This paper provides a comprehensive review of the literature on the estimation, inference, and model selection approaches for the cross-sectional matrix exponential spatial models. We discuss summary measures for the marginal effects of regressors and detail the matrix-vector product method for efficient estimation. Our aim is not only to summarize the main findings from the spatial econometric literature but also to make them more accessible to applied researchers. Additionally, we contribute to the literature by introducing some new results. We propose an M-estimation approach for models with heteroskedastic error terms and demonstrate that the resulting M-estimator is consistent and has an asymptotic normal distribution. We also consider some new results for model selection exercises. In a Monte Carlo study, we examine the finite sample properties of various estimators from the literature alongside the M-estimator.
\end{abstract}

\noindent%

\spacingset{1.45} 



\newpage
\section{Introduction and motivation}

Spatial econometric models deal with estimation and inference problems that arise from (weak) cross-sectional dependence or correlation in data marked with location stamps. The spatial autoregressive (SAR) specification has been a widely used approach for modeling spatial dependence since its inception in \cite{Whittle:1954} and \cite{Cliff:1969,Cliff:1973}. The matrix exponential spatial specification (MESS) was introduced by \cite{Lesage:2007} as an alternative to the SAR specification, primarily due to its computationally appealing properties in likelihood-based estimation schemes. Despite various estimation and inference methods proposed in the econometrics literature for models using either specification, the empirical literature is predominantly populated with papers utilizing the SAR specification.

This highly skewed preference towards the SAR specification by applied researchers is unfortunate in the sense that the MESS attains some attractive properties. First, we must emphasize that the MESS and the SAR imply different rates of decay for cross-sectional dependence. While it is a geometric rate in the case of SAR, the MESS implies an exponential rate of decay for spatial correlation. Consequently, they also imply different reduced forms for the cross-sectional models. Second, contrary to the SAR specification, the MESS does not require any restrictions on the parameter space of the spatial autoregressive parameters as the reduced form of the MESS always exists. In particular, the MESS always yields a positive definite covariance matrix for the outcome variable. Third, in the likelihood based estimation, the SAR specification results in Jacobian determinant terms that can be difficult to compute when the number of cross-sections is large. The MESS likelihood, on the other hand, does not involve such Jacobian terms. 

In this paper, our aim is to provide a complete comprehensive review of the econometric literature on the estimation, inference and model selection methods for the cross-sectional MESS-type models.\footnote{We focus on cross-sectional MESS models as there are relatively few papers on panel data MESS models in the literature. See, e.g., \citet{Chih:2018}, \cite{Zhang2019} and \cite{yangunbalanced}.} More specifically, we aim to present the existing results from the literature in a more accessible way so that they can be  utilized easily in empirical applications by applied researchers. Furthermore, we extend the existing literature in some important aspects. Firstly, we propose a new estimation and inference methodology for the cross-sectional MESS-type models with an unknown form of heteroskedasticity. Second, for the model selection problems involving cross-sectional MESS-type models, we consider a new method for computing the marginal likelihoods of the competing models in a Bayesian framework. In an Monte Carlo study, we also assess the finite sample properties of some existing estimators from the literature along with our proposed method. The simulation results show that the suggested M-estimator performs satisfactorily in finite samples.

Various estimation methods for the MESS models have been considered in the literature \citep{Lesage:2007, Jin:2015, yangfast, yangunbalanced}. \citet{Lesage:2007} consider both the maximum likelihood and Bayesian estimation approaches. \citet{Jin:2015} formally investigate the large sample properties of the quasi maximum likelihood estimator (QMLE) and the generalized method of moments estimator (GMME). Although both estimators have the standard large sample properties,  the GMME can be more efficient than the QMLE when the innovations are non-normal or heteroskedastic.   \citet{Jin:2015} showed that, unlike the SAR-type models, in the presence of an unknown form of heteroskedasticity, the QMLE of the MESS model can remain consistent if the spatial weights matrices used in the model are commutative. In this paper, we extend on their results by introducing an M-estimation methodology that is robust to heteroskedasticity when the spatial weights matrices do not commute. We also formally establish the large sample properties of the resulting M-estimator. 

We provide Bayesian estimation algorithms for the MESS models in the case of both homoskedastic and heteroskedastic error terms \citep{Dogan:2023, yangfast, Lesage:2007}. In the case of heteroskedasticity, we assume that the error terms follow a scale mixture of normal distributions, where the latent scale variables generate distributions with different variances. The latent variable representation facilitates the estimation through the data augmentation techniques. In both homoskedastic and heteroskedastic models, the conditional posterior distributions of parameters take known forms, except for those of the spatial parameters. The posterior draws for the spatial parameters can be generated either by using the random-walk or the independence-chain Metropolis-Hastings algorithms \citep{Lesage:2009,yangfast, LeSage:2007b, Lesage:1997}. We also consider the estimation of the MESS with endogenous and Durbin's regressors. \citet{Jin.Lee2018} show that the popular nonlinear two stage least squares (N2SLS) estimator, although consistent, may suffer from slow rates of convergence, and may attain nonstandard asymptotic distributions when the true value of a subset of model parameters is zero. We highlight how an adaptive group lasso estimator can provide a solution to these problems.

One important issue for the estimation of MESS-type model relates to the computation of the matrix exponential terms. Although there are various methods suggested in the literature, there is no single method that outperforms the rest in all cases \citep{Moler:2003}. As such, we visit the computation of the matrix exponential terms and exhibit how the matrix-vector product approach originally suggested by \cite{Lesage:2007} can be utilized for quick computation of these terms. A further issue for the MESS models relates to the interpretation of the coefficient estimates for the explanatory variables. In spatial models, the interpretation of the coefficient estimates for the explanatory variables become more complicated due to the cross-sectional interactions. To this end, we review the existing results on the estimation and inference results for the impact measures for the MESS-type models \citep{Jin.Lee2018, Lesage:2009, Arbia:2020}. 

When modeling spatial dependence, researchers may encounter specification problems related to choosing a spatial weights matrix from a pool of candidates, or choosing between nested or non-nested alternative model specifications. Often, modeling is done in an ad hoc manner and there is no guidance from an underlying structural model to address these issues. To this end, we provide a complete review of the literature on testing based, criterion based and marginal likelihood based approaches for model selection problems involving MESS-type models \citep{HAN2013250, LIU2019434, Dogan:2023, yangms, Lesage:2009}. In this regard, we also visit the Bayesian approaches and consider the modified harmonic mean method of \citet{Dey:1994} for the computation of the marginal likelihoods of competing models.

The rest of this paper is organized as follows. Section \ref{sec:model} reviews several cross-sectional MESS-type models. This section also introduces the main  properties of a matrix exponential term.  Sections \ref{sec:mle}--\ref{sec:endogenous} discuss various estimation and inference techniques for the MESS-type models. Section \ref{sec:computation} details the matrix-vector product approach for the efficient computation of matrix exponential terms. Section \ref{sec:impact} presents the impact measures for the MESS-type models and illustrates the inference methods. Section \ref{sec:selection} considers various kinds of model selection approaches involving the MESS-type models. Section \ref{sec:mc} presents results of an Monte Carlo study, focusing on the M-estimation of the MESS-type models. Section \ref{sec:conclusion} ends our review with some concluding remarks for future research. Some technical results are relegated to an appendix.

\section{Model Specification}\label{sec:model}
We consider the following first order matrix exponential spatial model (for short MESS$(1,1)$)
\begin{align}\label{2.1}
e^{\lambda_0 \W} \Y=\X\bs{\beta}_0+\U, \quad e^{\rho_0 \M}\U=\V, 
\end{align}
where $\Y=\left(y_{1},\hdots,y_{n}\right)^{'}$ is the $n\times1$ vector of observations on the dependent variable,  $\X$ is the $n\times k$ matrix of non-stochastic exogenous variables with the associated parameter vector $\bs{\beta}_0$,  $\U=\left(u_{1},\hdots,u_{n}\right)^{'}$ is the $n\times1$ vector of regression error terms, and $\V=\left(v_{1},\hdots,v_{n}\right)^{'}$ is the $n\times1$ vector of idiosyncratic error terms. The matrix exponential term $e^{\lambda_0 \W}$  is defined by $e^{\lambda_0 \W}=\sum_{i=0}^\infty\frac{\lambda^{i}_0 \W^{i}}{i!}$, where $\W$ is an $n\times n$ spatial weights matrix with zero diagonal elements and $\lambda_0$ is a scalar spatial parameter. The matrix exponential $e^{\rho_0 \M}$ is defined in a similar way, where $\M$ is another $n\times n$ spatial weights matrix and $\rho_0$ is a scalar spatial parameter.

The MESS(1,1) in \eqref{2.1} can be considered as the matrix exponential counterpart of the spatial autoregressive model with spatial autoregressive disturbances (SARAR(1,1)), 
\begin{align}\label{eq2.2}
(\mf{I}_n - \alpha_0 \W) \Y=\X\bs{\beta}_0+\U, \quad (\mf{I}_n - \tau_0 \M)\U=\V,
\end{align}
where $\alpha_0$ and $\tau_0$ are scalar spatial autoregressive parameters. The MESS(1,1) specification is obtained from \eqref{eq2.2}  by replacing $(\mf{I}_n - \alpha_0 \W)$ and $(\mf{I}_n - \tau_0 \M)$ with $e^{\lambda_0 \W}$ and $e^{\rho_0 \M}$, respectively.  The matrix exponential terms satisfy the following properties \citep{Lesage:2007}: 
\begin{enumerate}
\item $e^{c \mf{A}}$ is non-singular, where $\mf{A}$ is an $n\times n$ matrix and $c$ is a scalar constant,
\item $(e^{c \mf{A}})^{-1}=e^{-c \mf{A}}$, 
\item $\left|e^{c \mf{A}}\right|=e^{c\, \tr(\mf{A})}$, where $|\cdot|$ is the determinant operator and $\tr(\cdot)$ is the trace operator,
\item $e^{\mf{A}}e^{\mf{B}}=e^{\mf{A}+\mf{B}}$, where $\mf{A}$ and $\mf{B}$ are two $n\times n$ matrices satisfying the commutative property $\mf{A}\mf{B}=\mf{B}\mf{A}$.
\end{enumerate}
Because of these properties, the spatial models formulated with the matrix exponential terms have some advantages over the spatial models formulated with the spatial autoregressive processes. The first and second properties ensure that the reduced form of matrix exponential models always exists and does not require any restrictions for the spatial parameters. In the context of \eqref{2.1}, the  reduced form can be expressed as 
\begin{align}
 \Y=e^{-\lambda_0 \W}\X\bs{\beta}_0+e^{-\lambda_0 \W}e^{-\rho_0 \M}\V.
\end{align}
The third property implies that $\left|e^{\lambda \W}\right|=e^{\lambda \tr(\W)}=1$ because $\W$ has zero diagonal elements. This property ensures that the likelihood function of matrix exponential models is free of any Jacobian terms that need to be computed many times during estimation (see Section~\ref{sec:mle} for the details). On the other hand, the likelihood functions of spatial models specified in terms of spatial autoregressive processes is not free of Jacobian terms. For example,  the likelihood function of the SARAR(1,1) model includes $|\mf{I}_n - \tau \M|$ and $|\mf{I}_n - \alpha \W|$, which need to be computed in each iteration during the estimation process.

The MESS(1,1) specification nests two alternative specifications, namely, the MESS(1,0) and MESS(0,1), which can be obtained by setting $\lambda_0=0$ and $\rho_0=0$, respectively. A spatial Durbin extension can be obtained by including the spatial lags of the explanatory variables as regressors:
\begin{align}
&e^{\lambda_0 \W} \Y=\X\bs{\beta}_0+\W\X\bs{\delta}_0+\U, \quad e^{\rho_0 \M}\U=\V, 
\end{align}
where $\W\X$ denotes the spatial lag of $\X$ and $\bs{\delta}_0$ is the corresponding vector of coefficients. 

In the MESS(1,1) model, spatial interactions in the outcome variable arise only trough $\W$, and in the disturbance terms only through $\M$. In some cases, spatial dependence may arise from different sources, requiring different matrix exponential terms formulated with different spatial weights matrices. Let $\{\W_i\}_{i=1}^p$ and $\{\M_j\}_{j=1}^q$ be two sequences of spatial weights matrices. Then, a high-order version  including the matrix exponential terms formulated with $\{\W_i\}_{i=1}^p$ and $\{\M_j\}_{j=1}^q$ can be specified as 
\begin{align}
e^{\sum_{i=1}^p\lambda_{i0} \W_i} \Y=\X\bs{\beta}_0+\U, \quad e^{\sum_{j=1}^q\rho_{j0} \M_j}\U=\V,
\end{align}
where $\{\lambda_{i0}\}_{i=1}^p$ and $\{\rho_{j0}\}_{j=1}^q$ are sequences of spatial parameters. This model can be called the MESS$(p,q)$ model. 

Finally, we specify the distribution of the elements of $\V$. We can consider both homoskedastic and heteroskedastic error terms as specified in the following assumptions. 
\begin{assumption}\label{a1}
The disturbance terms $v_{i}$'s are independent and identically distributed (i.i.d.) across $i$ with mean zero and variance $\sigma_{0}^2$, and $\E|v_{i}|^{4+\varrho}<\infty$ for some $\varrho>0$.
\end{assumption}
\begin{assumption}\label{a2}
The disturbance terms $v_{i}$'s are independently distributed over $i$ with $\E\left(v_{i}\right)=0$ and $\var\left(v_{i}\right)=\sigma_{i}^2$, and $\E\left|v_{i}\right|^{4+\varrho}<\infty$ for some $\varrho>0$.
\end{assumption}
Both assumptions require that the error terms have more than the fourth moment, which is required by the central limit theorem (CLT) considered by \citet{KP:2001, KP:2010} for the linear and quadratic forms of $\V$ (see Lemma 4 in the Appendix).

\section{Maximum likelihood estimation approach} \label{sec:mle}
\subsection{Estimation under homoskedasticity}
In this section, we will consider the quasi maximum likelihood estimation of the MESS(1,1) model under Assumption~\ref{a1}. Let $\thet=(\gam^{'},\sigma^2)^{'}$, $\gam=(\bet^{'},\zet^{'})^{'}$ and $\zet=(\lambda,\rho)^{'}$. Also let $\thetz=(\gam_0^{'},\sigmaz^2)^{'}$ denote the  true values of the parameters. Then, the quasi log-likelihood function for the MESS(1,1) is given by
\begin{align*}
\ln L(\thet)=-\frac{n}{2}\ln(2\pi\sigma^2)+\ln\left|e^{\lambda \W}\right|+\ln\left|e^{\rho \M}\right|-\frac{1}{2\sigma^2}(\elw \Y-\X\bet)^{'}\ermp\erm(\elw \Y-\X\bet).
\end{align*}
Since $\ln\left|\elw\right|=\ln(e^{\lambda\tr(\W)})=\ln1=0$ and $\ln\left|\erm\right|=\ln(e^{\rho\tr(\M)})=\ln1=0$, the two Jacobian terms disappear in the quasi log-likelihood function. Thus, the quasi log-likelihood function simplifies to
\begin{align}\label{2.3}
\ln L(\thet)=-\frac{n}{2}\ln(2\pi\sigma^2)-\frac{1}{2\sigma^2}(\elw \Y-\X\bet)^{'}\ermp\erm(\elw \Y-\X\bet).
\end{align}
We can concentrate out $\sigma^2$ from the quasi log-likelihood function to obtain the concentrated quasi log-likelihood function only involving $\gam$. From the first order condition with respect to $\sigma^2$, the quasi maximum likelihood estimator of $\sigma^2$ is given by
\begin{align}
&\hat{\sigma}^2(\gam)=\frac{1}{n}(\elw \Y-\X\bet)^{'}\ermp\erm(\elw \Y-\X\bet).\label{2.5}
\end{align}
Substituting \eqref{2.5} into \eqref{2.3}, we obtain the concentrated quasi log-likelihood function as
\begin{align}\label{2.6}
\ln L(\gam)=-\frac{n}{2}\ln(2\pi+1)-\frac{n}{2}\ln\hat{\sigma}^2(\gam).
\end{align}
 Then, the QMLE $\hat{\gam}$ of $\gamma_0$ is defined as
\begin{align*}
\hat{\gam}=\argmax_{\gam}\ln L(\gam),
\end{align*}
which is equivalent to 
\begin{align}\label{qmle3.4}
\hat{\gam}=\argmin_{\gam} Q(\gam),
\end{align}
where $Q(\gam)=(\elw \Y-\X\bet)^{'}\ermp\erm(\elw \Y-\X\bet)$. Substituting $\hat{\gam}$ into \eqref{2.5}, we obtain the QMLE of $\sigma^2$ as $\hat{\sigma}^2=\hat{\sigma}^2(\hat{\gam})$. 

The large sample properties of the QMLE $\hat{\gam}$ can be established under some regularity conditions. For consistency, the necessary conditions are identifiable uniqueness of $\gam_0$ and the uniform stochastic convergence of the quasi maximum likelihood function to its population counterpart \citep[Theorem 3.4]{white:1994}. For asymptotic normality of $\hat{\gam}$, the CLT for linear and quadratic forms can be utilized \citep{KP:2001, KP:2010}. The low level assumptions guaranteeing the large sample properties of $\hat{\gam}$ are (i) the existence of moments of disturbance terms up to the fourth moment, (ii) a manageable degree of spatial correlation, (iii) a compact parameter space for $\zet$, (iv) the non-singularity of certain matrices in large samples, and (v) certain restrictions to guarantee identification of $\gam_0$ in large samples.\footnote{See \citet{Jin:2015} for a complete formal list of these low level assumptions.}  
 
The score functions with respect to the elements of $\gam$ are given by
\begin{align}\label{qmlscore}
\frac{\partial  Q(\gam)}{\partial\gam} & = \left\{\begin{array}{rl}
\bet:&-2\X^{'}e^{\rho_0\M^{'}}\V(\gam), \\
\lambda:& 2\V^{'}(\gam)e^{\rho \M}\W e^{\lambda \W}\Y,\\
\rho:&2\V^{'}(\gam)\M \V(\gam), 
\end{array}\right.
\end{align}
where $\V(\gam) = \erm(\elw \Y-\X\bet)$. Define $\mathbf{B}= \var\left(\frac{1}{\sqrt{n}}\frac{\partial  Q(\gam_0)}{\partial\gam}\right)$ and $\mathbf{A}=\mathrm{E}\left(-\frac{1}{n} \frac{\partial^2 Q\left(\gam_0\right)}{\partial \gam \partial \gam^{\prime}}\right)$. To introduce the closed-forms of $\mf{A}$ and $\mf{B}$, let $\mu_3=\E(v_i^3)$,  $\mu_4=\E(v_i^4)$, $\mathbb{W}=\ermz \W\nermz$, $\mathbf{H}^s=\mathbf{H}+\mathbf{H}^{'}$ for any square matrix $\mathbf{H}$ and $\vvec_D(\mathbf{H})$ be a vector containing the diagonal elements of $\mathbf{H}$. Then, using \eqref{qmlscore} and Lemma 2 in Appendix~\ref{lemmas}, we obtain 
 \begin{align*} 
 \mathbf{A} = -\frac{1}{n}\left(\begin{array}{ccc} 2\left(e^{\rho_0 \M} \X\right)^{'}\left(e^{\rho_0 \M} \X\right) & * & * 
 \\-2\left(\mathbb{W} e^{\rho_0 \M} \X \bet_0\right)^{'}e^{\rho_0 \M} \X  &\mathbf{A}_{\lambda\lambda}& * 
 \\  \mathbf{0}&\sigma_0^2 \tr\left(\mathbb{W}^s \M^s\right)   & \sigma_0^2 \tr\left(\M^s \M^s\right) \end{array}\right),
 \end{align*}
 \begin{align*}
\mathbf{B}=2\sigma_0^2 \mathbf{A}+\frac{1}{n}\left(\begin{array}{ccc} \mathbf{0}& * & * \\ 
-2 \mu_3\left( \vvec_D\left(\mathbb{W}^s\right)\right)^{\prime} e^{\rho_0 \M} \X&  \mathbf{0} & * \\ 
  \mathbf{0}& \mathbf{0} & \mathbf{B}_{\rho\rho}\end{array}\right).
\end{align*}
 where  the elements are defined as $\mathbf{A}_{\lambda\lambda}=\sigma_0^2 \tr\left(\mathbb{W}^s \mathbb{W}^s\right)+2\left(\mathbb{W} e^{\rho_0 \M} \X \bet_0\right)^{'}\left(\mathbb{W} e^{\rho_0 \M} \X \bet_0\right)$ and $\mathbf{B}_{\rho\rho}=\left(\mu_4-3 \sigma_0^4\right) \vvec_D^{\prime}\left(\mathbb{W}^s\right) \vvec_D\left(\mathbb{W}^s\right)+4 \mu_3\left(\mathbb{W} e^{\rho_0 \M} \X \bet_0\right)^{\prime} \vvec_D\left(\mathbb{W}^s\right)$.
The asymptotic distribution of $\hat{\gam}$ can be derived by applying the mean value theorem to $\frac{\partial Q(\hat{\gam})}{\partial \gam}$ around $\gam_0$. By the mean value theorem, we can write $\sqrt{n}(\hat{\gam}-\gam_0)=-\left(\frac{1}{n}\frac{\partial^2 Q(\bar{\gam})}{\partial\gam\partial\gam^{'}}\right)^{-1}\frac{1}{\sqrt{n}}\frac{\partial Q(\gam_0)}{\partial\gam}$, where $\bar{\gam}$ lies between $\hat{\gam}$ and $\gam_0$ elementwise. Then, the desired result follows by showing that $\frac{1}{n}\frac{\partial^2 Q(\bar{\gam})}{\partial\gam\partial\gam^{'}}-\frac{1}{n}\E\left(\frac{\partial^2 Q(\gam_0)}{\partial\gam\partial\gam^{'}}\right)=o_p(1)$ and the asymptotic normality of $\frac{1}{\sqrt{n}}\frac{\partial Q(\gam_0)}{\partial\gam}$ by Lemma 4 in Appendix \ref{lemmas}. Thus, it follows that
\begin{align}\label{eq3.6}
\sqrt{n}(\hat{\gam}-\gam_0)\xrightarrow{d}N\Bigl(\mathbf{0},\lim_{n\rightarrow\infty}\mathbf{A}^{-1}\mathbf{B}\mathbf{A}^{-1}\Bigr).
 \end{align}

Note that there are two cases that yield $\mf{B}=\sigma^2_0\mf{A}$. The first case arises when $\W$ and $\M$ commute. Under the commutative property, we have $\mathbb{W}=\W$ and $\vvec_D(\mathbb{W}^s) = \mathbf{0}$, suggesting that $\mf{B}=\sigma^2_0\mf{A}$. The second case occurs when the disturbance terms are normally distributed. Under the normality, we have $\mu_4=3\sigma^4_0$ and $\mu_3=0$, yielding $\mf{B}=\sigma^2_0\mf{A}$. In either case, the result in \eqref{eq3.6} reduces to $\sqrt{n}(\hat{\gam}-\gam_0)\xrightarrow{d}N(\mathbf{0},\lim_{n\rightarrow\infty} \sigma_0^2\mathbf{A}^{-1})$.  

Finally, for inference, the plug-in estimators of $\mathbf{A}$ and $\mathbf{B}$ can be utilized. To that end,  $\sigma_0^2$ can be consistently estimated by evaluating \eqref{2.5} at $\hat{\gam}$, and $\mu_3$ and $\mu_4$ can be consistently estimated by their sample analogs using the residuals $\V(\hat{\gam})$. Thus, the standard error of $\hat{\gam}$ can be obtained as the square root of the diagonal elements of $\frac{1}{n}\mathbf{A}^{-1}(\hat{\gam})\mathbf{B}(\hat{\gam})\mathbf{A}^{-1}(\hat{\gam})$, where  $\mathbf{A}(\hat{\gam})$ and $\mathbf{B}(\hat{\gam})$ are the plug-in estimators of $\mf{A}$ and $\mf{B}$, respectively.

\subsection{Estimation under heteroskedasticity}  

In this subsection, we consider the quasi maximum likelihood estimation of the MESS(1,1) under the assumption of heteroskedastic disturbance terms. Let $\boldsymbol{\Sigma}$ be the variance covariance matrix of the disturbance terms, i.e., $\boldsymbol{\Sigma}=\Diag(\sigma_1^2,\hdots,\sigma_n^2)$, the diagonal matrix formed by $\sigma_i^2$'s. The score functions of the quasi likelihood function evaluated at $\gam_0$ are given by
\begin{align*}
\frac{\partial  Q(\gam_0)}{\partial\gam} & = \left\{\begin{array}{rl}
\bet:&-2\X^{'}e^{\rho_0\M^{'}}\V,\\
\lambda:& 2\V^{'}e^{\rho_0\M}\W\left(\X \bet_0 + e^{-\rho_0\M}\V\right),\\
\rho:& 2\V^{'}\M \V.
\end{array}\right.
\end{align*}
The expectation of the score functions with respect to   $\bet$ and $\rho$ at  $\gam_0$ are zero by Lemma 2 in Appendix \ref{lemmas}. However, the expectation of the score function with respect to $\lambda$ at  $\gam_0$ is $\tr(\mathbb{W}\boldsymbol{\Sigma})$.  By Lemma 1 in Appendix \ref{lemmas}, the order of this term is $O(n)$ under the assumption that $\W$ and $\M$ are bounded in matrix column sum and row sum norms. Hence, the QMLE $\hat{\gam}$ may not be consistent. However, when $\W$ and $\M$ commute, we have $\mathbb{W}=\W$, yielding $\tr(\W\boldsymbol{\Sigma})=0$. Hence, when $\W$ and $\M$ commute, the QMLE of MESS(1,1) may remain consistent under the assumption of heteroskedastic disturbance terms.
 
The consistency and asymptotic normality of the QMLE $\hat{\gam}$ can be proved similarly to the homoskedastic case. Let $\mathbf{D}=\mathrm{E}\left(-\frac{1}{n} \frac{\partial^2 Q\left(\gam_0\right)}{\partial \gam \partial \gam^{\prime}}\right)$ and $\mathbf{F}= \var\left(\frac{1}{\sqrt{n}}\frac{\partial  Q(\gam_0)}{\partial\gam}\right)$. Using Lemma 2 in Appendix \ref{lemmas}, we obtain
\begin{align*}
&\D=-\frac{2}{n}\left(\begin{array}{ccc}
\left(e^{\rho_0 \M} \X\right)^{'}\left(e^{\rho_0 \M} \X\right) & * & * \\
 -\left( \W e^{\rho_0 \M} \X \bet_0\right)^{'}e^{\rho_0 \M} \X& \mf{D}_{\lambda\lambda}& * \\
  \mathbf{0} & \tr\left(\M^s \W \boldsymbol{\Sigma}\right)& \tr\left(\M^s \M \boldsymbol{\Sigma}\right) 
\end{array}\right),\\
&\mathbf{F}=\frac{2}{n}\left(\begin{array}{ccc}
2\left(e^{\rho_0 \M} \X\right)^{'} \boldsymbol{\Sigma}\left(e^{\rho_0 \M} \X\right)& * & * \\
-2\left( \boldsymbol{\Sigma} \W e^{\rho_0 \M} \X \bet_0 \right)^{'} e^{\rho_0 \M} \X&\mf{F}_{\lambda\lambda} & * \\
 \mathbf{0}&\tr\left(\boldsymbol{\Sigma} \M^s \boldsymbol{\Sigma} \W^s\right) &  \tr\left(\boldsymbol{\Sigma} \M^s \boldsymbol{\Sigma} \M^s\right) 
\end{array}\right),
\end{align*}
where 
\begin{align*}
&\mf{D}_{\lambda\lambda}=\tr\left(\W^s \W \boldsymbol{\Sigma}\right)+\left(\W e^{\rho_0 \M} \X \bet_0\right)^{'}\left(\W e^{\rho_0 \M} \X \bet_0\right),\\
&\mf{F}_{\lambda\lambda}=\tr\left(\boldsymbol{\Sigma} \W^s \boldsymbol{\Sigma} \W^s\right)+2\left(\W e^{\rho_0 \M} \X \bet_0\right)^{'} \boldsymbol{\Sigma}\left(\W e^{\rho_0 \M} \X \bet_0\right).
\end{align*}
Then, it can be shown that 
\begin{align}
\sqrt{n}(\hat{\gam}-\gam_0)\xrightarrow{d}N\Bigl(\mathbf{0},\lim_{n\rightarrow\infty} \mathbf{D}^{-1}\mathbf{F}\mathbf{D}^{-1}\Bigr).
\end{align}
For inference, the standard error of $\hat{\gam}$ can be obtained as the square root of the diagonal elements of $\frac{1}{n}\mathbf{D}^{-1}(\hat{\gam})\mathbf{F}(\hat{\gam})\mathbf{D}^{-1}(\hat{\gam})$, where $\mathbf{D}(\hat{\gam})$ and $\mathbf{F}(\hat{\gam})$ are the plug-in estimators of $\mf{D}$ and $\mf{F}$, respectively. Also, note that $\mathbf{D}$ and $\mathbf{F}$ involve the unknown diagonal matrix $\boldsymbol{\Sigma}$. As in \citet{White:1980}, the terms involving $\boldsymbol{\Sigma}$ can be consistently estimated by replacing $\boldsymbol{\Sigma}$ with $\hat{\bs{\Sigma}}=\Diag\left(v_1^2(\hat{\gam}),\hdots,v_n^2(\hat{\gam})\right)$, where $v_i(\hat{\gam})$ is the $i$th element of $\mf{V}(\hat{\gam})$.

\section{ M-estimation approach}\label{sec:m-est}

Under heteroskedasticity, when $\W$ and $\M$ do not commute, we can use the M-estimation method to formulate a consistent estimator of $\gam$ based on the adjusted score functions. Denote $\V(\bet,\bs{\zeta})=\erm(\elw \Y-\X\bet)$ and $\mf{V}=\V(\bet_0,\bs{\zeta}_0)$. Then, the score functions based on \eqref{2.3} can be determined as\footnote{Under heteroskedasticity, we ignore $\sigma^2$ and aim to construct the adjusted score functions such that $\E\left(S(\bet_0,\bs{\zeta}_0)\right)=0$. } 
\begin{align}\label{4.1}
S(\bet,\bs{\zeta})=
\left\{\begin{array}{rl}
\bet:&\X^{'}\ermp\V(\bet,\zeta),\\
\lambda: &-\Y^{'}\elwp\W^{'}\ermp\V(\bet,\zeta),\\
\rho: &-\V^{'}(\bet,\zeta)\M\V(\bet,\zeta).
\end{array}\right.
\end{align}
The essential reason why the QMLE is not consistent is $\plim_{n\rightarrow\infty}\frac{1}{n}S(\gam_0)\ne0$.  In the case of the score functions with respect to $\bet$ and $\rho$, we have $\E(\X^{'}\ermzp\V)=0$, and $\E(\V^{'}\M\V)=\tr(\boldsymbol{\Sigma}\M)=0$ because $\bs{\Sigma}$ is a diagonal matrix and $\M$ has zero diagonal elements.  In the case of the score function with respect to $\lambda$, we have 
\begin{align}\label{4.2}
\E(\Y^{'}\elwzp\W^{'}\ermzp\V)&=\E(\Y^{'}\elwzp\ermzp\ermzpn\W^{'}\ermzp\V)=\E(\Y^{'}\elwzp\ermzp\mathbb{W}^{'}\V)\nonumber\\
&=\E(\V^{'}\mathbb{W}\V)=\tr(\SSigma\WW).
\end{align}
If $\W$ and $\M$ commute, i.e., $\W\M=\M\W$, then we have $\WW=\W$, which yields $\tr(\SSigma\WW)=\tr(\SSigma\W)=0$. Thus, when the commutative property holds, we have $\plim_{n\rightarrow\infty}\frac{1}{n}S(\gam_0)=0$, suggesting that the QMLE can be consistent under heteroskedasticity.  However, if $\W\M\ne\M\W$, then we have $\tr(\SSigma\WW)=O(n)$ and $\plim_{n\rightarrow\infty}\frac{1}{n}S(\gam_0)\ne0$ in general, indicating that the QMLE may not be consistent under heteroskedasticity. We will adjust the score function with respect to $\lambda$ so that $\plim_{n\rightarrow\infty}\frac{1}{n}S(\gam_0)=0$ holds in all cases.  

To adjust the score function with respect to $\lambda$, we use the trace property $\tr(\mathbf{D}\mathbf{A})=\tr(\mathbf{D}\Diag(\mathbf{A}))$, where $\mf{D}$ is an $n\times n$ diagonal matrix and $\mf{A}$ is a conformable matrix. Using this property, we can express $\E(\Y^{'}\elwzp\W^{'}\ermzp\V)$ as
\begin{align}\label{4.3}
\E(\Y^{'}\elwzp\W^{'}\ermzp\V)&=\E(\Y^{'}\elwzp\ermzp\mathbb{W}^{'}\V)=\tr(\SSigma\WW)=\tr(\SSigma\Diag(\WW))\nonumber\\
&=\E(\V^{'}\Diag(\WW)\V)=\E(\Y^{'}\elwzp\ermzp\Diag(\WW)\V).
\end{align}
Then, subtracting the last term from the second term in \eqref{4.3}, we obtain
\begin{align}
&\E(\Y^{'}\elwzp\ermzp\mathbb{W}^{'}\V)-\E(\Y^{'}\elwzp\ermzp\Diag(\WW)\V)=0\nonumber\\
&\implies\E(\Y^{'}\elwzp\ermzp\WW_D\V)=0,
\end{align}
where $\WW_D=\WW-\Diag(\WW)$. Thus, we suggest using the sample counter part of  $\E(\Y^{'}\elwzp\ermzp\WW_D\V)$ as the adjusted score function with respect to $\lambda$. Then, our suggested adjusted score functions take the following form:
\begin{align}\label{4.5}
S^{*}(\gam)=
\left\{\begin{array}{rl}
\bet:&\X^{'}\ermp\V(\bet,\zet),\\
\lambda: &-\Y^{'}\elwp\ermp\WW_D(\rho)\V(\bet,\zet),\\
\rho: &-\V^{'}(\bet,\zeta)\M\V(\bet,\zet),
\end{array}\right.
\end{align}
where $\WW_D(\rho)=\WW(\rho)-\Diag(\WW(\rho))$ and $\WW(\rho)=\erm\W\ermn$. Note that $\E\left(S^{*}(\gam_0)\right)=0$ holds by construction. We first derive the estimator of $\bet_0$ for a given $\zet$ value, which is given by 
\begin{align}\label{4.6}
\hat{\bet}_M(\zet)=(\X^{'}\ermp\erm\X)^{-1}\X^{'}\ermp\erm\elw\Y.
\end{align} 
Then, substituting $\hat{\bet}_M(\zet)$ into the $\lambda$ and $\rho$ elements of \eqref{4.5}, we obtain the concentrated score functions as 
\begin{align}\label{4.7}
S^{c*}(\zet)=
\left\{\begin{array}{rl}
\lambda: &-\Y^{'}\elwp\ermp\WW_D(\rho)\hat{\V}(\zet),\\
\rho: &-\hat{\V}^{'}(\zet)\M\hat{\V}(\zet),
\end{array}\right.
\end{align}
where $\hat{\V}(\zet)=\V(\hat{\bet}_M(\zet),\zet)$. Then, the M-estimator (ME) of $\zet_0$ is defined by
\begin{align}\label{me}
\hat{\zet}_M=\argsolve\{S^{c*}(\zet)=0\}.
\end{align}
Substituting $\hat{\zet}_M$ into \eqref{4.6}, we get the M-estimator for $\bet$ as $\hat{\bet}_M=\hat{\bet}_M(\hat{\zet}_M)$. To prove the consistency of $\hat{\gam}_M=(\hat{\bet}_M^{'},\hat{\zet}_M^{'})^{'}$, we only need to prove the consistency of $\hat{\zet}_M$ since $\hat{\bet}_M=\hat{\bet}_M(\hat{\zet}_M)$. To that end, we let $\bar{S}^{*}(\bet,\zet)=\E(S^{*}(\bet,\zet))$ be the population counterpart of the adjusted score functions in \eqref{4.5}. Given $\zet$, we can write $\bar{\bet}_M(\zet)=(\X^{'}\ermp\erm\X)^{-1}\X^{'}\ermp\erm\elw\E(\Y)$, which can be substituted into the $\lambda$ and $\rho$ elements of $\bar{S}^{*}(\bet,\zet)$ to obtain
\begin{align}\label{4.9}
\bar{S}^{c*}(\zet)=
\left\{\begin{array}{rl}
\lambda: &-\E\left(\Y^{'}\elwp\ermp\WW_D(\rho)\bar{\V}(\zet)\right),\\
\rho: &-\E\left(\bar{\V}^{'}(\zet)\M\bar{\V}(\zet)\right),
\end{array}\right.
\end{align}
where $\bar{\V}(\zet)=\V(\bar{\bet}_M(\zet),\zet)$. To investigate the asymptotic properties of $\hat{\zet}_M$, we maintain the following assumptions.
\begin{assumption}\label{a3}
The spatial weights matrices $\W$ and $\M$ are uniformly bounded in both row sum and column sum matrix norms.
\end{assumption}
\begin{assumption}\label{a4}
There exists a constant $c>0$ such that $\left|\lambda\right|\leq c$ and $\left|\rho\right|\leq c$, and the true parameter vector $\zet_0$ lies in the interior of $\bs{\Delta}=[-c,c]\times[-c,c]$. 
\end{assumption}
\begin{assumption}\label{a5}
$\X$ is exogenous, with uniformly bounded elements, and has full column rank. Also, $\lim_{n\rightarrow\infty}\frac{1}{n}\X^{'}\ermp\erm \X$ exists and is nonsingular, uniformly in $\rho\in [-c,\,c]$. 
\end{assumption}
\begin{assumption}\label{a6}
$\inf_{\zet:\,d(\zet,\zet_0)\ge \vartheta}\left\Vert\bar{S}^{* c}(\zet)\right\Vert>0$ for every $\vartheta>0$, where $d(\zet,\zet_0)$ is a measure of distance between $\zet$ and $\zet_0$.
\end{assumption}
Assumption~\ref{a3} provides the essential properties of the spatial weights matrices. It ensures that the spatial correlation is limited to a manageable degree \citep{KP:2001,KP:2010}. Assumption~\ref{a4} requires that the parameter space of the parameters in the matrix exponential terms is compact. Assumption~\ref{a3} and Assumption~\ref{a4} imply that the matrix exponential terms are uniformly bounded in both row sum and column sum matrix norms. This can be seen from $\left\Vert \elw\right\Vert=\left\Vert \sum_{i=0}^{\infty}\lambda^i\W^i/i!\right\Vert\leq\sum_{i=0}^{\infty}|\lambda|^{i}\Vert \W\Vert^{i}/i!=e^{|\lambda|\Vert \W\Vert}$, which is bounded if $|\lambda|$ and $\Vert \W\Vert$ are bounded, where $\Vert\cdot\Vert$ is either the row sum or the column sum matrix norm. Assumption~\ref{a5} provides some regularity conditions and corresponds to Assumption 4 of  \citet{Jin:2015}. Assumption~\ref{a6} is a high-level assumption and ensures the identification of  $\zet_0$. In Appendix C, we provide two low-level conditions that are sufficient for Assumption~\ref{a6}.

The uniform convergence $\sup_{\zet\in\DDDelta}\frac{1}{n}\left\Vert S^{* c}(\zet)-\bar{S}^{* c}(\zet)\right\Vert\xrightarrow{\enskip p\enskip}0$ and Assumption~\ref{a6} ensure the consistency of $\hat{\zet}_M$.
 \begin{thm}\label{thm1}
Under Assumptions \ref{a2}--\ref{a6} , we have $\hat{\gam}_M\xrightarrow{p}\gam_0$.
\end{thm}
\begin{proof}
See Section \ref{appB1} in the Appendix.
\end{proof}
To derive the asymptotic distribution of $\hat{\gam}_M$, we apply the mean value theorem to $S^{*}(\hat{\gam}_M)=0$ at $\gam_0$, to obtain $\sqrt{n}(\hat{\gam}_M-\gam_0)=-\left(\frac{1}{n}\frac{\partial S^{*}(\overline{\gam})}{\partial \gam^{'}}\right)^{-1}\frac{1}{\sqrt{n}}S^{*}(\gam_0)$, where $\overline{\gam}$ lies between $\gam_0$ and $\hat{\gam}_M$ elementwise. By substituting the reduced form $\Y=e^{-\lambda_0 \W}\left(\X\bs{\beta}_0+e^{-\rho_0 \M}\V\right)$ into $S^{*}(\gam_0)$, we obtain a linear-quadratic form in $\V$:
\begin{align}\label{4.10}
S^{*}(\gam_0)=
\left\{\begin{array}{rl}
\bet:&\X^{'}\ermzp\V,\\
\lambda: &-\bet_0^{'}\X^{'}\ermzp\WW_D\V-\V^{'}\WW_D\V,\\
\rho: &-\V^{'}\M\V,
\end{array}\right.
\end{align}
where $\WW_D=\WW_D(\rho_0)$.  Thus, the CLT for the linear-quadratic forms of $\mf{V}$ in Lemma 4 of the Appendix can be used to establish the asymptotic normality of $\frac{1}{\sqrt{n}}S^{*}(\gam_0)$. Also, our assumptions ensure that $\frac{1}{n}\frac{\partial S^{*}(\overline{\gam})}{\partial \gam^{'}}-\frac{1}{n}\E\left(\frac{\partial S^{*}(\gam_0)}{\partial \gam^{'}}\right)=o_p(1)$. Using these results, we determine the asymptotic distribution of $\hat{\gam}_M$ in Theorem~\ref{thm2}.
\begin{thm}\label{thm2}
Under Assumptions \ref{a2}--\ref{a6}, we have
\begin{equation}
\sqrt{n}(\hat{\gam}_M-\gam_0)\xrightarrow{\enskip d \enskip}N\left(0,\lim_{n\rightarrow\infty}\boldsymbol{\Psi}^{-1}(\gam_0)\boldsymbol{\Omega}(\gam_0)\boldsymbol{\Psi}^{-1'}(\gam_0)\right),
\end{equation}
where $\boldsymbol{\Psi}(\gam_0)=-\frac{1}{n}\operatorname{E}\left(\frac{\partial S^*(\gam_0)}{\partial \gam^{'}}\right)$ and $\boldsymbol{\Omega}(\gam_0)=\var\left(\frac{1}{\sqrt{n}}S^{*}(\gam_0)\right)$ are assumed to exist and $\boldsymbol{\Psi}(\gam_0)$ is assumed to be positive definite for sufficiently large $n$.
\end{thm}
\begin{proof}
See Section \ref{appB2}  in the Appendix.
\end{proof}
To conduct inference, we need consistent estimators of $\boldsymbol{\Psi}(\gam_0)$ and $\OOmega(\gam_0)$. For $\boldsymbol{\Psi}(\gam_0)$, we can use its observed counterpart given by $\boldsymbol{\Psi}(\hat{\gam}_M)=-\frac{1}{n}\frac{\partial S^*(\gam)}{\partial \gam^{'}}|_{\gam=\hat{\gam}_M}$. The elements of $\boldsymbol{\Psi}(\gam)$ are given by
\begin{align*}
&\boldsymbol{\Psi}_{\bet\bet}(\gam)=\frac{1}{n}\X^{'}\ermp\erm\X,   \quad \boldsymbol{\Psi}_{\bet\lambda}(\gam)=-\frac{1}{n}\X^{'}\ermp \Y(\zet),\\
&\boldsymbol{\Psi}_{\bet\rho}(\gam)=-\frac{1}{n}\X^{'}\ermp\M^{s}\V(\bet,\zet),\,\boldsymbol{\Psi}_{\lambda\bet}(\gam)=-\frac{1}{n}\Y^{'}\elwp\ermp\WW_D(\rho)\erm\X,\\
& \boldsymbol{\Psi}_{\lambda\lambda}(\gam)=\frac{1}{n}\Y^{'}(\zet)\WW_D(\rho)\V(\bet,\zet)+\frac{1}{n}\Y^{'}\elwp\ermp\WW_D(\rho)\Y(\zet),\\
& \boldsymbol{\Psi}_{\lambda\rho}(\gam)=\frac{1}{n}\Y^{'}\elwp\ermp\M^{'}\WW_D(\rho)\V(\bet,\zet)+\frac{1}{n}\Y^{'}\elwp\ermp\dot{\WW}_D(\rho)\V(\bet,\zet)\\
& \quad\quad\quad\quad+\frac{1}{n}\Y^{'}\elwp\ermp\WW_D(\rho)\M\V(\bet,\zet),\quad  \\
&\boldsymbol{\Psi}_{\rho\bet}(\gam)=\boldsymbol{\Psi}_{\bet\rho}(\gam), \quad  \boldsymbol{\Psi}_{\rho\lambda}(\gam)=\frac{1}{n}\Y^{'}(\zet)\M^{s}\V(\bet,\zet),\\
& \boldsymbol{\Psi}_{\rho\rho}(\gam)=\frac{1}{n}\V^{'}(\bet,\zet)\M^{s}\M\V(\bet,\zet),
\end{align*} 
where $\dot{\WW}_D(\rho)=\frac{\partial \WW_D(\rho)}{\partial \rho}=\M\WW_D(\rho)-\WW_D(\rho)\M-\Diag\left(\M\WW_D(\rho)-\WW_D(\rho)\M\right)$ and  $\Y(\zet)=\erm\W\elw\Y$. In the proof of Theorem \ref{thm2}, we show that $\boldsymbol{\Psi}(\hat{\gam}_M)$ is a consistent estimator of $\boldsymbol{\Psi}(\gam_0)$. Using Lemma 2 in the Appendix, we determined the closed form of $\OOmega(\gam_0)$ as
\begin{align*}
\OOmega(\gam_0)=
\left(\begin{array}{ccc} 
\X^{'}\ermzp\SSigma\ermz\X & -\X^{'}\ermzp\SSigma\WW_D^{'}\ermz\X\bet_0  & \mathbf{0}_{k\times1}\\
 * & \bs{\Omega}_{22}& \tr(\SSigma\WW_D\SSigma \M^s)\\
 *&*&\tr(\SSigma\M\SSigma \M^s)
\end{array}\right),
\end{align*} 
where $\bs{\Omega}_{22}=\bet_0^{'}\X^{'}\ermzp\WW_D\SSigma\WW_D^{'}\ermz\X\bet_0+\tr(\SSigma\WW_D\SSigma\WW_D^s)$.  Let $\OOmega(\hat{\gam}_M)$ be the plug-in estimator of $\OOmega(\gam_0)$, where we replace $\boldsymbol{\Sigma}$ with $\hat{\bs{\Sigma}}=\Diag\left(v_1^2(\hat{\gam}),\hdots,v_n^2(\hat{\gam})\right)$ and $v_i(\hat{\gam}_M)$ is the $i$th element of $\mf{V}(\hat{\gam}_M)$. 
\begin{thm}\label{thm3}
Under Assumptions \ref{a2}--\ref{a6}, we have $\OOmega(\hat{\gam}_M)=\OOmega(\gam_0)+o_p(1)$.
\end{thm}
\begin{proof}
See Section \ref{appB3}  in the Appendix.
\end{proof}
Thus, the standard error of $\hat{\gam}_M$ can be obtained as the square root of the diagonal elements of $\frac{1}{n}\bs{\Psi}^{-1}(\hat{\gam}_M)\bs{\Omega}(\hat{\gam}_M)\bs{\Psi}^{-1'}(\hat{\gam}_M)$.

\section{GMM estimation approach} \label{sec:gmm}
\subsection{Estimation under homoskedasticity}  

In this section, we consider the GMM estimation of the MESS(1,1) model under Assumption~\ref{a1}. Recall again from the definition of MESS(1,1) that $\V(\gam)=\erm(\elw \Y-\X\bet)$, where $\gam=(\zet^{'},\bet^{'})^{'}$ and $\zet=(\lambda,\rho)^{'}$. We consider the following vector of moment functions consisting of $k_p$ quadratic moments and $k_f$ linear moments: 
\begin{align*}
g(\gam)=\frac{1}{n}\left(\V^{'}(\gam)\p_1\V(\gam),\hdots,\V^{'}(\gam)\p_{k_p}\V(\gam),\V^{'}(\gam)\F\right)^{'},
\end{align*}
where $\p_{m}$'s are $n\times n$ matrices constants with $\tr(\mf{P}_m)=0$ for $m=1,\hdots,k_p$, and $\F$ is the $n\times k_f$ matrix of instrumental variables (IV). Given an arbitrary symmetric weighting matrix $\pphi$ with rank greater than or equal to $k+2$,  the GMM objective function is given by $g^{'}(\gam)\pphi g(\gam)$. Then, an initial GMME (IGMME) can be obtained by 
\begin{align}\label{igmme5.1}
\hat{\gam}=\argmin_{\gam} g^{'}(\gam)\pphi g(\gam). 
\end{align}
 Define $\mathbf{G}=\E\left(\frac{\partial g(\gam_0)}{\partial\gam^{'}}\right)$ and $\HH=n \mathrm{E}\left(g\left(\gam_0\right) g^{\prime}\left(\gam_0\right)\right)$. Let $\vvec(\mathbf{A})$ denote the column vector formed by stacking the columns of matrix $\mathbf{A}$ and recall that $\vvec_D(\mathbf{A})$ denotes the column vector formed by the diagonal elements of the matrix $\mathbf{A}$. Then, by Lemma 2 in Appendix \ref{lemmas}, we obtain
\begin{align*}
\HH&=\frac{1}{n}\left(\begin{array}{cc}\frac{\sigma_0^4}{2} \oomega^{'} \oomega+\frac{1}{4}\left(\mu_4-3 \sigma_0^4\right) \oomega_{d}^{'} \oomega_{d} & \frac{1}{2} \mu_3 \oomega_{d}^{'} \F \\ \frac{1}{2} \mu_3 \F^{'} \omega_{d} & \sigma_0^2 \F^{\prime} \F\end{array}\right),
\end{align*} 
and  
\begin{align*}
\mathbf{G}&
= \frac{1}{n}\left(\begin{array}{ccc}\mathbf{0}&\frac{\sigma_0^2}{2} \oomega^{\prime} \operatorname{vec}\left(\mathbb{W}^s\right) & \frac{\sigma_0^2}{2} \oomega^{'} \operatorname{vec}\left(\M^s\right) \\ 
-\F^{'} e^{\tau_0 \M} \X&\F^{'} \mathbb{W} e^{\tau_0 \M} \X \bet_0 & \mathbf{0}  \end{array}\right),
 \end{align*}
where $\oomega=(\vvec(\p_1^s),\hdots,\vvec(\p_{k_p}^s))$ and $\oomega_d=(\vvec_D(\p_1^s),\hdots,\vvec_D(\p_{k_p}^s))$. 

The large sample properties of the IGMME $\hat{\gam}$ can be established under some regularity conditions. For consistency, the necessary conditions are identification of $\gam_0$ from the population moments and the uniform stochastic convergence of the generalized method of moments objective function to its population counterpart. For the asymptotic normality of $\hat{\gam}$, the central limit theorem for linear and quadratic forms can be utilized.\footnote{ The low level assumptions guaranteeing the large sample properties are provided in \citet{Jin:2015}.} The  asymptotic distribution of $\hat{\gam}$ can be derived  by applying the mean value theorem to $\frac{\partial g^{'}(\hat{\gam})}{\partial \gam}\pphi g(\hat{\gam})=0$ at $\gam_0$ to get $\sqrt{n}(\hat{\gam}-\gam_0)=-\left(\frac{\partial g^{'}(\hat{\gam})}{\partial\gam}\pphi \frac{\partial g(\bar{\gam})}{\partial\gam^{'}}\right)^{-1}\frac{\partial g^{'}(\hat{\gam})}{\partial\gam}\pphi\sqrt{n} g(\gam_0)$, where $\bar{\gam}$ lies between $\hat{\gam}$ and $\gam_0$ elementwise. Then, the asymptotic distribution of $\sqrt{n}(\hat{\gam}-\gam_0)$ follows by applying the CLT in Lemma 4 in the Appendix to $\sqrt{n} g(\gam_0)$ and
showing that $\frac{\partial g(\hat{\gam})}{\partial\gam^{'}}-\E\left(\frac{\partial g(\gam_0)}{\partial\gam^{'}}\right)=o_p(1)$. Thus, we have 
\begin{align}
\sqrt{n}\left(\hat{\gam}-\gam_0\right)\stackrel{d}{\rightarrow} N\left(\mathbf{0}, \lim _{n \rightarrow \infty}(\mathbf{G}^{'} \pphi \mathbf{G})^{-1} \mathbf{G}^{'} \pphi \HH \pphi \mathbf{G}(\mathbf{G}^{'} \pphi \mathbf{G})^{-1}\right).
\end{align}
From the expression for the variance-covariance matrix of IGMME, we can see that the precision of the estimator can be improved by replacing the arbitrary weighing matrix $\pphi$ in the objective function with $\HH^{-1}$. The resulting GMME is called the optimal GMME \citep{Hansen:1982}. However, this estimator is not feasible as $\HH^{-1}$ is unknown. To make it feasible, a plug-in estimator $\hat{\HH}\equiv\HH(\hat{\gam})$ based on the initial GMME $\hat{\gam}$ can be formulated. Then, the feasible optimal GMME is defined by
\begin{align}
\hat{\gam}_o=\argmin_{\gam} g^{'}(\gam)\hat{\HH}^{-1} g(\gam).
\end{align}
Under some conditions, \citet{Jin:2015}  show that
\begin{align}\label{bgmme}
\sqrt{n}\left(\hat{\gam}_o-\gam_0\right)\stackrel{d}{\rightarrow} N\left(\mathbf{0}, \lim _{n \rightarrow \infty}(\mathbf{G}^{'}\HH^{-1}\mathbf{G})^{-1}\right).
\end{align} 
\citet{Jin:2015} determine the best set of moment functions that provide the most efficient GMME for the MESS(1,1) under homoskedasticity. Their idea is to decompose the components of the inverse of the variance-covariance matrix of the optimal GMME, and then use the Cauchy-Schwarz inequality in such a way that an upper bound on the inverse of the variance-covariance matrix that is free of arbitrary pieces of the moment functions (free of $\p_i$'s and $\F$) can be attained. The resulting GMME is termed as the best GMME (BGMME). When the disturbance terms are normally distributed, the BGMME turns out to be asymptotically as efficient as the QMLE. However, when the disturbance terms are not normally distributed, and $\W$ and $\M$ do not commute, the BGMME can be asymptotically more efficient than the QMLE. The best set of moment functions is 
\begin{align}
g^{*}(\gam)=\frac{1}{n}\left(\V^{'}(\gam)\p_1^{*}\V(\gam),\hdots,\V^{'}(\gam)\p_{k^{*}+4}^{*}\V(\gam),\V^{'}(\gam)\F^{*}\right)^{'},
\end{align}
 where $\p^{*}_1=\mathbb{W}$, $\p^{*}_2=\Diag(\mathbb{W})$, $\p^{*}_3=\Diag(\ermz\W\X\bet_0)^{(t)}$, $\p^{*}_4=\M$, $\p^{*}_{m+4}=\Diag(\ermz\X_{m})^{(t)}$ for $m=1,...,k^{*}$, and $\F^{*}=(\F^{*}_1,\F^{*}_2,\F^{*}_3,\F^{*}_4)$ with $\F_1^{*}=\ermz\X^{*}$, $\F_2^{*}=\ermz\W\X\bet_0$, $\F_3^{*}=\bm{l}$,  $\F_4^{*}=\vvec_D(\mathbb{W})$, where $\X^*$ excludes the intercept term in $\X$ if $\M$ is row-normalized so that $\F_1^{*}$ does not contain the intercept term generated in $\ermz\X$, $k^*$ is the number of columns in $\X^{*}$, $\A^{(t)}=\A-\mathbf{I}_n\tr(\A)/n$ for any $n\times n$ matrix $\A$ and $\bm{l}$ is an $n\times1$ vector of ones.  

\subsection{Estimation under heteroskedasticity} 

In this section, we consider the GMM estimation of MESS(1,1) under Assumption~\ref{a2}. Recall that $\boldsymbol{\Sigma}$ denotes the variance-covariance matrix of the disturbance terms, i.e., $\boldsymbol{\Sigma}=\Diag(\sigma_1^2,\hdots,\sigma_n^2)$. Similar to the homoskedastic case, we again employ the following vector of moment functions consisting of  $k_p$ quadratic moment functions and $k_f$ linear moment functions:
\begin{align*}
g(\gam)=\frac{1}{n}(\V^{'}(\gam)\p_1\V(\gam),\hdots,\V^{'}(\gam)\p_{k_p}\V(\gam),\V^{'}(\gam)\F)^{'}.
\end{align*}
At $\gam_0$, we have $\mathrm{E}\left(\V^{'}\p_m\V\right) = \tr\left(\p_m\boldsymbol{\Sigma}\right)=\tr\left(\boldsymbol{\Sigma}\Diag(\p_i)\right)$, which is equal to zero if the diagonal elements of $\p_m$ are zeros. Hence, in the heteroskedastic case, we require that the diagonal elements of $\p_m$ are zeros, i.e., $\Diag(\p_m)=\mf{0}$ for $m=1,2, \hdots, k_p$. Then, an initial GMME based on an arbitrary symmetric weighting matrix $\pphi$, with rank greater than or equal to $k+2$, can be defined as 
\begin{align}
\hat{\gam}=\argmin g^{'}(\gam)\pphi g(\gam). 
\end{align}
Let $\mathbf{G}=\E\left(\frac{\partial g(\gam_0)}{\partial\gam^{'}}\right)$ and $\HH=n \mathrm{E}\left(g\left(\gam_0\right) g^{\prime}\left(\gam_0\right)\right)$. By Lemma 2 in Appendix \ref{lemmas}, we can show that
\begin{align*}
\HH&=\frac{1}{n}\left(\begin{array}{cc}\frac{1}{2} \oomega^{'} \oomega & \mathbf{0}\\ 
\mathbf{0} &  \F^{'} \boldsymbol{\Sigma}\F\end{array}\right),
\end{align*} 
and  
\begin{align*}
\mathbf{G}&=\frac{1}{n}\left(\begin{array}{ccc} \mathbf{0}&\frac{1}{2} \oomega^{'} \vvec\left(\SSigma^{1/2}\left(\SSigma^{-1} \mathbb{W}\right)^s \SSigma^{1 / 2}\right) & \frac{1}{2} \oomega^{'} \operatorname{vec}\left(\SSigma^{1/2}\left(\SSigma^{-1} \M\right)^s \SSigma^{1/2}\right)  \\ 
 -\F^{'} e^{\tau_0 \M} \X&\F^{'} \mathbb{W} e^{\tau_0 \M} \X \bet_0 & \mathbf{0} \end{array}\right),
 \end{align*}
where $ \oomega=\vvec(\SSigma^{1/2}\p_1^{s}\SSigma^{1/2},\hdots,\SSigma^{1/2}\p_{k_p}^{s}\SSigma^{1/2})$. It follows again that 
\begin{align}
\sqrt{n}\left(\hat{\gam}-\gam_0\right)\stackrel{d}{\rightarrow} N\left(\mathbf{0}, \lim _{n \rightarrow \infty}(\mathbf{G}^{'} \pphi \mathbf{G})^{-1} \mathbf{G}^{'} \pphi \HH \pphi \mathbf{G}(\mathbf{G}^{'} \pphi \mathbf{G})^{-1}\right).
\end{align}
Note that $\mathbf{G}$ and $\mathbf{H}$ involve the unknown diagonal matrix $\boldsymbol{\Sigma}$. These terms can be consistently estimated by replacing $\boldsymbol{\Sigma}$ with $\Diag(v_1^2(\hat{\gam}),\hdots,v_n^2(\hat{\gam}))$. Let $\hat{\HH}$ be the plug-in estimator of $\HH$ based on the initial GMME $\hat{\gam}$. Then, a feasible optimal robust GMME (RGMME) can be obtained as 
\begin{align}\label{rgmme5.8}
\hat{\gam}_o=\argmin g^{'}(\gam)\hat{\HH}^{-1} g(\gam).
\end{align}
 It can be shown that
 \begin{align}
 \sqrt{n}\left(\hat{\gam}_o-\gam_0\right)\stackrel{d}{\rightarrow} N\left(\mathbf{0}, \lim _{n \rightarrow \infty}(\mathbf{G}^{'}\HH^{-1}\mathbf{G})^{-1}\right).
 \end{align} 
In the heteroskedastic case, the best set of moment functions is not feasible because the moment functions involve the unknown $\boldsymbol{\Sigma}$, which cannot be consistently estimated. In practice, we can formulate the RGMME based on the following vector of moment functions \citep{Jin:2015}: 
\begin{align}
g^{*}(\gam)=\frac{1}{n}\left(\V^{'}(\gam)(\hat{\mathbb{W}}-\Diag(\hat{\mathbb{W}}))\V(\gam),\V^{'}(\gam)\M\V(\gam),\V^{'}(\gam)\F\right)^{'},
\end{align}
 where $\F=(\hat{\mathbb{W}}e^{\hat{\tau}\M}\X\hat{\bet},e^{\hat{\tau}\M}\X)$ with $\hat{\mathbb{W}}=e^{\hat{\tau}\M} \W e^{-\hat{\tau}\M}$.

\section{Bayesian estimation approach} \label{sec:Bayesian}
\subsection{Estimation under homoskedasticity} 
Following \citet{Lesage:2007}, we assume the following independent prior distributions: $\lambda\sim N(\mu_\lambda ,V_{\lambda})$, $\rho\sim N(\mu_\rho,V_{\rho})$, $\bet \sim N(\boldsymbol{\mu}_{\bet}, \V_{\bet})$ and $\sigma^2 \sim IG(a, b)$, where $IG$ denotes the inverse-gamma distribution.  Under these prior distributions, the posterior distribution of parameters can be expressed as\footnote{We use $p(\cdot)$ to denote the relevant density functions, and ignore $\mf{X}$ in the conditional sets for the sake of simplicity.}
\begin{align*}
p(\thet|\Y)\propto p(\Y|\thet)p(\thet)=p(\Y|\thet)p(\bet)p(\sigma^2)p(\lambda)p(\rho),
\end{align*} 
where $p(\thet)$ is the joint prior distribution of $\thet$ and $p(\Y|\thet)$ is the likelihood function given as
\begin{align*}
p(\mf{Y}|\thet)=(2\pi\sigma^2)^{-n/2}\exp\left(-\frac{1}{2\sigma^2}(\elw \Y-\X\bet)^{'}\ermp\erm(\elw \Y-\X\bet)\right).
\end{align*}
Algorithm \ref{alg1} describes a Gibbs sampler that can be used to generate random draws from $p(\thet|\Y)$. 
\begin{algorithm}[Estimation of \eqref{2.1} under homoskedasticity]\label{alg1}\hfill
\begin{enumerate}
\item Sampling step for $\bet$:
\begin{equation*}
\bet|\Y, \lambda, \rho, \sigma^2 \sim N(\hat{\bet}, \K_{\bet}),
\end{equation*}
where $\K_{\bet}=(\V_{\bet}^{-1}+\sigma^{-2}\X^{'}e^{\rho \M^{'}}e^{\rho \M}\X)^{-1}$ and $\hat{\bet}=\K_{\bet}(\sigma^{-2}\X^{'}e^{\rho \M^{'}}e^{\rho \M}e^{\lambda \W}\Y+\V_{\bet}^{-1}\boldsymbol{\mu}_{\bet})$.
\item Sampling step for $\sigma^2$:
\begin{equation*}
\sigma^2 | \Y, \lambda, \rho, \bet \sim IG(\hat{\sigma}^2, K_{\sigma^2}),
\end{equation*}
where $\hat{\sigma}^2=a+\frac{n}{2}$ and $K_{\sigma^2}=b+\frac{1}{2}(e^{\lambda \W}\Y-\X\bet)^{'}e^{\rho \M^{'}}e^{\rho \M}(e^{\lambda \W}\Y-\X\bet)$.
\item Sampling step for $\lambda$:
\begin{align*}
&p(\lambda|\Y, \bet, \rho, \sigma^2)\\
&\propto \exp\left(-\frac{1}{2}\left(\sigma^{-2}(e^{\lambda \W}\Y-\X\bet)^{'}e^{\rho \M^{'}}e^{\rho \M}(e^{\lambda \W}\Y-\X\bet)+V_{\lambda}^{-1}(\lambda^2-2\mu_{\lambda} \lambda)\right)\right).\nonumber
\end{align*}
Generate a candidate value $\lambda^{new}$ according to
\begin{equation*}
\lambda^{new}=\lambda^{old}+c_{\lambda}\times N(0,1),
\end{equation*}
where $c_{\lambda}$ is a tuning parameter.\footnote{The tuning parameter is determined during the estimation such that the acceptance rate falls between $40\%$ and $60\%$ \citep{Lesage:2009}.} Then, accept the candidate value $\lambda^{new}$ with probability 
\begin{equation*}
\mathbb{P}(\lambda^{new}, \lambda^{old})=\min\left(1, \frac{p(\lambda^{new}|\Y,\bet, \sigma^2, \rho)}{p(\lambda^{old}|\Y,\bet, \sigma^2, \rho)}\right).
\end{equation*} 
\item Sampling step for $\rho$:
\begin{align*}
&p(\rho|\Y, \bet, \lambda, \sigma^2)\\
&\propto \exp\left(-\frac{1}{2}\left(\sigma^{-2}(e^{\lambda \W}\Y-\X\bet)^{'}e^{\rho \M^{'}}e^{\rho \M}(e^{\lambda \W}\Y-\X\bet)+V_{\rho}^{-1}(\rho^2-2\mu_{\rho} \rho)\right)\right).\nonumber
\end{align*}
Use the random-walk Metropolis-Hastings algorithm described in Step 3 to generate random draws from $p(\rho|\Y, \bet, \lambda, \sigma^2)$.
\end{enumerate}
\end{algorithm}
In Algorithm~\ref{alg1}, the conditional posterior distributions of $\bs{\beta}$ and $\sigma^2$ are determined from $p(\bet|\Y, \lambda, \rho, \sigma^2)\propto p(\mf{Y}|\bs{\theta})p(\bs{\beta})$ and $p(\sigma^2 | \Y, \lambda, \rho, \bet)\propto p(\mf{Y}|\bs{\theta})p(\sigma^2)$, respectively. Since we assume conjugate priors for $\bs{\beta}$ and $\sigma^2$, these conditional posterior distributions take known forms as shown in Algorithm \ref{alg1}. The Bayesian argument used to determine these conditional posterior distributions is analogous to the one used for a linear regression model. On the other hand, the conditional posterior distributions of spatial parameters are non-standard because the likelihood function is non-linear in terms of these parameters. To sample these parameters, we use the random walk Metropolis-Hastings algorithm suggested by \citet{Lesage:2009}. 

\subsection{Estimation under heteroskedasticity} 
Following  \citet{Lesage:1997} and \citet{Lesage:2009}, we assume that the disturbance terms have a scale mixture of normal distributions such that the scale mixture variables generate different distributions with distinct variance terms. Thus, we have $v_i|\eta_i\sim N(0,\eta_i\sigma^2)$, where $\eta_i$'s are independent scale mixture variables with $\eta_i\sim \text{IG}(\nu/2,\nu/2)$ for $i=1,\hdots,n$. Let $\bs{\eta}=(\eta_1,\hdots,\eta_n)^{'}$ and $\mf{H}(\bs{\eta})=\text{Diag}\left(\eta_1,\hdots,\eta_n\right)$ be the $n\times n$ diagonal matrix with the $i$th diagonal element $\eta_i$. Then, we can derive the conditional likelihood function $p(\mf{Y}|\bs{\theta},\bs{\eta})$ as
\begin{align}\label{6.1}
p(\mf{Y}|\bs{\theta},\bs{\eta})&=(2\pi\sigma^2)^{-n/2}\left(\prod_{i=1}^n\eta_{i}\right)^{-1/2}\\
&\times\exp\left(-\frac{1}{2\sigma^2}\left(e^{\lambda \mf{W}}\mf{Y}-\mf{X}\bs{\beta}\right)^{'}e^{\rho \mf{M}^{'}}\mf{H}^{-1}(\bs{\eta})e^{\rho\mf{M}}\left(e^{\lambda\mf{W}}\mf{Y}-\mf{X}\bs{\beta}\right)\right).\nonumber
\end{align}
To introduce a Bayesian estimation approach, we adopt the prior distributions assumed in Section 6.1 for $\bs{\beta}$, $\lambda$, $\rho$ and $\sigma^2$. In the heteroskedastic case, we also need to determine a prior distribution for $\nu$. To that end, we note that  the marginal distribution of $v_i$ is a $t$ distribution with mean zero, scale parameter $\sigma^2$ and $\nu$ degrees of freedom, i.e., $v_i\sim t_\nu(0,\sigma^2)$. Thus, we assume the following prior $\nu\sim \text{Uniform}(2,\bar{\nu})$, where $\text{Uniform}(a,b)$ denotes the uniform distribution over the interval $(a,b)$, and $\bar{\nu}$ is a known positive number. This prior distribution ensures  that the variance of $v_i$ exists because $\nu>2$. Also, we can set $\bar{\nu}$ to a large positive number so that the $t$ distribution is allowed to approximate the normal distribution well-enough. 

The posterior distribution of parameters  then takes the following form:
\begin{align*}
p(\thet,\eeta|\Y)\propto p(\Y|\thet,\eta)p(\thet,\eeta)=p(\Y|\thet,\eeta)p(\bet)p(\sigma^2)p(\lambda)p(\rho)p(\eeta|\nu)p(\nu),
\end{align*} 
where $p(\Y|\thet,\eeta)$ is the conditional likelihood function stated in \eqref{6.1} and $p(\thet,\eeta)$ is the joint prior distribution of $\thet$ and $\eeta$. Algorithm~\ref{alg2} describes a Gibbs sampler that can be used to generate random draws from $p(\thet,\eeta|\Y)$. 
\begin{algorithm}[Estimation of \eqref{2.1} under heteroskedasticity]\label{alg2}\hfill
\begin{enumerate}
\item Sampling step for $\bet$:
\begin{equation*}
\bet|\Y, \lambda, \rho, \sigma^2,\eeta \sim N(\hat{\bet}, \K_{\bet}),
\end{equation*}
where $\K_{\bet}=(\V_{\bet}^{-1}+\sigma^{-2}\X^{'}e^{\rho \M^{'}}\mathbf{H}^{-1}(\eeta)e^{\rho \M}\X)^{-1}$,  $\mathbf{H}(\eeta)=\Diag(\eta_1,\hdots,\eta_n)$ and\\ $\hat{\bet}=\K_{\bet}(\sigma^{-2}\X^{'}e^{\rho \M^{'}}\mathbf{H}^{-1}(\eeta)e^{\rho \M}e^{\lambda \W}\Y+\V_{\bet}^{-1}\boldsymbol{\mu}_{\bet})$.
\item Sampling step for $\sigma^2$:
\begin{equation*}
\sigma^2 | \Y, \lambda, \rho, \bet,\eeta \sim IG(\hat{\sigma}^2, K_{\sigma^2}),
\end{equation*}
where $\hat{\sigma}^2=a+\frac{n}{2}$ and $K_{\sigma^2}=b+\frac{1}{2}(e^{\lambda \W}\Y-\X\bet)^{'}e^{\rho \M^{'}}\mathbf{H}^{-1}(\eeta)e^{\rho \M}(e^{\lambda \W}\Y-\X\bet)$.
\item Sampling step for $\lambda$:
\begin{align*}
&p(\lambda|\Y, \bet, \rho, \sigma^2,\eeta)\\
&\propto \exp\left(-\frac{1}{2}\left(\sigma^{-2}(e^{\lambda \W}\Y-\X\bet)^{'}e^{\rho \M^{'}}\mathbf{H}^{-1}(\eeta)e^{\rho \M}(e^{\lambda \W}\Y-\X\bet)+\V_{\lambda}^{-1}(\lambda^2-2\mu_{\lambda} \lambda)\right)\right).\nonumber
\end{align*}
Use the random-walk Metropolis-Hastings algorithm described in Step 3 of Algorithm~\ref{alg1} to sample this parameter. 
\item Sampling step for $\rho$:
\begin{align*}
&p(\rho|\Y, \bet, \lambda, \sigma^2,\eeta)\\
&\propto \exp\left(-\frac{1}{2}\left(\sigma^{-2}(e^{\lambda \W}\Y-\X\bet)^{'}e^{\rho \M^{'}}e^{\rho \M}(e^{\lambda \W}\Y-\X\bet)+\V_{\rho}^{-1}(\rho^2-2\mu_{\rho} \rho)\right)\right).\nonumber
\end{align*}
Use the  random-walk Metropolis-Hastings algorithm described in Step 3 of Algorithm~\ref{alg1} to generate random draws from $p(\rho|\Y, \bet, \lambda, \sigma^2,\eeta)$.
\item Sampling step for $\eeta$:
\begin{align*}
\eta_{i}|\Y,\lambda,\rho,\bet,\sigma^2,\nu\sim\text{IG}\left(\frac{\nu+1}{2},\,\frac{\nu}{2}+\frac{Y^2_{i}(\gam)}{2\sigma^2}\right)\quad\text{for}\quad i=1,2,\hdots,n,
\end{align*}
where $Y_{i}(\gam)$ is the $i$th element of $\Y(\gam)=e^{\rho \M}\left(e^{\lambda \W}\Y-\X\bet\right)$.
\item Sampling step for $\nu$: 
\begin{align*}
p(\nu|\eeta)\propto\frac{(\nu/2)^{n\nu/2}}{\Gamma^n(\nu/2)}\left(\prod_{i=1}^n\eta_{i}\right)^{-(\frac{\nu}{2}+1)}\exp\left(-\sum_{i=1}^n\frac{\nu}{2\eta_{i}}\right).
\end{align*}
Use the Griddy-Gibbs sampler to sample this parameter.
\end{enumerate}
\end{algorithm}
The conditional posterior distributions of $\bs{\beta}$, $\bs{\eta}$, and $\sigma^2$ take known forms as shown in Algorithm~\ref{alg2}. In the case of spatial parameters, we again resort to the random walk Metropolis-Hastings algorithm suggested by \citet{Lesage:2009}. The conditional posterior distribution of $\nu$ can be determined from $p(\nu|\mf{Y},\bs{\beta},\bs{\eta}, \lambda,\rho,\sigma^2)=p(\nu|\bs{\eta})\propto p(\eeta|\nu)p(\nu)$. However, this distribution does not take any known form.  Since $\nu$ has support over $(2,\bar{\nu})$, we suggest using the Griddy-Gibbs sampler to sample this parameter. Algorithm~\ref{alg3} describes the Griddy-Gibbs sampler.
\begin{algorithm}[The Griddy-Gibbs sampler for $\nu$]\label{alg3}\hfill
\begin{enumerate}
 \item Construct a grid of points $\nu_1,\hdots,\nu_m$ from the interval $(2,\bar{\nu})$.
 \item Compute $p_i=\sum_{j=1}^{i}p(\nu_j|\eeta)$ for $i=1,\hdots,m$, and generate $u$ from $\text{Uniform}(0,1)$.
 \item Determine the smallest $k$ such that $p_k\geq u$ and return $\nu=\nu_k$.
 \end{enumerate}
\end{algorithm}

\section{Estimation in the presence of endogenous and Durbin regressors} \label{sec:endogenous}
The preceding sections consider a regression model with spatial dependence specified by the MESS processes, where no endogenous regressors are included. In this section, we consider a MESS model with endogenous and Durbin regressors. The popular nonlinear two-stage least squares (N2SLS) estimation method in such a setting can have some irregular features.

Consider the following model:
\begin{equation} \label{mod:endog}
	e^{\lambda_0\W}\Y = \X^* \bm\beta_{10}+\W \bm{l}\bm\beta_{20} + \W \X_1 \bm{\beta}_{30} +\Z\bm{\beta}_{40}+\V,
\end{equation}   
where $\bm{l}$ is an $n\times 1$ vector of ones, $\X_1 $ excludes the intercept term from the exogenous variable matrix $\X$, $\Z$ is an $n\times k_z$ matrix of endogenous regressors,  and $\X^*=\X=[\bm{l}, \X_1]$ if $\W$ is not row-normalized to have row sums equal to one, and   $\X^{*}=\X_1$ otherwise. The $\bm\beta_{10}$, $\bm\beta_{20}$, $\bm\beta_{30}$ and $\bm\beta_{40}$ are conformable true parameters, and $\W$, $\Y$ and $\V$ have the same meanings as those in  \eqref{2.1}.   The Durbin regressors $\W \X_1$ are neighbors' characteristics and capture exogenous externalities. When  $\W$ is row-normalized, $\W \bm{l}=\bm{l}$ is the intercept term; when $\W$ is not row-normalized, $\W \bm{l}$ is also a Durbin regressor. In particular, if $\W$ is not row-normalized and has binary elements, $\W \bm{l}$ is a vector of out-degrees that measure the overall numbers of links for each spatial unit. Model \eqref{mod:endog} includes Durbin regressors explicitly since the MESS structure and the Durbin regressors lead to some irregular features of the N2SLS estimator. Model \eqref{2.1} has not considered Durbin regressors explicitly but can allow for that, where the theoretical analysis will not be affected although the related expressions for estimators need to be modified accordingly.  To focus on the N2SLS estimation, a MESS process for the disturbances is not considered in  \eqref{mod:endog}.\footnote{If there is a MESS process for the disturbances, then as in \cite{Jin.Wang2022}, the GMM estimation with both linear and quadratic moments can be considered, since instrumental variables alone are not enough to identify parameters for the disturbance process.}    

Let $\F$ be an $n\times k_f$ full rank IV matrix for the N2SLS estimation, where $k_f$ is not smaller than the total number of parameters in $\bm\theta=(\lambda, \bm\beta')'$ with $\bm\beta = (\bm\beta_1',\bm\beta_2,\bm\beta_3',\bm\beta_4')'$. For example, $\F$ can be the matrix formed by the independent columns of $[\bm{l},\X_1,\W\bm{l},\W\X_1,\W^2\bm{l},\W^2\X_1,\bar\Z]$, where $\bar\Z$ is the IV matrix for $\Z$.\footnote{If $\W$ is row-normalized, then $\W\bm{l}$ and $\W^2\bm{l}$ are redundant.} Assume that the elements of $\V$ are independent conditional on $\F$ but can have different conditional variances so that $\bm\Sigma=\E(\V \V'|\F)$ is a diagonal matrix of conditional variances. Denote
 $\D=[\X^*,\W\bm{l},\W\X_1,\Z]$ and $\bm\Pi=\F'\bm\Sigma\F$. The infeasible N2SLS estimation, as if $\bm\Sigma$ were known, has the objective function
 \begin{equation}
 	Q(\bm\theta) = (e^{\lambda\W}\Y-\D\bm\beta)' \F\bm\Pi^{-1}\F'(e^{\lambda\W}\Y-\D\bm\beta).
 \end{equation}
The N2SLS estimator $\hat{\bm\theta}$ derived by minimizing  $Q(\bm\theta)$ is consistent under regularity conditions. 

Let $\bm\delta =(\bm\beta_1',\bm\beta_2)'$ and $\bm\xi=(\bm\beta_3',\bm\beta_4')'$ when $\W$ is row-normalized, and let $\bm\delta=\bm\beta_1$ and $\bm\xi=(\bm\beta_2,\bm\beta_3',\bm\beta_4')'$ when $\W$ is not row-normalized. Then, $\bm\xi$ contains the coefficients for the Durbin and endogenous regressors. When $\bm\xi_0\ne 0$, all components of $\hat{\bm\theta}$ are $\sqrt n$-consistent and $\hat{\bm\theta}$ has the asymptotic distribution
\begin{equation}\label{dis:regular}
	\sqrt n(\hat{\bm\theta} -\bm\theta_0)\xrightarrow{d} N\Bigl(0,\lim_{n\to\infty}\Bigl\{ \frac{1}{n} \E[(-\W\D\bm\beta_0,\D)'\F] \bar{\bm\Pi}^{-1} \E[\F'(-\W\D\bm\beta_0,\D)] \Bigr\}^{-1}\Bigr).
\end{equation}  
However, some components of $\hat{\bm\theta}$ have a rate of convergence slower than $\sqrt n$ and are not asymptotically normal in the case that $\bm\xi_0= 0$, i.e., the Durbin and endogenous regressors are irrelevant, which is unknown when estimation is considered. 

When $\bm\xi_0=0$, we have 
\begin{equation*}
	\frac{1}{\sqrt{n}}\frac{\partial Q(\bm\theta_0)}{\partial\lambda} +\frac{1}{\sqrt{n}}\frac{\partial Q(\bm\theta_0)}{\partial\bm\beta'} (0_{1\times k^*},\bm\delta_{20},\bm\delta_{10}',0_{1\times k_z})'=o_p(1),
\end{equation*} 
where $k^*$ is the number of columns in $\X^*$, $\bm\delta_{20}$ is the last element of $\bm\delta_{0}$ and $\bm\delta_{10}$ contains the remaining elements. Thus, $\frac{1}{\sqrt{n}}\frac{\partial Q(\bm\theta_0)}{\partial\lambda}$ and $\frac{1}{\sqrt{n}}\frac{\partial Q(\bm\theta_0)}{\partial\bm\beta}$ are linearly dependent with probability approaching one (w.p.a.1.). As a result, $\frac{1}{n}\frac{\partial Q(\bm\theta_0)}{\partial\bm\theta} \frac{\partial Q(\bm\theta_0)}{\partial\bm\theta'}$ is singular w.p.a.1. In addition, we can show that $\frac{1}{n}\frac{\partial^2 Q(\bm\theta_0)}{\partial\bm\theta \partial\bm\theta'} $ is also singular for large $n$. Hence, the usual method of deriving the asymptotic distribution of an estimator based on the mean value theorem expansion of the first order condition will not work.  

The asymptotic distribution of $\hat{\bm\theta}$ in the case with $\bm\xi_0=0$ can be derived by first raparameterizing the model so that the derivative of the new N2SLS objective function with respect to a new parameter is exactly zero and then investigating a third order Taylor expansion of the first order condition at the true parameter
vector.  Let $\bar{\bm\Pi}=\E(\bm\Pi)$, $k_d$ be the number of columns in $\D$, $J$ be a random vector that follows the normal distribution $N(0,\bm\Delta)$, where 
\begin{equation*}
	\bm\Delta = \lim_{n\to\infty} 
	\begin{pmatrix}
		2 & 0 \\ 0 & \mf{I}_{k_d}
	\end{pmatrix} 
	\Bigl( \frac{1}{n}\E[(-\W^2\X\bm\delta_0,\D)'\F]\bar{\bm\Pi}^{-1} \E[\F'(-\W^2\X\bm\delta_0,\D)] \Bigr)^{-1} 
	\begin{pmatrix}
	2 & 0 \\ 0 & \mf{I}_{k_d}
	\end{pmatrix},
\end{equation*}  
and $L= J_2 - \lim_{n\to\infty}[\frac{2}{n}\E(\D'\F)\bar{\bm\Pi}^{-1} \E(\F'\D)]^{-1} \frac{1}{n}\E(\D'\F)\bar{\bm\Pi}^{-1} \E(\F'\W^2\X)\bm\delta_0 J_1$, where $J_1$ is the first element of $J$ and $J_2$ contains the remaining elements of $J$. Then, in the case with $\bm\xi_0=0$, the N2SLS estimator $\hat{\bm\theta}=(\hat\lambda,\hat{\bm\beta}_1',\hat{\bm\beta}_2,\hat{\bm\beta}_3',\hat{\bm\beta}_4')'$ has the asymptotic distribution
\begin{equation}\label{dis:irregular}
	\begin{pmatrix}
		n^{1/4}(\hat\lambda-\lambda_0) \\ n^{1/2}(\hat{\bm\beta}_1-\bm\beta_{10}) \\ n^{1/4}(\hat{\bm\beta}_2-\bm\beta_{20}) \\ n^{1/4}(\hat{\bm\beta}_3-\bm\beta_{30}) \\ n^{1/2}(\hat{\bm\beta}_4-\bm\beta_{40})
	\end{pmatrix} \xrightarrow{d} 
	\begin{pmatrix}
	(-1)^B J_1^{1/2} \\ J_{2x^*} \\ (-1)^{B}\bm\delta_{20}J_1^{1/2} \\ (-1)^B\bm\delta_{10}J_1^{1/2} \\ J_{2z}
	\end{pmatrix} I(J_1>0) + \begin{pmatrix} 0 \\ L_{x^*} \\ 0_{k\times 1} \\ L_z \end{pmatrix} I(J_1<0),
\end{equation}
where $I(\cdot)$ denotes the indicator function,  $J_{2x^*}$ and $L_{x^*}$ are vectors consisting of the first $k^*$ elements of $J_2$ and $L$ respectively, $J_{2z}$ and $L_z$ are vectors consisting of the last $k_z$ elements of $J_2$ and $L$ respectively, and $B$ is a Bernoulli random variable with success probability described in \cite{Jin.Lee2018}. Thus, only $\hat{\bm\beta_1}$ and $\hat{\bm\beta_4}$ are $\sqrt n$-consistent, and the remaining components of $\hat{\bm\theta}$ have a slow rate $n^{1/4}$ of convergence and  follow non-standard asymptotic distributions. 

The above N2SLS estimator is an infeasible estimator as $\bm\Pi$ is unknown. A feasible N2SLS estimator can be derived as follows. We may first derive an initial consistent but inefficient N2SLS estimator, e.g., the minimizer $\check{\bm{\theta}}$ of $(e^{\lambda \W}\Y -\D\bm{\beta})'\F(\F'\F)^{-1}\F'(e^{\lambda \W}\Y -\D\bm{\beta})$, and then consider the feasible N2SLS estimation with the objective function $\check Q(\bm\theta) =(e^{\lambda \W}\Y -\D\bm{\beta})'\F(\F'\check{\bm\Sigma}\F)^{-1}\F'(e^{\lambda \W}\Y -\D\bm{\beta})$, where $\check{\bm\Sigma}=\text{Diag}(\check v_1^2,\cdots,\check v_n^2)$ with  $\check v_i$ the $i$th element of $e^{\check\lambda \W}\Y -\D\check{\bm{\beta}}$. The feasible N2SLS estimator $\tilde{\bm\theta}$ has the same asymptotic distribution as the infeasible estimator $\hat{\bm\theta}$. 

As $\bm\xi_0=0$ and  $\bm\xi_0\ne 0$ lead to different asymptotic distributions of $\hat{\bm\theta}$, \cite{Jin.Lee2018} propose several tests for the hypothesis that $\bm\xi_0=0$. Depending on whether $\bm\xi_0=0$ is rejected or not, inference can be based on \eqref{dis:regular} or \eqref{dis:irregular}. Consider the case with $\bm{\xi}_0\ne 0$ as an example. By \eqref{dis:regular}, the variance of $\tilde{\bm\theta}$ can be estimated by  $[(-\W\D\tilde{\bm\beta},\D)'\F (\F'\tilde{\bm\Sigma}\F)^{-1} \F'(-\W\D\tilde{\bm\beta},\D)]^{-1}$, where  $\tilde{\bm\Sigma}=\text{Diag}(\tilde v_1^2,\cdots,\tilde v_n^2)$ with  $\tilde v_i$ the $i$th element of $e^{\tilde\lambda \W}\Y -\D\tilde{\bm{\beta}}$. 

An interesting alternative estimation method is the adaptive group LASSO (AGLASSO), which can implement model selection and estimation simultaneously. The resulting estimator has the oracle properties \citep{Fan.Li2001}, so that the true model can be selected w.p.a.1.\ and the estimator always has the $\sqrt n$-rate of convergence and asymptotic normal distribution. The AGLASSO objective function to be minimized is
\begin{equation}
\frac{1}{n}\check Q(\bm\theta) + \alpha_n \|\check{\bm\xi}\|^{-\mu}\|\bm\xi\|,
\end{equation} 
where $\alpha_n$ is a tuning parameter that is positive and converges to zero, $\check{\bm\xi}$ is an initial consistent estimator, and $\mu$ is some positive number such as $1$ or $2$. Under regularity conditions, the AGLASSO estimator $\dot{\bm\theta}$ is consistent. In the case that $\bm{\xi}_0=0$, the probability that $\dot{\bm\xi}=0$ goes to one as $n$ goes to infinity, that is, $\dot{\bm\theta}$ has the sparsity property, and for the remaining parameters $\bm\psi=(\lambda,\bm\delta')'$, the AGLASSO estimator has an asymptotic normal distribution as if $\bm\xi_0$ were known: 
\begin{equation*}
	\sqrt n(\dot{\bm\psi} - \bm\psi_0)\xrightarrow{d} N\Bigl(0,\lim_{n\to\infty}\frac{1}{n}\bigl\{ \E[(-\W\X\bm\delta_0,\X)'\F]\bar{\bm\Pi}^{-1} \E[\F'(-\W\X\bm\delta_0,\X)] \bigr\}^{-1} \Bigr),
\end{equation*}
in the case that $\bm\xi_0\ne 0$, under the condition that $\alpha_n =o(n^{-1/2})$ and other regularity conditions, $\dot{\bm\theta}$ has the same asymptotic normal distribution as that stated in \eqref{dis:regular}. Similar to the variance estimation of $\tilde{\bm\theta}$,  the variance of $\dot{\bm\psi}$ for the case with $\bm{\xi}_0=0$ can be estimated by  $[(-\W\X\dot{\bm\delta},\X)'\F (\F'\dot{\bm\Sigma}\F)^{-1} \F'(-\W\X\dot{\bm\delta},\X)]^{-1}$, where $\dot{\bm\Sigma}$ is defined similarly to $\tilde{\bm\Sigma}$.  
  
A practical question for the AGLASSO estimator is the selection of the tuning parameter $\alpha_n$. We can use an information criterion to choose $\alpha_n$. To make the dependence of $\dot{\bm\theta}$ on $\alpha_n$ explicit, denote the minimizer of $\frac{1}{n}\check Q(\bm\theta) + \alpha \|\tilde{\bm\xi}\|^{-\mu}\|\bm\xi\|$ by $\dot{\bm\theta}_\alpha$. Correspondingly, the AGLASSO estimator of $\bm\xi$ is $\dot{\bm\xi}_\alpha$. Consider the following information criterion:
\begin{equation*}
	h_n(\alpha) = \frac{1}{n}\check Q(\dot{\bm\theta}_\alpha) - I(\dot{\bm\xi}_\alpha=0)\Gamma_n,
\end{equation*}
where $\Gamma_n>0$ satisfies $\Gamma_n\to 0$ and $n^{1/2}\Gamma_n\to \infty$ as $n\to\infty$. For example, we may take $\Gamma_n=O(n^{-1/4})$. The tuning parameter chosen by minimizing $h_n(\alpha)$ can achieve model selection consistency.

The Monte Carlo results presented in \cite{Jin.Lee2018} show that the N2SLS and AGLASSO estimators have similar performance in the regular case with $\bm{\xi}_0\ne 0$, but the AGLASSO estimator performs significantly better in the irregular case with  $\bm{\xi}_0= 0$.  Thus, we suggest the use of the AGLASSO estimator.

\section{A fast computational approach} \label{sec:computation}
There are alternative methods in the literature that can be used to compute $e^{\alpha \mf{A}}$, where $\alpha$ is a scalar parameter and $\mf{A}$ is an $n\times n$ matrix. The computation methods include the Taylor series approximation, Pad{\'e} approximation, ordinary differential equation methods, polynomial methods,  matrix decomposition methods, splitting methods and Krylov space methods. \citet{Moler:1978, Moler:2003} assess the effectiveness of nineteen methods according to the following attributes: (i) generality, (ii) reliability, (iii) stability, (iv) accuracy, (v) efficiency, (vi) storage requirements, (vii) ease of use, and (viii) simplicity. They conclude that though ``none (of the methods in their paper) are completely satisfactory,'' a scaling and squaring method with either the rational Pad{\'e} or Taylor approximants can be the most effective one to compute the matrix exponential terms.  Popular software such as Python, R, MATLAB and Mathematica provides functions that can be used to compute the matrix exponential of a given matrix. For example, MATLAB (function  \texttt{expm}), Mathematica (function \texttt{MatrixExp}) and Python (function \texttt{scipy.linalg.expm}) utilize a scaling and squaring method combined with a Pad{\'e} approximation for the computation of matrix exponential terms. 

As pointed out by \citet{Moler:1978, Moler:2003}, all methods suggested in the literature are ``dubious'' in the sense that a sole method may not be entirely reliable for all applications. For example, in the context of MESS-type models, the  scaling and squaring method combined with the Pad{\'e} approximation as implemented in MATLAB through its \texttt{expm} function can be highly costly in terms of computation time \citep{yangfast}.  Our analysis on the estimation of the MESS(1,1) model indicates that we need to compute terms such as $\elw \erm \Y$ and $\erm \X$. \citet{Lesage:2007} suggest that we should provide approximations to $\elw \erm \Y$ and $\erm \X$ in terms of the matrix-vector product terms instead of providing approximations to $\elw$ and $\erm$. This matrix-vector product method can reduce the computation time significantly. 

In the following, we show how to apply the matrix-vector product method  to $\elw \erm \Y$ and $\erm \X$. Let $\text{Diag}(a_1,\hdots,a_n)$ be the $n\times n$ diagonal matrix with the $i$th diagonal element $a_i$. We first consider $\elw \erm \Y$. We can truncate the matrix exponential terms at the $(q+1)$th order and express $\elw \erm \Y$ as
\begin{align}\label{8.1}
\erm\elw\Y&\approx\sum_{i=0}^{q} \frac{\rho^i \M^i}{i!}\sum_{j=0}^{q} \frac{\lambda^j \W^j}{j!}\Y\nonumber\\
&= \sum_{i=1}^q  \sum_{j=0}^{i-1}\frac{\rho^i \lambda^j \M^i \W^j}{i!j!}\Y+\sum_{i=1}^q  \sum_{j=0}^{i-1}\frac{\rho^j  \lambda^i  \M^j \W^i}{i!j!}\Y+\sum_{i=0}^{q}\frac{\rho^i \lambda^i \M^i \W^i}{(i!)^2}\Y\nonumber\\
&=\Y_1\D_1\kap_1(\lambda, \rho)+\Y_2\D_2\kap_2(\lambda, \rho)+\Y_3\D_3\kap_3(\lambda, \rho),
\end{align} 
where 
\begin{align*}
&\Y_1=\left[\M\Y,    \M^2\Y,  \M^2\W\Y,   \M^3\Y,  \M^3\W\Y,   \M^3\W^2\Y,  \hdots,   \M^q\Y,   \M^q\W\Y,  \hdots,   \M^q\W^{q-1}\Y\right],  \\
&\Y_2=\left[\W\Y,  \W^2\Y, \M\W^2\Y, \W^3\Y,  \M\W^3\Y, \M^2\W^3\Y, \hdots,  \W^q\Y,  \M\W^q\Y, \hdots, \M^{q-1}\W^q\Y\right], \\
&\Y_3=\left[\Y, \M\W\Y, \M^2\W^2\Y, \hdots, \M^q\W^q\Y\right],\\
& \D_1=\D_2=\text{Diag}\left(\frac{1}{0!1!}, \frac{1}{0!2!},\frac{1}{1!2!},\hdots,\frac{1}{0!q!},\hdots,\frac{1}{(q-1)!q!}\right),\\
&\D_3=\text{Diag}\left(\frac{1}{(0!)^2},\frac{1}{(1!)^2},\frac{1}{(2!)^2},\hdots,\frac{1}{(q!)^2}\right),\\
&\kap_1(\lambda, \rho)=\left[\rho, \rho^2, \rho^2\lambda, \rho^3,  \rho^3\lambda,  \rho^3\lambda^2,  \hdots, \rho^q, \rho^q\lambda, \hdots, \rho^q\lambda^{q-1}\right]^{'}, \\
&\kap_2(\lambda, \rho)=\left[\lambda, \lambda^2, \lambda^2\rho,\lambda^3,  \lambda^3\rho, \lambda^3\rho^2,\hdots, \lambda^q,  \lambda^q\rho,\hdots, \lambda^q\rho^{q-1}\right]^{'}, \\
&\kap_3(\lambda, \rho)=\left[1, \rho \lambda,  \rho^2 \lambda^2, \rho^3 \lambda^3,  \hdots, \rho^q \lambda^q\right]^{'}.
\end{align*}
The result in \eqref{8.1} expresses $\erm\elw\Y$ in terms of $\mf{Y}_j$ and $\mf{D}_j$ for $j\in\{1,2,3\}$. These terms can be computed once, and then supplied as the inputs of the objective function in an optimization solver.  

Let $\X=[\X_1,\X_2,\hdots,\X_k]$, where $\X_i$ is the $i$th column of $\X$. Then, we can express $\erm \X$ as
\begin{align}\label{eq8.2}
\erm\X
&=\left[\erm\X_1, \erm\X_2,\hdots, \erm\X_k\right]\approx \mathbb{X}\D_4\bs{\kappa}_4(\rho),
\end{align}
where 
\begin{align*}
&\mathbb{X}=\left[\X_1, \M\X_1,\hdots, \M^q\X_1, \X_2, \M\X_2, \hdots, \M^q\X_2, \hdots, \X_k, \M\X_k, \hdots,\ \M^q\X_k\right],\\
&\D_4=\mf{I}_k\otimes\text{Diag}\left(\frac{1}{0!},\frac{1}{1!},\frac{1}{2!},\hdots,\frac{1}{q!} \right),\\
&\bs{\kappa}_4(\rho)=\mf{I}_k\otimes\left[1,\rho,\rho^2,\hdots,\rho^q\right]^{'}.
\end{align*}
The approximation in \eqref{eq8.2} indicates that the computation of $\erm\X$ also requires only the matrix-vector product operations.  We can compute $\mathbb{X}$ and $\D_4$ only one time and then pass these terms as the inputs of the objective function in an optimization solver.

In an extensive Monte Carlo simulation study, \citet{yangfast} compared the computation time required by the matrix-vector product method with the \texttt{expm} function of MATLAB. For the QMLE, they demonstrated that the matrix-vector product method reduced computation time by $98\%$ to $99\%$ compared to the \texttt{expm} function. In the case of GMME, the computation time decreased by $95\%$ to $97\%$. In the context of the Bayesian estimator, the computation time was reduced by at least $99\%$.

\section{Impact measures}\label{sec:impact}
For the MESS(1,1) model in \eqref{2.1}, the marginal effect of a change in $\X_k$ on $\E(\Y)$ is given by $\elwzn\bet_{0k}$,  where $\bet_{0k}$ is the $k$th element of the true coefficient vector $\bet_0$. \citet{Lesage:2009} define three scalar measures for the marginal effect to ease the interpretation and presentation of this marginal effect: 
\begin{enumerate}
	\item Average Direct Impact (ADI): $\frac{1}{n}\tr(e^{-\hat{\lambda}\W}\hat{\bet}_{k})$, 
	\item Average Indirect Impact (AII): $\frac{1}{n}\left(\hat{\bet}_{k}\bm{l}' e^{-\hat{\lambda}\W}\bm{l}-\tr(e^{-\hat{\lambda}\W}\hat{\bet}_{k})\right)$, 
	\item Average Total Impact (ATI): $\frac{1}{n}\hat{\bet}_{k}\bm{l}' e^{-\hat{\lambda}\W}\bm{l}$.
\end{enumerate}
The ADI, AII and ATI are, respectively,  the average of the main diagonal elements of $\elwzn\bet_{0k}$, the average of the off-diagonal elements of $\elwzn\bet_{0k}$, and the average of all the elements of $\elwzn\bet_{0k}$. 

There are alternative ways that can be used to determine the dispersions of these scalar impact measures \citep{Arbia:2020}. In the Bayesian estimation approach, a sequence of random draws for each impact measure  can be obtained by using the posterior draws. Then, the mean and the standard deviation calculated from each sequence of impact measures can be used for inference.

In the classical estimation approaches, the delta method can be used to determine the asymptotic distributions of impact measures.     Applying the mean value theorem to ADI yields 
\begin{align*}
	\frac{1}{\sqrt{n}}\left(\tr(e^{-\hat{\lambda}\W}\hat{\bet}_{k})-\tr(e^{-\lambda_0\W}\bet_{0k})\right)
	&=\frac{1}{\sqrt{n}}\left(-\tr(e^{-\hat{\lambda}\W}\W\hat{\bet}_{k})(\hat{\lambda}-\lambda_0)+\tr(e^{-\hat{\lambda}\W})(\hat{\bet}_{k}-\bet_{0k})\right)+o_p(1)\nonumber\\
	&=\A_{1}\times \sqrt{n}(\hat{\lambda}-\lambda_0,\hat{\bet}_{k}-\bet_{0k})^{'}+o_p(1)\xrightarrow{d} N\Bigl(0,\lim_{n\rightarrow \infty}\A_{1}\B\A^{'}_{1}\Bigr), 
\end{align*}
where $\A_{1}=\left(-\frac{1}{n}\tr(e^{-\lambda_0 \W}\W \bet_{0k}),\frac{1}{n}\tr(e^{-\lambda_0 \W})\right)$ and $\B$ is the asymptotic covariance of $\sqrt{n}(\hat{\lambda}-\lambda_0,\hat{\bet}_{k}-\bet_{0k})$. Thus, we can estimate  the asymptotic variance of the direct impact as  $\frac{1}{n}\hat{\A}_{1}\hat{\B}\hat{\A}^{'}_{1}$, where $\hat{\A}_{1}=\left(-\frac{1}{n}\tr(e^{-\hat{\lambda}\W}\W\hat{\bet}_{k}), \frac{1}{n}\tr(e^{-\hat{\lambda}\W})\right)$, and $\hat{\B}$ is the estimated asymptotic covariance of $\sqrt{n}(\hat{\lambda}-\lambda_0,\hat{\bet}_{k}-\bet_{0k})$. Applying the mean value theorem to  ATI$=\frac{1}{n}\hat{\bet}_{k}\bm{l}' e^{-\hat{\lambda}\W}\bm{l}$, we obtain
\begin{align*}
	\frac{1}{\sqrt{n}}\left(\hat{\bet}_{k}\bm{l}' e^{-\hat{\lambda}\W}\bm{l} -\bet_{0k}\bm{l}' e^{-\lambda_0\W}\bm{l}\right)
	&=\A_{2}\times \sqrt{n}(\hat{\lambda}-\lambda_0,\hat{\bet}_{k}-\bet_{0k})^{'}+o_p(1)\xrightarrow{d}N\Bigl(0,\lim_{n\rightarrow\infty}\A_{2}\B\A^{'}_{2}\Bigr),
\end{align*}
where $\A_{2}=\left(-\frac{1}{n}\bet_{k}\bm{l}' e^{-\lambda_0\W}\W\bm{l},\frac{1}{n}\bm{l}' e^{-\lambda_0\W}\bm{l} \right)$. Thus, Var$(\frac{1}{n}\hat{\bet}_{k}\bm{l}' e^{-\hat{\lambda}\W}\bm{l} )$ can be estimated by $\frac{1}{n}\hat{\A}_{2}\hat{\B}\hat{\A}^{'}_{2}$, where $\hat{\A}_{2}=\bigl(-\frac{1}{n}\hat{\bet}_{k}\bm{l}' e^{-\hat{\lambda}\W}\W\bm{l},\frac{1}{n}\bm{l}' e^{-\hat{\lambda}\W}\bm{l}\bigr)$. Finally, applying the mean value theorem to the estimator of AII, we obtain
\begin{align*}
	&\frac{1}{\sqrt{n}}\left(\left(\hat{\bet}_{k}\bm{l}' e^{-\hat{\lambda}\W}\bm{l} -\tr(e^{-\hat{\lambda}\W}\hat{\bet}_{k})\right)-\left(\bet_{0k}\bm{l}' e^{-\hat{\lambda}_0\W}\bm{l} -\tr(e^{-\hat{\lambda}_0\W}\bet_{0k})\right)\right)\\
	&\quad\quad=(\A_{2}-\A_{1})\times \sqrt{n}(\hat{\lambda}-\lambda_0,\hat{\bet}_{k}-\bet_{0k})^{'}+o_p(1)\xrightarrow{d}N\Bigl(0,\lim_{n\rightarrow \infty}(\A_{2}-\A_{1})\B(\A_{2}-\A_{1})^{'}\Bigr).
\end{align*}
Then, an estimate of $\var\left(\frac{1}{n}\left(\hat{\bet}_{k}\bm{l}' e^{-\hat{\lambda}\W}\bm{l} -\tr(e^{-\hat{\lambda}\W}\hat{\bet}_{k})\right)\right)$ is given by $\frac{1}{n}(\hat{\A}_{2}-\hat{\A}_{1})\hat{\B}(\hat{\A}_{2}-\hat{\A}_{1})^{'}$.

\section{Model selection} \label{sec:selection}
Various approaches have been proposed in the literature to implement model selection. In this section we review the available approaches.

\subsection{Testing approach\label{sec:select:testing}}
The classical tests, such as the Wald, Lagrange multiplier (Rao score) and likelihood ratio tests, for inference on spatial parameters can be formulated by using the results on the asymptotic distributions of the estimators \citep{Anselin:1988, Anselin:1996,Anselin:2001, Lesage:2009, Elhorst:2014, Dogan:2018}. In the literature, to test non-nested hypotheses, the Cox statistic and the J statistic are adapted for mainly spatial autoregressive models \citep{Anselin:1984, Anselin:1986, Kelejian:2008, Kelejian:2011, Burridge:2012, Jin:2013}. These non-nested testing approaches can also be used for the model selection problem between the spatial autoregressive models and the MESS models.

In the J-test approach, we augment the null model with the predictor from the alternative model and then check whether the predictor can add significantly to the explanatory power of the augmented model \citep{Davidson:1981}. \citet{HAN2013250} consider the J-test for the model selection problem between the SARAR(1,0) and MESS (1,0) models. When the SARAR(1,0) model is the null model, we can formulate the null and the alternative hypotheses as
\begin{align*}
&H_0: \Y=\alpha \W\Y+\X\bet+\V,\\
&H_1:\Ss^{ex}(\lambda)\Y=\X\bet^{ex}+\V,
\end{align*}
where $\Ss^{ex}(\lambda)=\elw$ and $\bet^{ex}$ is a conformable parameter vector for $\X$ in the alternative model. As in \citet{Kelejian:2011}, \citet{HAN2013250} consider two predictors based on the alternative model. These predictors are $\hat{\Y}_{1}=\Ss^{ex}(\hat{\lambda})^{-1}\X\hat{\bet}^{ex}$ and $\hat{\Y}_{2}=(\I_n-\Ss^{ex}(\hat{\lambda}))\Y+\X\hat{\bet}^{ex}$,  where $\hat{\lambda}$ and $\hat{\bet}^{ex}$ are the QMLEs of $\lambda$ and $\bs{\beta}^{ex}$. Note that the first predictor is based on the reduced form of the alternative model while the second predictor is derived from the identity $\mf{Y}=(\I_n-\Ss^{ex}(\lambda))\Y+\X\bet^{ex}+\mf{V}$. Then, the null model can be augmented with these predictors to obtain the following testing equation:
\begin{align}\label{eq9.1}
 \Y=\alpha \W\Y+\X\bet+\hat{\Y}_{r_1}\delta_{r_1}+\V,
\end{align}
for $r_1=1,2$. Denote $\V(\eeta_{r_1})=(\I_n-\alpha \W)\Y-\X\bet-\hat{\Y}_{r_1}\delta_{r_1}$, where $\eeta_{r_1}=(\alpha,\bet^{'},\delta_{r_1})^{'}$. To estimate the augmented model, \citet{HAN2013250} consider a GMME based on the following vector of linear and quadratic moment functions: 
\begin{align*}
g(\eeta_{r_1})=(\V^{'}(\eeta_{r_1})\mathbf{P}_1\V(\eeta_{r_1}),\hdots,\V^{'}(\eeta_{r_1})\mathbf{P}_q\V(\eeta_{r_1}),\F^{'}\V(\eeta_{r_1})),
\end{align*}
where $\mf{F}$ is a full-column rank matrix of IVs and $\mf{P}_m$'s are $n\times n$ matrices of constants with $\tr(\mf{P}_m)=0$ for $m=1,\hdots,q$. Following \citet{KP:2010}, the IV matrix $\mf{F}$ can consist of the linearly independent columns of $\left(\mf{X}, \mf{W}\mf{X},\hdots, \mf{W}^d\mf{X}\right)$, where $d$ is a positive constant.  Let $\bs{\Xi}=\E[g(\eeta_{0r_1})g^{'}(\eeta_{0r_1})]$, where $\eeta_{0r_1}=(\alpha_0,\bet^{'}_0,0)^{'}$ is the true parameter vector under $H_0$. Then, using Lemma~\ref{l2}, it can be shown that
\begin{align*}
&\bs{\Xi}=
\begin{pmatrix}
(\mu_4-3\sigma_0^4)\oomega^{'}\oomega & \mu_3\oomega^{'}\F \\
\mu_3\F^{'}\oomega & 0
\end{pmatrix} +
\begin{pmatrix}
\tr(\mathbf{P}_1\mathbf{P}_1^{s}) &\hdots & \tr(\mathbf{P}_1\mathbf{P}_q^{s}) &0 \\
\vdots & \vdots &\vdots & \vdots \\
\tr(\mathbf{P}_q\mathbf{P}_1^{s}) &\hdots & \tr(\mathbf{P}_q\mathbf{P}_q^{s}) &0 \\
0 &\hdots & 0 &\frac{1}{\sigma_0^2}\F^{'}\F \\
\end{pmatrix},
\end{align*}
 where $\oomega=[\vec_D(\mathbf{P}_1),\hdots,\vec_D(\mathbf{P}_q)]$. Let $\frac{1}{n}\hat{\XXi}$ be a consistent estimator of $\frac{1}{n}\XXi$. Then, the feasible optimal GMME of $\eeta_{0r_1}$ is defined by $\hat{\eeta}_{r_1}=\argmin_{\eeta_{r_1}} g^{'}(\eeta_{r_1})\hat{\XXi}^{-1}g(\eeta_{r_1})$. Let  $\lambda_{sar}^{*}$ and $\bet^{ex*}_{sar}$ be the pseudo true values of $\lambda_0$ and $\bs{\beta}^{ex}_0$ under the null model, respectively. Define $\Ss^{ex*}_{sar}=e^{\lambda_{sar}^{*}\W}$, $\Ss=\I_n-\alpha_0\W$ and $\mathbf{G}=\W\Ss^{-1}$. Then, under some assumptions, \citet{HAN2013250}  show that 
\begin{align}\label{eq9.2}
\sqrt{n}(\hat{\eeta}_{r_1}-\eeta_{0r_1})\xrightarrow{d}N\left(\mathbf{0},\lim_{n\rightarrow\infty}\left(\D^{'}_{r_1}\XXi^{-1}\D_{r_1}\right)^{-1}\right),
\end{align}
for $r_1=1,2$, where
\begin{align*}
&\D_1=\begin{pmatrix}
\sigma_0^2\tr(\mathbf{P}_1^s\mathbf{G}) &0 & 0\\
\vdots & \vdots &\vdots \\
\sigma_0^2\tr(\mathbf{P}_q^s\mathbf{G}) &0 & 0\\
\F^{'}\mathbf{G}\X\bet_0 &\F^{'}\X& \F^{'}\Ss_{sar}^{ex*-1}\X\bet^{ex*}_{sar}
\end{pmatrix},\\
&\D_2=\begin{pmatrix}
\sigma_0^2\tr(\mathbf{P}_1^s\mathbf{G}) &0 & \sigma_0^2\tr(\mathbf{P}_1^{s}(\I_n-\Ss_{sar}^{ex*})\Ss^{-1})\\
\vdots & \vdots &\vdots \\
\sigma_0^2\tr(\mathbf{P}_q^s\mathbf{G}) &0 &  \sigma_0^2\tr(\mathbf{P}_q^{s}(\I_n-\Ss_{sar}^{ex*})\Ss^{-1})\\
\F^{'}\mathbf{G}\X\bet_0 &\F^{'}\X& \F^{'}((\I_n-\Ss_{sar}^{ex*})\Ss^{-1}\X\bet_0+\X\bet_{sar}^{ex*})
\end{pmatrix}.
\end{align*}
We summarize the estimation of $\eeta_{r_1}$  in Algorithm~\ref{alg4}. 
\begin{algorithm}[Estimation of the augmented model in \eqref{eq9.1}]\label{alg4}
\leavevmode  
\begin{enumerate}
\item Estimate the alternative model by the QMLE suggested in Section 3.1 and then compute the predictors $\hat{\mf{Y}}_{r_1}$ for $r_1=1,2$. 
\item Estimate the null model by one of the methods given in  Section 3. Use the estimated values to get a plug-in estimate of $\bs{\Xi}$.
\item Compute $\hat{\eeta}_{r_1}=\argmin_{\eeta_{r_1}}  g^{'}(\eeta_{r_1})\hat{\XXi}^{-1}g(\eeta_{r_1})$.
\end{enumerate}
\end{algorithm}
The result in \eqref{eq9.2} can be used to  construct the J statistic in three different ways: (i) the Wald (W) statistic, (ii) the distance difference
(DD) statistic, and  (iii) the gradient (G) statistic \citep{Newey:1987}. Let $\R=(\boldsymbol{0}_{1\times(k+1)},1)$ and $\hat{\D}_{r_1}$ be the plug-in estimator of $\D_{r_1}$ based on $\hat{\eeta}_{r_1}$ for $r_1=1,2$. Then, the first two statistics are given as
\begin{align}
&W_{r_1}=(\R\hat{\eeta}_{r_1})^{'}\left(\R\left(\hat{\D}_{r_1}^{'}\hat{\XXi}^{-1}\hat{\D}_{r_1}\right)^{-1}\R^{'}\right)^{-1}(\R\hat{\eeta}_{r_1}),\\
&\text{DD}_{r_1}=\min_{\{\eeta_{r_1}|\delta_{r_1}=0\}} g^{'}(\eeta_{r_1})\hat{\XXi}^{-1}g(\eeta_{r_1})-\min_{\eeta_{r_1}} g^{'}(\eeta_{r_1})\hat{\XXi}^{-1}g(\eeta_{r_1}).
\end{align}
Let $\tilde{\eeta}_{r_1}=\argmin_{\{\eeta_{r_1}|\delta_{r_1}=0\}} g^{'}(\eeta_{r_1})\hat{\XXi}^{-1}g(\eeta_{r_1})$ be the restricted optimal GMME.  Then, the gradient test statistic is defined by
\begin{align}
&\text{G}_{r_1}=g^{'}(\tilde{\eeta}_{r_1})\hat{\XXi}^{-1}\tilde{\D}_{r_1}(\tilde{\D}_{r_1}^{'}\hat{\XXi}^{-1}\tilde{\D}_{r_1})^{-1}\tilde{\D}_{r_1}\hat{\XXi}^{-1}g(\tilde{\eeta}_{r_1}),
\end{align}
where $\tilde{\D}_{r_1}$ is the plug-in estimator of $\D_{r_1}$ based on $\tilde{\eeta}_{r_1}$ for $r_1=1,2$. Under $H_0$, these statistics have a chi-squared distribution with one degree of freedom. Thus, we will reject  $H_0$ at the $5\%$ significance level if the  test statistics are larger than $3.84$.

When using the MESS model as the null model, the null and the alternative hypotheses take the following form:
\begin{align*}
&H_0:\Ss^{ex}(\lambda)\Y=\X\bet^{ex}+\V,\\
&H_1: \Y=\alpha \W\Y+\X\bet+\V.
\end{align*}
Let $\hat{\alpha}$ and $\hat{\bet}$ be the QML estimates of $\alpha_0$ and $\bs{\beta}_0$ from the alternative model. Again, we consider two predictors $\hat{\Y}_{1}=(\I_n-\hat{\alpha}\W)^{-1}\X\hat{\bet}$ and $\hat{\Y}_{2}=\hat{\alpha}\W\Y+\X\hat{\bet}$. Thus, the augmented model is given by
\begin{align}\label{eq9.6}
\Ss^{ex}(\lambda)\Y=\X\bet^{ex}+\hat{\Y}_{r_2}\delta_{r_2}+\V,
\end{align}
for $r_2=1,2$. Let $\bs{\psi}_{r_2}=(\lambda,\bet^{ex'},\delta_{r_2})^{'}$, $\bs{\psi}_{0r_2}=(\lambda_0,\bet^{ex'},0)^{'}$ for $r_2=1,2$, and $\alpha_{ex}^{*}$ and $\bet_{ex}^{*}$ be the pseudo true values of $\alpha$ and $\bs{\beta}$ under the null model.  \citet{HAN2013250} consider the non-linear 2SLS estimator (N2SLSE) for the estimation of the augmented model. Let $g(\bs{\psi}_{r_2})=\mf{F}^{'}\V(\bs{\psi}_{r_2})$ be the vector of linear moment functions, where $\V(\bs{\psi}_{r_2})=\Ss^{ex}(\lambda)\Y-\X\bet^{ex}+\hat{\Y}_{r_2}\delta_{r_2}$ for $r_2=1,2$. Then, the N2SLSE is defined by
\begin{align}
\hat{\bs{\psi}}_{r_2}=\argmin_{\bs{\psi}_{r_2}} \V^{'}(\bs{\psi}_{r_2})\F(\F^{'}\F)^{-1}\F^{'}\V(\bs{\psi}_{r_2}).
\end{align}
Under some assumptions, it can be shown that 
\begin{align}\label{eq9.8}
\sqrt{n}(\hat{\bs{\psi}}_{r_2}-\bs{\psi}_{0r_2})\xrightarrow{d}N\left(\mathbf{0},\sigma_0^{ex2}\left(\plim_{n\to\infty}\frac{1}{n}\D^{'}_{r_2}(\F^{'}\F)^{-1}\D_{r_2}\right)^{-1}\right),
\end{align}
where $\D_{1}=\F^{'}\left(\W\X\bet^{ex}_0,\X,\Ss_{ex}^{*-1}\X\bet^{*}_{ex}\right)$ and $\D_{2}=\F^{'}\left(\W\X\bet^{ex}_0,\X,\alpha_{ex}^{*}\W \Ss^{ex-1}\X\bet_{0}^{ex}+\X\bet^{*}_{ex}\right)$ with $\Ss_{ex}^{*}=\I_n-\alpha_{ex}^{*}\W$ and $\Ss^{ex}=e^{\lambda_0\W}$. Algorithm~\ref{alg5} summarizes the estimation of the augmented model in \eqref{eq9.6}.
\begin{algorithm}[Estimation of the augmented model in \eqref{eq9.6}]\label{alg5}
\leavevmode  
\begin{enumerate}
\item Estimate the alternative model by the QMLE and then compute the predictors $\hat{\mf{Y}}_{r_2}$ for $r_2=1,2$. 
\item Use $\mf{F}=\left(\mf{X}, \mf{W}\mf{X},\hdots, \mf{W}^d\mf{X}\right)$ to compute $\hat{\bs{\psi}}_{r_2}=\argmin_{\bs{\psi}_{r_2}} \V^{'}(\bs{\psi}_{r_2})\F(\F^{'}\F)^{-1}\F^{'}\V(\bs{\psi}_{r_2})$.
\end{enumerate}
\end{algorithm}
Similar to the previous case in which the SARAR(1,0) model  was the null model, we can use the result in \eqref{eq9.8} to derive the three test statistics.   When the disturbance terms are heteroskedastic, robust methods are necessary to derive consistent estimators.   However, the process to derive the three test statistics are similar to the homoskedastic case. 

Instead of using the critical value $3.84$ from the asymptotic distribution, we can use the bootstrap method to generate the empirical distribution of the test statistics. In this approach, we can report the bootstrapped p-value, which  is the
percentage of test statistics based on the bootstrapped samples that are greater
than the corresponding test statistic obtained from the actual sample, to decide between $H_0$ and $H_1$ \citep{Mac:2009}.  The bootstrap procedure for testing $H_0: \Y=\alpha \W\Y+\X\bet+\V$ against $H_1:\Ss^{ex}(\lambda)\Y=\X\bet^{ex}+\V$ is described in Algorithm~\ref{alg6}.
\begin{algorithm}[Bootstrap testing procedure]\label{alg6}
\leavevmode  
\begin{enumerate}
\item Compute $W_{r_1}$, $DD_{r_1}$ and $G_{r_1}$ for $r_1=1,2$.
\item Estimate the null model by one of the methods provided in Section 3. Let $\hat{\mf{V}}$ be the vector of residuals. 
\item Generate a random sample of size $n$ from  $\hat{\mf{V}}$ using sampling with replacement. Denote this re-sampled residual vector by $\hat{\mf{V}}^b$.
\item Use parameter estimates from Step 2 to compute
$
\mf{Y}^b=(\mf{I}_n-\hat{\lambda}\mf{W})^{-1}(\mf{X}\hat{\bs{\beta}}+\hat{\mf{V}}^b).
$  Compute the bootstrapped versions of test statistics $W^b_{r_1}$, $DD^b_{r_1}$ and $G^b_{r_1}$ for $r_1=1,2$ by using $\mf{Y}^b$.
\item Repeat  Steps 3--4 for $99$ times. Then, a test statistic rejects $H_0$ if the proportion of its bootstrapped versions that exceed the corresponding one computed in Step 1 is less than $5\%$.
\end{enumerate}
\end{algorithm}
In the heteroskedastic case, besides using the heteroskedasticity robust estimation methods, we also need to use a wild bootstrap approach to generate the bootstrapped versions of the test statistics. The details of this approach are summarized in \citet{HAN2013250}. The extensive simulation results reported in \citet{HAN2013250} indicate that all versions of the J-statistic can perform satisfactorily when the sample size is large.

\citet{LIU2019434} propose a non-degenerate Vuong-type model selection test for the model selection between the SARAR(1,1) and MESS(1,1) models.  The log-likelihood function of the SARAR(1,1) model in  \eqref{eq2.2} can be expressed as
 \begin{align*} 
\ln L_1(\thet_1)&=-\frac{n}{2}\ln(2\pi)-\frac{n}{2}\ln\sigma^2+\ln\left|\I_n-\alpha\W\right|+\ln\left|\I_n-\tau\W\right|-\frac{1}{2\sigma^2}\sum_{i=1}^nz_{i}(\thet_1)^2,
 \end{align*}
where $\thet_1=(\bet^{'},\alpha,\tau,\sigma^2)^{'}$ and $z_{i}(\thet_1)=y_{i}-\alpha \W_{i\cdot}\Y-\tau\M_{i\cdot}\Y+\alpha\tau\sum_{k=1}^{n}m_{ik}\W_{k\cdot}\Y-\X_{i}\bet+\tau\M_{i\cdot}\X\bet$, with $\X_{i}$ being the $i$th row of $\X$, $m_{ik}$ being the $(i,k)$th element of $\M$, and  $\W_{i\cdot}$ and  $\M_{i\cdot}$ being the $i$th row of $\W$ and $\M$, respectively. Then, we can write the log-likelihood function as $\ln L_1(\thet_1)=\sum_{i=1}^nl_{1i}(\thet_1)$, where $l_{1i}(\thet_1)=-\frac{1}{2}\ln(2\pi)-\frac{1}{2}\ln\sigma^2+\frac{1}{n}\ln\left|\I_n-\alpha\W\right|+\frac{1}{n}\ln\left|\I_n-\tau\W\right|-\frac{1}{2\sigma^2}z_{i}(\thet_1)^2$.  Similarly, we can express the log-likelihood function of the MESS(1,1)  as 
 \begin{align*} 
\ln L_2(\thet_2)=\sum_{i=1}^nl_{2i}(\thet_2),
\end{align*}
where $\thet_2=(\bet^{'},\lambda,\rho,\sigma^2)^{'}$, $l_{2i}(\thet_2)=-\frac{n}{2}\ln(2\pi)-\frac{n}{2}\ln\sigma^2-\frac{1}{2\sigma^2}h_{i}(\thet_2)^2$ and  $h_{i}(\thet_2)$ is the $i$th element of $\erm(\elw\Y-\X\bet)$. \citet{LIU2019434} first show that the QMLEs of both models, where one of the models or both models are possibly misspecified, are consistent estimators of their pseudo-true values and are asymptotically normal. 

\citet{LIU2019434}  assume that the true data generating process is unknown, and one of the two models or both models might be misspecified. Let $\thet^{*}_1$ and $\thet^{*}_2$ be the pseudo-true parameter vectors in the SARAR(1,1) and MESS(1,1) models, respectively. Then, the null hypothesis and alternative hypotheses are given by 
\begin{align*}
&H_0:\lim_{n\rightarrow\infty}\frac{1}{\sqrt{n}}\E\left[\ln\frac{L_1(\thet_1^{*})}{L_2(\thet_2^{*})}\right]=0 \quad \text{(Models 1 and 2 are asymptotically equivalent)},\\
&H_1:\lim_{n\rightarrow\infty}\frac{1}{\sqrt{n}}\E\left[\ln\frac{L_1(\thet_1^{*})}{L_2(\thet_2^{*})}\right]=+\infty \quad \text{(Model 1  is asymptotically better than model 2)},\\
&H_2:\lim_{n\rightarrow\infty}\frac{1}{\sqrt{n}}\E\left[\ln\frac{L_1(\thet_1^{*})}{L_2(\thet_2^{*})}\right]=-\infty \quad \text{(Model 1  is asymptotically worse than model 2)}.
\end{align*}
Let $\text{LR}(\hat{\thet}_1,\hat{\thet}_2)=\ln L_{1}(\hat{\thet})-\ln L_{2}(\hat{\thet})$, where $\hat{\thet}_1$ and  $\hat{\thet}_2$ are the QMLEs of the two models. Define $\omega^2=\var(\frac{1}{\sqrt{n}}\text{LR}(\hat{\thet}_1,\hat{\thet}_2))$ and $g_i(\thet_1,\thet_2)=l_{1i}(\thet_1)-l_{2i}(\thet_2)$. Then, following \citet{Shi:2017}, \citet{LIU2019434} consider the following test statistic:
 \begin{align*} 
\hat{T}=\frac{\frac{1}{\sqrt{n}}\sum_{i=1}^ng_i(\hat{\thet}_1,\hat{\thet}_2)+\hat{\sigma}U}{\sqrt{\hat{\omega}^2+\hat{\sigma}^2}},
 \end{align*}
where $\hat{\sigma}$ is a data-dependent scalar, $\hat{\omega}^2$ is an estimator of $\omega^2$ and $U\sim N(0,1)$. Under some regularity assumptions, it is shown that the test statistic converges to the standard normal distribution under the null hypothesis, i.e., $\hat{T}\xrightarrow{d}N(0,1)$ under $H_0$. Under the alternative hypotheses, they show that $\hat{T}\rightarrow+\infty$ under $H_1$, and $\hat{T}\rightarrow-\infty$ under $H_2$. In a Monte Carlo study,  \citet{LIU2019434} show that the test statistic has good size and power properties.

\subsection{Information criteria approach\label{sec:ic}}
 The predictive accuracy of a model is usually measured through an information criterion, which is typically defined based on the deviance term $-2\ln p(\mf{Y}|\bs{\theta})$ \citep{Gelman:2003}.  The widely used Akaike information criterion (AIC) takes the following form:
\begin{align}
\text{AIC}=-2\ln p(\mf{Y}|\hat{\bs{\theta}})+2p,
\end{align}
where $\hat{\bs{\theta}}$ is an estimate of $\bs{\theta}$ and $p$ is the dimension of $\bs{\theta}$. In a Bayesian context, \citet{Spiegelhalter:2002} suggest another criterion called the deviance information criterion (DIC):
\begin{align*}
&\text{DIC}=\bar{D}(\thet)+p_D,
\end{align*}
where $\bar{D}(\thet)$ is called the posterior mean deviance and $p_D$ is a measure of the effective number of parameters in the model. The posterior mean deviance is defined by $\bar{D}(\thet)=-2\E\left(\ln p(\Y|\thet)|\Y \right)$, where the expectation is taken with respect to the posterior distribution of $\bs{\theta}$. This term serves as a Bayesian measure of model fit. The effective number of parameters is defined by $p_D =\bar{D}(\thet)-D(\bar{\thet})=-2\E\left(\ln p(\Y|\thet)|\Y \right)+2\ln p(\Y|\bar{\thet})$, where $\bar{\bs{\theta}}$ is the posterior mean of $\bs{\theta}$. Thus, the DIC can be written as 
\begin{align*}
&\text{DIC}=-4\E\left(\ln p(\Y|\thet)|\Y \right)+2\ln p(\Y|\bar{\thet}).
\end{align*}
Let $\{\bs{\theta}^r\}_{r=1}^R$ be a sequence of posterior draws. Then, the first term $\E\left(\ln p(Y|\theta)|Y \right)$ in the $\text{DIC}$ can be computed by $\E\left(\ln p(\mf{Y}|\bs{\theta})|\mf{Y} \right)\approx\frac{1}{R}\sum_{r=1}^R\ln p(\mf{Y}|\bs{\theta}^r)$. The second term $\ln p(\mf{Y}|\bar{\bs{\theta}})$ in the DIC is computed by evaluating the log-likelihood function at the posterior mean $\bar{\bs{\theta}}$. Using a decision-theoretic perspective, it can be shown that both AIC and DIC choose the model whose predictive distribution is close to the true data generating process \citep{Li:2020}. 

In our heteroskedastic model, there are alternative likelihood functions: (i) the conditional likelihood function denoted by $p(\mf{Y}|\bs{\theta},\bs{\eta})$, (ii) the complete-data likelihood function denoted by  $p(\mf{Y},\bs{\eta}|\bs{\theta})$, and (iii) the integrated (or observed) likelihood function denoted by $p(\mf{Y}|\bs{\theta})=\int p(\mf{Y},\bs{\eta}|\bs{\theta})\text{d}\bs{\eta}$. The log-conditional likelihood function is readily available and given by
\begin{align}
\ln p(\mf{Y}|\bs{\theta},\bs{\eta})&=-\frac{n}{2}\ln(2\pi)-\frac{n}{2}\ln\sigma^2-\frac{1}{2}\sum_{i=1}^n\ln\eta_{i}\\
&\quad-\frac{1}{2\sigma^2}(e^{\lambda \mf{W}}\mf{Y}-\mf{X}\bs{\beta})^{'}e^{\rho \mf{M}^{'}}\mf{H}^{-1}(\bs{\eta})e^{\rho\mf{M}}(e^{\lambda\mf{W}}\mf{Y}-\mf{X}\bs{\beta}),\nonumber
\end{align}
where $\mf{H}(\bs{\eta})=\text{Diag}\left(\eta_{1},\hdots,\eta_{n}\right)$ is the $n\times n$ diagonal matrix with the $i$th diagonal element $\eta_{i}$.  As shown in Section 6.2, this function facilitates the Bayesian estimation of the heteroskedastic model. Both the conditional likelihood function and the complete-data likelihood function depend on the high-dimensional latent scale mixture variables. Since these high-dimensional variables can not be estimated precisely, the AIC and DIC formulated with the conditional and complete-data likelihood functions may not perform satisfactorily in model selection exercises \citep{Chan:2016}. Indeed, the latent variable models violate the conditions of the decision-theoretic perspective, indicating that the AIC and DIC cannot be used as a measure of predictive accuracy \citep{Li:2020}.  Hopefully, the log-integrated likelihood function can be obtained analytically by integrating out the scale mixture variables $\bs{\eta}$ from the complete-data likelihood function, i.e., $p(\mf{Y}|\bs{\theta})=\int p(\mf{Y},\bs{\eta}|\bs{\theta})\text{d}\bs{\eta}=\int p(\mf{Y}|\bs{\eta},\bs{\theta})p(\bs{\eta}|\bs{\theta})\text{d}\bs{\eta}$.  This function can be derived as \citep{Dogan:2023}
\begin{align*}
\ln p(\mf{Y}|\bs{\theta})&=-\frac{n}{2}\ln(2\pi)-\frac{n}{2}\ln\sigma^2+\frac{n\nu}{2}\ln(\nu/2)\\
&\quad+n\ln\Gamma\left(\frac{\nu+1}{2}\right)-n\ln\Gamma(\nu/2)-\frac{\nu+1}{2}\sum_{i=1}^n\ln\left(\frac{\nu}{2}+\frac{y^2_{i}(\bs{\delta})}{2\sigma^2}\right),
\end{align*} 
where $y_{i}(\bs{\delta})$ is the $i$th element of $\mf{Y}(\bs{\delta})=e^{\rho\mf{M}}\left(e^{\lambda \mf{W}}\mf{Y}-\mf{X}\bs{\beta}\right)$ with $\bs{\delta}=(\lambda,\rho,\bs{\beta}^{'})^{'}$. This function can be used to formulate AIC and DIC in the heteroskedastic model.

Another popular criterion is the Bayesian information criterion, which can be derived from a large sample approximation to the log-marginal likelihood of a candidate model. Let $\{M_k\}_{k=1}^K$ be a sequence of candidate models. Then, the marginal likelihood of the model $M_k$ can be expressed as $p(\mf{Y}_k|M_k)=\int_{\bs{\Theta}_k}p(\mf{Y}|\bs{\theta}_k,M_k)p(\bs{\theta}_k|M_k)\text{d}\bs{\theta}_k$, where $\bs{\theta}_k$ is the $p_k\times1$ vector of parameters in $M_k$. Then, the Laplace approximation to $\ln p(\mf{Y}_k|M_k)$ yields the following BIC measure \citep{Schwarz:1978}:
\begin{align}
\text{BIC}_k=-2\ln p(\mf{Y}|\hat{\bs{\theta}})+2p\ln(n).
\end{align}
The Laplace approximation to $\ln p(\mf{Y}_k|M_k)$ can also be used to show that \citep{Kass:1995}
\begin{align}\label{9.4}
\lim_{n\to\infty}P\left(\left|\frac{\text{BIC}_k-\text{BIC}_l}{\ln\text{BF}_{kl}}-1\right|>\epsilon\right)= 0,
\end{align}
where $\epsilon>0$ is an arbitrary number and $\text{BF}_{kl}=p(\mf{Y}|M_k)/p(\mf{Y}|M_l)$ is the Bayes factor of $M_k$ against $M_l$. The result in \eqref{9.4} indicates that the BIC is also a consistent model selection criterion like the Bayes factor. Moreover, both BIC and the Bayes factor can be interpreted as the measures of predictive accuracy because the marginal likelihood function can be interpreted as the predictive density evaluated at $\mf{Y}$ \citep{Chan:2016}.

In the Bayesian setting described in Section 6.2, \citet{Dogan:2023} investigate the performance of AIC, DIC and BIC for both nested and non-nested model selection problems through simulations. They consider four popular MESS specifications and aim to see whether the information criteria can select correct model specification and the correct spatial weights matrix from a pool of candidates. Their extensive simulation results show that these criteria perform satisfactorily and can be useful for selecting the correct model in the specification search exercises.

\citet{yangms} suggest using a Mallows $C_p$ type selection criterion for selecting  a spatial weights matrix from a pool of candidates. Let $\mathcal{W}=\left\{\left(\mf{W}_s, \mf{M}_s\right): s \in \{1, 2, \hdots, S\}\right\}$ be the pool of spatial weights matrices. The quasi log-likelihood function based on the tuple $(W_{s}, M_{s})$ can be expressed as
\begin{align}\label{9.5}
\ell_s&=-\frac{n}{2}\ln2\pi\sigma^2-\frac{1}{2\sigma^2}\left\Vert e^{\tau \mf{M}_s}(e^{\alpha \mf{W}_s}\mf{Y}-\mf{X}\bs{\beta})\right\Vert^2,
\end{align}
where $\Vert\cdot\Vert$ denotes the Euclidean norm. For  a given $(\hat{\alpha}_s,\hat{\tau}_s)$ value, the first order conditions of \eqref{9.5} with respect to $\bs{\beta}$ and $\sigma^2$ yield
 \begin{align}
&\hat{\bs{\beta}}_s=\left(\mf{X}^{'}e^{\hat{\tau}_s\mf{M}^{'}_s}e^{\hat{\tau}_s\mf{M}_s}\mf{X}\right)^{-1}\mf{X}^{'}e^{\hat{\tau}_s\mf{M}_s^{'}}e^{\hat{\tau}_s\mf{M}_s}e^{\hat{\alpha}_s\mf{W}_s}\mf{Y}, \label{9.6}\\
&\hat{\sigma}^2_s=\frac{1}{n}\left\Vert e^{\hat{\tau}_s\mf{M}_s}(e^{\hat{\alpha}_s\mf{W}_s}\mf{Y}-\mf{X}\hat{\bs{\beta}}_s)\right\Vert^2.
\end{align}
Let $\bs{\mu}=\E(\mf{Y})=e^{-\alpha_0 \mf{W}}\mf{X}\bs{\beta}_0$. Substituting \eqref{9.6} into $\hat{\bs{\mu}}_s=e^{-\hat{\alpha}_s\mf{W}_s}\mf{X}\hat{\bs{\beta}}_s$, we obtain
\begin{align} \label{2.8}
\hat{\bs{\mu}}_s&=e^{-\hat{\alpha}_s\mf{W}_s}\mf{X}\left(\mf{X}^{'}e^{\hat{\tau}_s \mf{M}_s^{'}}e^{\hat{\tau}_s\mf{M}_s}\mf{X}\right)^{-1}\mf{X}^{'}e^{\hat{\tau}_s\mf{M}_s^{'}}e^{\hat{\tau}_s\mf{M}_s}e^{\hat{\alpha}_s\mf{W}_s}\mf{Y}=\widetilde{\mf{P}}_s\mf{Y},
\end{align}
where $\widetilde{\mf{P}}_s=e^{-\hat{\alpha}_s\mf{W}_s}e^{-\hat{\tau}_s\mf{M}_s}\widehat{\mf{P}}_se^{\hat{\tau}_s \mf{M}_s}e^{\hat{\alpha}_s\mf{W}_s}$ with $\widehat{\mf{P}}_s=e^{\hat{\tau}_s \mf{M}_s}\mf{X}\left(\mf{X}^{'}e^{\hat{\tau}_s \mf{M}_s^{'}}e^{\hat{\tau}_s \mf{M}_s}\mf{X}\right)^{-1}\mf{X}^{'}e^{\hat{\tau}_s\mf{M}_s^{'}}$. Then, \citet{yangms} consider the following selection criterion function:
\begin{align*}
C_s=\left\Vert\widetilde{\p}_s\Y-\Y\right\Vert^2+2\left(\tr(\widetilde{\p}_s\Ome)+\frac{\partial \hat{\lambda}_s}{\partial \Y^{'}}\Ome\frac{\partial \widetilde{\p}_s}{\partial \hat{\lambda}}\Y+\frac{\partial \hat{\rho}_s}{\partial \Y^{'}}\Ome\frac{\partial \widetilde{\p}_s}{\partial \hat{\rho}}\Y\right),
\end{align*}
where $\Ome=\sigma_0^2 e^{-\lambda_0 \W}e^{-\rho_0 \M}e^{-\rho_0 \M^{'}}e^{-\lambda_0 \W^{'}}$ is the variance of $\Y$, and the closed forms of $\frac{\partial\hat{\lambda}_s}{\partial \Y^{'}}$, $\frac{\partial \widetilde{\p}_s}{\partial \hat{\lambda}}$ and $\frac{\partial \hat{\rho}_s}{\partial \Y^{'}}$ can be found in \citet{yangms}. Given an estimator of $\Ome$, we can compute $C_s$ for each $s$. Thus, the selected model is defined  as $\hat{s}=\argmin_{s\in\{1, \hdots, S\}}C_s$. Under certain assumptions,  \citet{yangms} show that the selection estimator $\hat{\bs{\mu}}_{\hat{s}}$ is asymptotically optimal in the sense that it is as efficient as the infeasible estimator that uses the best candidate spatial weights matrix. They also show that the selection procedure is selection consistent in the sense that it chooses the true tuple of weight matrices with probability approaching one as $n\rightarrow\infty$.

Instead of selecting the asymptotically optimal model, it is also possible to use a model averaging scheme that compromises across a set of candidate models.  Let $\mf{w}=(w_1, \hdots, w_S)^{'}$ be a vector of weights, and $\mathcal{N}=\left\{\mf{w}\in[0,1]^S: \sum_{s=1}^Sw_s=1\right\}$ be the set of model weights vectors. Let $\widetilde{\mf{P}}(\mf{w})=\sum_{s=1}^Sw_s\widetilde{\mf{P}}_s$ be the weighted average of $\left\{\widetilde{\mf{P}}_1, \hdots, \widetilde{\mf{P}}_S\right\}$. Then, the model averaging estimator for $\bs{\mu}$ is given by 
\begin{equation}\label{4.1}
\hat{\bs{\mu}}(\mf{w})=\sum_{s=1}^Sw_s\hat{\bs{\mu}}_s=\sum_{s=1}^S w_s\widetilde{\mf{P}}_s
\mf{Y}=\widetilde{\mf{P}}(\mf{w})\mf{Y}.
\end{equation} 
Then, \citet{yangms} consider the following model weights choice criterion function:
\begin{equation*}
C(\mf{w})=\left\Vert \widetilde{\p}(\w)\Y-\Y\right\Vert^2+2\left(\tr\left(\widetilde{\p}(\w)\Ome\right)+\sum_{s=1}^Sw_s\left(\frac{\partial \hat{\lambda}_s}{\partial \Y^{'}}\Ome \frac{\partial \widetilde{\p}_s}{\partial \hat{\lambda}_s}\Y+\frac{\partial \hat{\rho}_s}{\partial \Y^{'}}\Ome \frac{\partial \widetilde{\p}_s}{\partial \hat{\rho}_s}\Y\right)\right).
\end{equation*}
The optimal model weights vector is thus given by $\hat{\mf{w}}=\argmin_{\mf{w}\in\mathcal{N}}\widehat{C}(\mf{w})$. Similar to the model selection procedure, the model averaging estimator $\hat{\bs{\mu}}(\hat{\mf{w}})$ is also asymptotically optimal.

The selection and averaging estimators can also be  considered for the high order MESS models. In the case of heteroskedastic models, \citet{yangms} use a heteroskedasticity robust GMM estimator to formulate the selection and model averaging criterion functions. The extensive simulation results in \citet{yangms} indicate that the model selection and averaging estimators perform satisfactorily.

\subsection{Marginal likelihood approach}
In the Bayesian approach, the Bayes factor can be used for both nested and non-nested model selection problems. As shown in Section \ref{sec:ic}, the Bayes factor for two models is simply the ratio of the corresponding marginal likelihood functions: $\text{BF}_{kl}=p(\mf{Y}|M_k)/p(\mf{Y}|M_l)$, where $p(\mf{Y}|M_j)=\int_{\bs{\Theta}_j}p(\mf{Y}|\bs{\theta}_j,M_j)p(\bs{\theta}_j|M_j)\text{d}\bs{\theta}_j$ for $j\in\{k,l\}$. Thus, the Bayes factor chooses $M_k$ if $p(\mf{Y}|M_k)$ is larger than $p(\mf{Y}|M_l)$. If the data is generated from $M_k$, then the Bayes factor will consistently choose $M_k$ over $M_l$. To see this, consider the expectation of the log-Bayes factor under $p(\mf{Y}|M_k)$:
\begin{align}
\E\left(\log\frac{p(\mf{Y}|M_k)}{p(\mf{Y}|M_l)}\right)=\int \log\frac{p(\mf{Y}|M_k)}{p(\mf{Y}|M_l)} p(\mf{Y}|M_k)\text{d}\mf{Y},
\end{align}
which is simply the Kullback-Leibler divergence between $p(\mf{Y}|M_k)$ and $p(\mf{Y}|M_l)$. Thus, the expectation is strictly positive, unless $p(\mf{Y}|M_k)=p(\mf{Y}|M_l)$ in which case it is zero.

The Bayes factor reduces to the Savage-Dickey density ratio (SDDR) for the nested model selection problems \citep{VW:1995}. For example, consider the following null and alternative hypotheses: $H_0:\lambda=0$ against $H_1:\lambda\ne0$ or $H_0:\rho=0$ against $H_1:\rho\ne0$. Let $M_R$ and $M_U$ be respectively the restricted and the unrestricted model. Then, the Bayes factor in favor of the unrestricted model is 
\begin{align}\label{9.11}
\text{BF}_{UR}=\frac{p(\mf{Y}|M_U)}{p(\mf{Y}|M_R)},
\end{align}
where $p(\mf{Y}|M_j)$ for $j\in\{U, R\}$ is the corresponding marginal likelihood function.  Since our prior distributions are independent, the Bayes factor in \eqref{9.11} reduces to the SDDR given by
\begin{align}
\text{BF}_{UR}=\frac{p(\lambda=0|M_U)}{p(\lambda=0|\mf{Y},M_U)},
\end{align}
where $p(\lambda=0|M_U)$ and $p(\lambda=0|\mf{Y},M_U)$ are respectively the prior and the marginal posterior densities of $\lambda$ evaluated at $\lambda=0$. The $\text{BF}_{UR}$ indicates that if $\lambda=0$ is more likely under the prior relative to the marginal posterior, then the $\text{BF}_{UR}$ provides evidence in favor of $H_1$. Under the prior $\lambda\sim N(\mu_\rho,V_\rho)$, we have $p(\lambda=0|M_U)=(2\pi V_\lambda)^{-1/2}\exp(-\mu^2_\lambda/2V_{\lambda})$. Let $\{\bs{\beta}^r,\lambda^r, \rho^r,\sigma^{2r}\}_{r=1}^R$ be a sequence of posterior draws. Then, one way to estimate the marginal posterior $p(\lambda=0|\mf{Y},M_U)$ is to use the following Rao-Blackwell estimator \citep{Smith:1990}:
\begin{align}
\hat{p}(\lambda=0|\mf{Y},M_U)=\frac{1}{R}\sum_{r=1}^Rp(\lambda=0|\mf{Y}, \bs{\beta}^r,\rho^r,\sigma^{2r}).
\end{align}
This Rao-Blackwell estimator cannot be used in our case because the conditional posterior density of $\lambda$ does not take a standard form.  If we assume that the parameter space of $\lambda$ is contained in the interval $(-\tau,\,\tau)$, where $\tau$ is a finite positive constant, then we may resort to a Griddy-Gibbs sampler to estimate $p(\lambda=0|\mf{Y},M_U)$. Algorithm~\ref{alg7} describes how we can use this approach.
\begin{algorithm}[Computing SDDR]\label{alg7}
\leavevmode  
\begin{enumerate}
\item Construct a grid with random points $\lambda_1,\hdots,\lambda_m$ from the interval $(-\tau,\,\tau)$. The grid must also include $\lambda_k=0$.
\item Compute $p_r(\lambda_i)=\frac{p(\lambda_i|\mf{Y}, \bs{\beta}^r,\rho^r,\sigma^{2r})}{\sum_{j=1}^m p(\lambda_j|\mf{Y}, \bs{\beta}^r,\rho^r,\sigma^{2r})}$ for $i=1,\hdots,m$ and $r=1,\hdots, R$.
\item Compute $p(\lambda_i)=\sum_{r=1}^Rp_r(\lambda_i)$ for $i=1,\hdots,m$.
\item Return $\hat{p}(\lambda=0|\mf{Y},M_U)=p(\lambda_k)$.
\end{enumerate}
\end{algorithm}

The marginal likelihood function of the MESS type models does not take a closed form.  There are alternative methods that can be used to estimate or to approximate the marginal likelihood function. In the homoskedastic case, we can analytically integrate out $\bs{\beta}$ and $\sigma^2$ under the following priors: (i) $\bs{\beta}|\sigma^2N(\bs{\mu}_{\beta},\sigma^2\mf{V}_{\beta})$ and $\sigma^2\sim IG(a,b)$ or (ii) $p(\bs{\beta},\sigma^2)\propto1/\sigma$. However, in order to get the marginal likelihood function, we also need to integrate out the spatial parameters, which is not possible analytically. Then, one approach for computing the marginal likelihood function can be based on a numerical integration method \citep{Lesage:2009, Han:2013, Hepple:1995}. In the case of MESS(1,1), this approach requires a double numerical integration over the parameter space of $\lambda$ and $\rho$. It is clear that this approach may not be feasible for high order MESS models and for models with heteroskedasticity. 

Alternatively, since the conditional posterior distributions of the spatial parameters are in non-standard forms, we may resort to the method suggested by \citet{Chib:2001} to estimate the marginal likelihood function.  This approach is general enough and only requires the MCMC draws of parameters. In the heteroskedastic case, this approach requires the MCMC draws of the high-dimensional scale mixture variables, therefore it may not produce precise estimates.  

The modified harmonic mean method of \citet{Dey:1994} can also be used to estimate the marginal likelihood function.  This method requires a probability density function $g$ whose support lies in the support of the posterior distribution. The method produces an approximation based on $\E\left(\frac{g(\bs{\theta})}{p(\mf{Y}|\bs{\theta})p(\bs{\theta})}\big|\mf{Y}\right)$, where the expectation is taken with respect to $p(\bs{\theta}|\mf{Y})$. The expectation gives  the following relationship:
\begin{align}
\E\left(\frac{g(\bs{\theta})}{p(\mf{Y}|\bs{\theta})p(\bs{\theta})}\big|\mf{Y}\right)&=\int\frac{g(\bs{\theta})}{p(\mf{Y}|\bs{\theta})p(\bs{\theta})}p(\bs{\theta}|\mf{Y})\text{d}\bs{\theta}=\int\frac{g(\bs{\theta})}{p(\mf{Y}|\bs{\theta})p(\bs{\theta})}\frac{p(\mf{Y}|\bs{\theta})p(\bs{\theta})}{p(\mf{Y})}\text{d}\bs{\theta}\nonumber\\
&=p^{-1}(\mf{Y})\int g(\bs{\theta})\text{d}\bs{\theta}=p^{-1}(\mf{Y}).
\end{align}
Thus, the marginal likelihood function $p(Y)$ can be estimated by the following estimator:
 \begin{align}
 \hat{p}(\mf{Y})=\left(\frac{1}{R}\sum_{r=1}^R\frac{g(\bs{\theta}^{r})}{p(\mf{Y}|\bs{\theta}^{r})p(\bs{\theta}^{r})}\right)^{-1},
 \end{align}
 where $\{\bs{\theta}^r\}_{r=1}^R$ is a sequence of the posterior draws from $p(\bs{\theta}|\mf{Y})$. Under the  condition that $g(\bs{\theta})/\left(p(\mf{Y}|\bs{\theta})p(\bs{\theta})\right)$ is bounded above over the support of the posterior distribution, it can be shown that this estimator is a simulation consistent estimator when $R$ goes to infinity \citep{Geweke:1999}. To guarantee this boundedness condition, following \citet{Geweke:1999}, we can consider a truncated multivariate normal density for $g$.  Let $ A=\{\bs{\theta}\in\mathbb{R}^p:(\bs{\theta}-\hat{\bs{\theta}})^{'}\hat{\bs{\Omega}}^{-1}(\bs{\theta}-\hat{\bs{\theta}})<\chi^2_{\alpha,p}\}$ be the truncation set, where $\hat{\bs{\theta}}$ is the posterior mean of $\bs{\theta}$, $\hat{\bs{\Omega}}$ is the posterior covariance of $\bs{\theta}$, and $\chi^2_{\alpha,p}$ is the $(1-\alpha)$ quantile of the $\chi^2_p$ distribution. Then, $g$ takes the following form: 
 \begin{align}
 g(\bs{\theta})=(1-\alpha)^{-1}(2\pi)^{-p/2}\left|\hat{\bs{\Omega}}\right|^{-1/2}\exp\left(-\frac{1}{2}(\bs{\theta}-\hat{\bs{\theta}})^{'}\hat{\bs{\Omega}}^{-1}(\bs{\theta}-\hat{\bs{\theta}})\right)\times\mathbf{1}_A(\bs{\theta}),
 \end{align}
 where $\mf{1}_A(\bs{\theta})$ is the indicator function taking value $1$ if $\bs{\theta}\in A$, otherwise $0$. 
 
Note that the computation of the modified harmonic mean estimator requires the integrated likelihood function which is available for both homoskedastic and heteroskedastic models.  In the context of spatial autoregressive models, \citet{Dogan:2023b} investigates the finite sample performance of this estimator along with some other popular information criteria for both nested and non-nested model selection problems. His simulation results show that the modified harmonic mean estimator performs satisfactorily, and can be useful for the specification search exercises in spatial econometrics.

\section{A Monte Carlo study} \label{sec:mc}

\subsection{Design}\label{design}

In this section, we conduct a Monte Carlo study to  investigate the finite sample properties of the estimators considered in Sections \ref{sec:mle} through \ref{sec:Bayesian}. To this end, we consider the following data generating process:
\begin{align*}
 \Y=e^{-\lambda_0 \W}\X\bs{\beta}_0+e^{-\lambda_0 \W}e^{-\rho_0 \M}\V,
\end{align*}
where  $\X$ contains two explanatory variables with $\bs{\beta}_0 =(\beta_{10},\beta_{20})'=(1,\,1)^{'}$. The observations for the first explanatory variable are drawn independently from the standard normal distribution, while the observations for the second explanatory variable are drawn independently from Uniform($0, \sqrt{12}$). The spatial parameters $\lambda_0$ and $\rho_0$ can take values from the set $\{(-2, -1), (-2, 1), (0.5, -1), (0.5, 1)\}$. For $\mathbf{V}$, in the homoskedastic scenario, the elements $v_i$ are i.i.d.\ draws from either (i) the standard normal distribution or (ii) the standardized chi-squared distribution with three degrees of freedom, i.e., $(\chi^2_3-3)/\sqrt{6}$. In the heteroskedastic scenario, we set  $\mathbf{V}\sim N\left(\mf{0},\text{Diag}(\gamma_1,\hdots,\gamma_n)\right)$ with either (i) $\gamma_{i}=2 \vartheta_i / (\sum_{j=1}^n \vartheta_j/n)$, where $\vartheta_i$ is the number of neighbors for unit $i$ using the description of $\W_1$ below, or (ii) $\gamma_{i}=\exp(0.1+0.35X_{2i})$ and $X_{2i}$ is the $i$th element of $X_2$.

For the spatial weights matrix $\W$, we consider the interaction scenario described in \citet{Arraiz:2010}. To this end, let $n$ entities be distributed across four quadrants of a square grid in such a way that the number of entities in each quadrant can be arranged to allow for sparse or dense quadrants. The location of each entity across the grid is determined by the $xy$-coordinates on the grid. Let $\underline{c}$ and $\overline{c}$ be two integers. The entities in the northeast quadrant of the grid have discrete coordinates satisfying $(\underline{c}+1)\leq x\leq \overline{c}$ and $(\underline{c}+1)\leq y\leq \overline{c}$, with an increment value of $0.5$. For the other quadrants, the location coordinates are integers satisfying $1\leq x\leq \underline{c}$, $1\leq y\leq \overline{c}$, and $1\leq x\leq \overline{c}$, $1\leq y\leq \underline{c}$. The distance $d_{ij}$ between any two entities  $i$ and $j$, located respectively at $(x_1,y_1)$ and $(x_2,y_2)$, is measured by the Euclidean distance given by $d_{ij}=\left[(x_1-x_2)^2+(y_1-y_2)^2\right]^{1/2}$. Then, the $(i,j)$th element of $\W$ is set to $1$ if $0\leq d_{ij}\leq 1$, and to $0$ otherwise. We then row normalize $\W$. In this scenario, varying the values for $\underline{c}$ and $\overline{c}$ leads to a different sample size and a different share of units in the northeast quadrant. We consider the following two combinations: $(\underline{c},\overline{c})=(5, 15)$ and $(\underline{c},\overline{c})=(14, 20)$. The first combination produces a sample size of $486$ and locates $75$ percent of the entities in the northeast quadrant ($\mathbf{W}_1$), whereas the second combination generates a sample size of $485$ and locates $25$ percent of the entities in the northeast quadrant ($\mathbf{W}_2$). 

For the spatial weights matrix $\M$, we consider a nearest neighbors scheme. To this end, using the Euclidean distances ($d_{ij}$'s) from the construction of $\W$ above, we let entity $i$ be dependent on its 5 nearest neighbors so that the weights corresponding to these neighbors in the $i$'th row of $\M$ are set to $1$ and the rest are set to zero. Then, $\M$ is row normalized. We use the ``makeneighborsw'' function from the Spatial Econometrics Toolbox to generate $\M$ \citep{Lesage:2009}.

We evaluate the performance of the following estimators: (i) the QMLE in \eqref{qmle3.4}, (ii) the ME in \eqref{me}, (iii) the IGMME in \eqref{igmme5.1}, (iv) the BGMME in \eqref{bgmme}, (v) the RGMME in \eqref{rgmme5.8}, (vi) the Bayesian estimator (BE) based on Algorithm 1, and (vii) the robust Bayesian estimator (RBE) based on Algorithm 2.  For classical estimation methods, we conduct $1000$ repetitions. In the case of Bayesian estimation, we choose the following priors: $\lambda\sim N(0,100)$, $\rho\sim N(0,100)$, $\bs{\beta}\sim N(\mf{0},100\mf{I}_2)$ and $\sigma^2\sim\text{IG}(0.01,0.01)$. We set the number of repetitions to $100$, the number of draws to $1500$, and the burn-ins to $500$. For each method, we report bias, root mean squared error (RMSE) and empirical coverage ratio of a 95\% confidence interval.  

\subsection{Simulation results}
Tables~\ref{t1} and \ref{t2} present the simulation results for two homoskedastic cases: (i) $v_i\sim N(0,1)$ and (ii) $v_i\sim(\chi^2_3-3)/\sqrt{6}$. Similarly, Tables~\ref{t3} and \ref{t4} report the simulation results for the heteroskedastic cases. Below, we summarize our main findings from these tables.

\begin{enumerate}
\item The results in Tables~\ref{t1} and \ref{t2} demonstrate that all estimators exhibit excellent finite sample performance in terms of bias across all cases. All estimators report negligible bias for all parameters. For instance, in Table~\ref{t1}, when $(\alpha_0,\tau_0,\beta_{10},\beta_{20})$ is $(-2,-1,1,1)$ in the case of $\mathbf{W}_1$, in terms of bias in estimating $\alpha_0$, the QMLE, IGMME, BGMME, ME and BE report $-0.0023$, $-0.0009$, $-0.0026$, $-0.0021$, and $-0.0024$, respectively. 
\item In Tables~\ref{t1} and \ref{t2}, in terms of finite sample efficiency, the QMLE, BGMME and BE outperform the other estimators and report smaller RMSE when the true disturbance terms are normally distributed. However, when the true disturbance terms are not normally distributed, we observe that the BGMME reports the smallest RMSE. This is not surprising because the theoretical results in \citet{Jin:2015} show that when the disturbance terms are not normally distributed and $\W$ and $\M$ do not commute, the BGMME can be more efficient than the QMLE. For example, in Table~\ref{t2}, when $(\alpha_0,\tau_0,\beta_{10},\beta_{20})$ is $(-2,-1,1,1)$ in the case of $\W_2$, for $\alpha_0$ the BGMME reports 0.045 for RMSE, whereas the QMLE, IGMME, ME and BE report 0.051, 0.060, 0.052 and 0.061, respectively. 
\item In Tables~\ref{t1} and \ref{t2}, in terms of finite sample coverage ratios, all estimators perform satisfactorily regardless of the distribution of the true disturbance terms or the denseness of $\W$. There are occasional negligible under coverage cases for the ME and the BE for $\alpha_0$. This is not surprising in the case of ME because it uses the adjusted quasi score (with respect to $\alpha$) that tries to correct the score for the potential heteroskedasticity in the disturbance terms. For example, in Table~\ref{t1}, when $(\alpha_0,\tau_0,\beta_{10},\beta_{20})$ is $(-2,-1,1,1)$ in the case of $\mathbf{W}_1$, for $\alpha_0$, the QMLE, IGMME, BGMME, ME and BE report 94.5\%, 93.4\%, 94.6\%, 92.7\%, and 95\%, respectively. Overall, all estimators perform satisfactorily. 
\item In the heteroskedastic cases, the results in Tables~\ref{t3} and \ref{t4} indicate that all estimators exhibit excellent finite sample performance in terms of bias. For example, in Table~\ref{t4}, when $(\alpha_0,\tau_0,\beta_{10},\beta_{20})$ is $(0.5,-1,1,1)$ in the case of $\mathbf{W}_2$, in terms of bias in estimating $\tau_0$, the QMLE, IGMME, RGMME, ME and BE report $0.0006$, $0.0077$, $0.0011$, $0.0004$, and $0.0099$, respectively. 
\item In Tables~\ref{t3} and \ref{t4}, in terms of finite sample efficiency, the QMLE, RGMME, ME and RBE perform similarly. The RGMME and RBE report smaller RMSE values, whereas the IGMME reports the largest RMSE values. For example, in Table~\ref{t4}, when $(\alpha_0,\tau_0,\beta_{10},\beta_{20})$ is $(0.5,-1,1,1)$ in the case of $\W_1$, for $\tau_0$ the RGMME and RBE report 0.089 and 0.082 for RMSE, respectively, whereas the QMLE, IGMME, and ME report 0.092, 0.106 and 0.092, respectively.
\item In Tables~\ref{t3} and \ref{t4}, in terms of finite sample coverage ratio, all estimators perform satisfactorily. For example, in Table~\ref{t4}, when $(\alpha_0,\tau_0,\beta_{10},\beta_{20})$ is $(0.5,1,1,1)$ in the case of $\mathbf{W}_1$, for $\tau_0$, the QMLE, IGMME, RGMME, ME and RBE report 94.5\%, 94.4\%, 94.8\%, 95\%, and 92\%, respectively. 

\end{enumerate}

\clearpage
\begin{table}[htbp]
  \centering
  \caption{Estimation results under homoskedasticity with $v_i\sim N(0,1)$}
  \resizebox{\textwidth}{!}{
    \begin{tabular}{lrrrrr}
    \toprule
          & \multicolumn{1}{c}{QMLE} & \multicolumn{1}{c}{IGMME} & \multicolumn{1}{c}{BGMME} & \multicolumn{1}{c}{ME} & \multicolumn{1}{c}{BE} \\
\cmidrule{2-6}          & \multicolumn{5}{c}{$\mathbf{W}_1$} \\
    \midrule
    $\alpha_0=-2$ & $-$0.0023(.043)[.945] & $-$0.0009(.054)[.934] & $-$0.0026(.043)[.946] & $-$0.0021(.045)[.927] & $-$0.0024(.045)[.950] \\
    $\tau_0=-1$ & 0.0015(.089)[.941] & 0.0041(.100)[.938] & 0.0013(.089)[.939] & 0.0013(.090)[.941] & 0.0141(.089)[.930] \\
    $\beta_{10}=1$ & 0.0034(.040)[.955] & 0.0034(.042)[.960] & 0.0031(.040)[.953] & 0.0030(.040)[.909] & 0.0036(.042)[.970] \\
    $\beta_{20}=1$ & $-$0.0007(.035)[.943] & $-$0.0011(.037)[.947] & $-$0.0009(.035)[.942] & 0.0010(.035)[.976] & 0.0002(.035)[.940] \\
    \midrule
          & \multicolumn{5}{c}{$\mathbf{W}_1$} \\
              \midrule
    $\alpha_0=-2$ & $-$0.0006(.038)[.942] & $-$0.0005(.044)[.949] & $-$0.0002(.039)[.939] & $-$0.0012(.048)[.908] & 0.0037(.036)[.970] \\
    $\tau_0=1$ & 0.0140(.088)[.942] & 0.0094(.097)[.941] & 0.0115(.088)[.941] & 0.0141(.091)[.954] & 0.0066(.080)[.960] \\
    $\beta_{10}=1$ & 0.0012(.040)[.946] & 0.0010(.043)[.945] & 0.0011(.040)[.946] & 0.0020(.041)[.972] & 0.0039(.040)[.960] \\
    $\beta_{20}=1$ & 0.0000(.038)[.944] & 0.0004(.044)[.938] & 0.0004(.038)[.939] & 0.0000(.047)[.940] & 0.0054(.036)[.960] \\
    \midrule
          & \multicolumn{5}{c}{$\mathbf{W}_1$} \\
              \midrule
    $\alpha_0=0.5$ & $-$0.0002(.048)[.939] & 0.0022(.061)[.948] & $-$0.0004(.048)[.942] & $-$0.0001(.049)[.883] & 0.0021(.048)[.960] \\
    $\tau_0=-1$ & 0.0032(.092)[.936] & 0.0037(.106)[.949] & 0.0028(.092)[.931] & 0.0031(.092)[.940] & $-$0.0032(.093)[.940] \\
    $\beta_{10}=1$ & 0.0020(.041)[.944] & 0.0014(.043)[.939] & 0.0017(.041)[.944] & 0.0024(.041)[.890] & $-$0.0001(.046)[.950] \\
    $\beta_{20}=1$ & 0.0000(.034)[.954] & 0.0001(.034)[.955] & $-$0.0002(.034)[.951] & 0.0020(.034)[.995] & $-$0.0026(.034)[.940] \\
    \midrule
          & \multicolumn{5}{c}{$\mathbf{W}_1$} \\
              \midrule
    $\alpha_0=0.5$ & $-$0.0009(.041)[.954] & $-$0.0005(.047)[.951] & $-$0.0005(.041)[.955] & $-$0.0034(.054)[.891] & 0.0037(.040)[.930] \\
    $\tau_0=1$ & 0.0147(.089)[.940] & 0.0094(.096)[.953] & 0.0121(.089)[.942] & 0.0168(.097)[.957] & $-$0.0073(.078)[.980] \\
    $\beta_{10}=1$ & $-$0.0017(.042)[.934] & $-$0.0016(.045)[.936] & $-$0.0020(.043)[.933] & $-$0.0006(.043)[.971] & 0.0011(.041)[.950] \\
    $\beta_{20}=1$ & $-$0.0006(.040)[.947] & 0.0000(.046)[.951] & $-$0.0002(.040)[.946] & $-$0.0003(.052)[.955] & 0.0060(.040)[.930] \\
    \midrule
          & \multicolumn{5}{c}{$\mathbf{W}_2$} \\
              \midrule
    $\alpha_0=-2$ & $-$0.0001(.051)[.951] & 0.0003(.061)[.952] & 0.0000(.051)[.952] & $-$0.0004(.052)[.894] & 0.0057(.053)[.960] \\
    $\tau_0=-1$ & 0.0074(.083)[.944] & 0.0111(.093)[.950] & 0.0065(.084)[.943] & 0.0077(.083)[.948] & 0.0068(.073)[.980] \\
    $\beta_{10}=1$ & 0.0034(.044)[.948] & 0.0036(.046)[.950] & 0.0031(.044)[.948] & 0.0036(.044)[.874] & $-$0.0049(.044)[.940] \\
    $\beta_{20}=1$ & $-$0.0006(.034)[.943] & $-$0.0004(.035)[.951] & $-$0.0013(.034)[.944] & 0.0005(.034)[.997] & $-$0.0002(.031)[.970] \\
    \midrule
          & \multicolumn{5}{c}{$\mathbf{W}_2$} \\
              \midrule
    $\alpha_0=-2$ & $-$0.0013(.038)[.959] & $-$0.0016(.045)[.952] & $-$0.0012(.038)[.958] & $-$0.0037(.054)[.845] & 0.0010(.044)[.900] \\
    $\tau_0=1$ & 0.0143(.083)[.941] & 0.0096(.087)[.950] & 0.0106(.082)[.940] & 0.0143(.086)[.956] & $-$0.0219(.083)[.950] \\
    $\beta_{10}=1$ & 0.0008(.036)[.948] & 0.0007(.039)[.949] & 0.0003(.036)[.949] & 0.0015(.038)[.985] & $-$0.0015(.040)[.950] \\
    $\beta_{20}=1$ & $-$0.0006(.037)[.955] & $-$0.0006(.044)[.954] & $-$0.0006(.037)[.957] & 0.0001(.051)[.953] & 0.0011(.043)[.900] \\
    \midrule
          & \multicolumn{5}{c}{$\mathbf{W}_2$} \\
              \midrule
    $\alpha_0=0.5$ & 0.0025(.050)[.953] & 0.0025(.062)[.933] & 0.0027(.050)[.952] & 0.0025(.052)[.845] & 0.0122(.056)[.880] \\
    $\tau_0=-1$ & $-$0.0030(.086)[.941] & 0.0033(.096)[.939] & $-$0.0036(.087)[.937] & $-$0.0031(.087)[.933] & $-$0.0014(.089)[.930] \\
    $\beta_{10}=1$ & 0.0001(.038)[.963] & $-$0.0006(.041)[.947] & $-$0.0003(.038)[.962] & 0.0008(.038)[.933] & 0.0046(.041)[.930] \\
    $\beta_{20}=1$ & $-$0.0025(.035)[.951] & $-$0.0024(.036)[.941] & $-$0.0030(.035)[.947] & $-$0.0015(.035)[.998] & 0.0035(.028)[.970] \\
    \midrule
          & \multicolumn{5}{c}{$\mathbf{W}_2$} \\
              \midrule
    $\alpha_0=0.5$ & $-$0.0005(.038)[.949] & $-$0.0008(.044)[.948] & $-$0.0005(.038)[.948] & 0.0010(.049)[.869] & $-$0.0010(.040)[.910] \\
    $\tau_0=1$ & 0.0099(.080)[.950] & 0.0056(.085)[.957] & 0.0063(.080)[.950] & 0.0082(.082)[.960] & $-$0.0023(.085)[.950] \\
    $\beta_{10}=1$ & 0.0001(.037)[.945] & 0.0001(.040)[.940] & $-$0.0004(.037)[.943] & 0.0006(.038)[.981] & $-$0.0036(.034)[.980] \\
    $\beta_{20}=1$ & $-$0.0001(.038)[.944] & $-$0.0002(.044)[.949] & $-$0.0001(.038)[.943] & $-$0.0002(.049)[.942] & 0.0010(.040)[.900] \\
    \bottomrule
    \end{tabular}%
    }
    {\raggedright\footnotesize\textit{Notes:} We report the bias (RMSE) [95\% coverage ratio].\par}
  \label{t1}%
\end{table}%

\clearpage
\begin{table}[htbp]
  \centering
  \caption{Estimation results under homoskedasticity with $v_i\sim(\chi^2_3-3)/\sqrt{6}$}
  \resizebox{\textwidth}{!}{
    \begin{tabular}{lrrrrr}
    \toprule
          & \multicolumn{1}{c}{QMLE} & \multicolumn{1}{c}{IGMME} & \multicolumn{1}{c}{BGMME} & \multicolumn{1}{c}{ME} & \multicolumn{1}{c}{BE} \\
\cmidrule{2-6}          & \multicolumn{5}{c}{$\mathbf{W}_1$} \\
    \midrule
    $\alpha_0=-2$ & 0.0012(.045)[.933] & 0.0009(.054)[.946] & $-$0.0024(.036)[.945] & 0.0010(.045)[.906] & $-$0.0006(.040)[.950] \\
    $\tau_0=-1$ & 0.0005(.088)[.944] & 0.0037(.099)[.942] & 0.0036(.086)[.939] & 0.0009(.088)[.944] & 0.0037(.074)[.990] \\
    $\beta_{10}=1$ & $-$0.0016(.043)[.944] & $-$0.0017(.045)[.943] & $-$0.0004(.032)[.948] & $-$0.0008(.043)[.890] & 0.0018(.041)[.950] \\
    $\beta_{20}=1$ & $-$0.0001(.035)[.948] & 0.0003(.037)[.945] & 0.0006(.027)[.961] & 0.0018(.035)[.981] & 0.0011(.036)[.930] \\
    \midrule
          & \multicolumn{5}{c}{$\mathbf{W}_1$} \\
              \midrule
    $\alpha_0=-2$ & 0.0016(.038)[.946] & 0.0020(.043)[.943] & 0.0002(.029)[.948] & 0.0007(.046)[.910] & 0.0022(.041)[.930] \\
    $\tau_0=1$ & 0.0071(.088)[.936] & 0.0013(.094)[.940] & 0.0057(.086)[.927] & 0.0079(.091)[.943] & $-$0.0006(.081)[.960] \\
    $\beta_{10}=1$ & 0.0013(.040)[.941] & 0.0016(.042)[.937] & 0.0002(.028)[.957] & 0.0013(.041)[.968] & $-$0.0027(.043)[.910] \\
    $\beta_{20}=1$ & 0.0027(.038)[.953] & 0.0033(.042)[.951] & 0.0015(.029)[.946] & 0.0006(.044)[.951] & 0.0034(.041)[.940] \\
    \midrule
          & \multicolumn{5}{c}{$\mathbf{W}_1$} \\
              \midrule
    $\alpha_0=0.5$ & 0.0024(.049)[.947] & 0.0027(.061)[.941] & $-$0.0003(.041)[.946] & 0.0025(.050)[.891] & 0.0020(.044)[.970] \\
    $\tau_0=-1$ & 0.0027(.092)[.942] & 0.0061(.106)[.953] & 0.0047(.089)[.937] & 0.0027(.092)[.942] & $-$0.0051(.082)[.950] \\
    $\beta_{10}=1$ & 0.0017(.041)[.941] & 0.0017(.043)[.932] & 0.0005(.030)[.946] & 0.0026(.041)[.890] & $-$0.0010(.046)[.950] \\
    $\beta_{20}=1$ & 0.0011(.034)[.934] & 0.0010(.035)[.936] & 0.0017(.027)[.951] & 0.0031(.034)[.990] & $-$0.0046(.036)[.930] \\
    \midrule
          & \multicolumn{5}{c}{$\mathbf{W}_1$} \\
              \midrule
    $\alpha_0=0.5$ & 0.0009(.042)[.940] & 0.0016(.046)[.945] & $-$0.0009(.031)[.949] & $-$0.0022(.055)[.881] & $-$0.0039(.040)[.940] \\
    $\tau_0=1$ & 0.0121(.087)[.945] & 0.0050(.096)[.944] & 0.0114(.083)[.948] & 0.0149(.094)[.950] & 0.0010(.089)[.920] \\
    $\beta_{10}=1$ & $-$0.0003(.042)[.941] & 0.0000(.044)[.943] & $-$0.0008(.031)[.939] & 0.0003(.043)[.976] & $-$0.0051(.039)[.930] \\
    $\beta_{20}=1$ & 0.0019(.041)[.944] & 0.0027(.045)[.949] & 0.0002(.030)[.946] & 0.0003(.052)[.948] & $-$0.0038(.040)[.950] \\
    \midrule
          & \multicolumn{5}{c}{$\mathbf{W}_2$} \\
              \midrule
    $\alpha_0=-2$ & 0.0036(.051)[.949] & 0.0025(.060)[.954] & $-$0.0011(.045)[.946] & 0.0036(.052)[.888] & $-$0.0059(.061)[.900] \\
    $\tau_0=-1$ & 0.0006(.084)[.939] & 0.0069(.093)[.948] & 0.0028(.083)[.939] & 0.0008(.085)[.941] & $-$0.0052(.080)[.980] \\
    $\beta_{10}=1$ & $-$0.0017(.042)[.948] & $-$0.0017(.044)[.943] & $-$0.0018(.032)[.956] & $-$0.0011(.042)[.883] & 0.0014(.041)[.930] \\
    $\beta_{20}=1$ & $-$0.0013(.033)[.952] & $-$0.0015(.034)[.953] & $-$0.0016(.025)[.949] & $-$0.0002(.033)[.999] & $-$0.0041(.036)[.950] \\
    \midrule
          & \multicolumn{5}{c}{$\mathbf{W}_2$}  \\
              \midrule
    $\alpha_0=-2$ & 0.0009(.039)[.939] & 0.0015(.045)[.945] & $-$0.0001(.030)[.951] & $-$0.0008(.053)[.847] & $-$0.0072(.036)[.940] \\
    $\tau_0=1$ & 0.0102(.075)[.958] & 0.0050(.082)[.959] & 0.0047(.075)[.960] & 0.0102(.079)[.963] & 0.0029(.083)[.920] \\
    $\beta_{10}=1$ & $-$0.0011(.036)[.943] & $-$0.0016(.039)[.953] & $-$0.0010(.026)[.952] & $-$0.0011(.039)[.983] & $-$0.0025(.039)[.920] \\
    $\beta_{20}=1$ & 0.0014(.038)[.945] & 0.0023(.044)[.950] & 0.0007(.029)[.949] & 0.0000(.050)[.952] & $-$0.0057(.036)[.950] \\
    \midrule
          & \multicolumn{5}{c}{$\mathbf{W}_2$} \\
              \midrule
    $\alpha_0=0.5$ & 0.0019(.050)[.947] & 0.0012(.061)[.949] & $-$0.0007(.045)[.943] & 0.0023(.051)[.857] & $-$0.0002(.054)[.930] \\
    $\tau_0=-1$ & 0.0006(.085)[.946] & 0.0077(.094)[.941] & 0.0011(.085)[.936] & 0.0004(.085)[.945] & 0.0099(.091)[.940] \\
    $\beta_{10}=1$ & $-$0.0033(.037)[.950] & $-$0.0017(.040)[.945] & $-$0.0026(.029)[.941] & $-$0.0027(.038)[.922] & $-$0.0022(.043)[.940] \\
    $\beta_{20}=1$ & $-$0.0004(.035)[.959] & $-$0.0008(.036)[.951] & 0.0003(.026)[.957] & 0.0006(.035)[.997] & $-$0.0004(.032)[.950] \\
    \midrule
          & \multicolumn{5}{c}{$\mathbf{W}_2$} \\
              \midrule
    $\alpha_0=0.5$ & 0.0006(.039)[.946] & $-$0.0006(.045)[.945] & 0.0009(.029)[.949] & $-$0.0016(.050)[.863] & 0.0124(.048)[.850] \\
    $\tau_0=1$ & 0.0142(.082)[.940] & 0.0114(.088)[.940] & 0.0093(.081)[.940] & 0.0149(.084)[.943] & 0.0038(.089)[.930] \\
    $\beta_{10}=1$ & $-$0.0012(.035)[.953] & $-$0.0022(.038)[.952] & $-$0.0007(.026)[.958] & $-$0.0010(.037)[.986] & 0.0090(.041)[.870] \\
    $\beta_{20}=1$ & 0.0008(.039)[.948] & 0.0000(.045)[.947] & 0.0011(.028)[.957] & $-$0.0004(.049)[.952] & 0.0138(.048)[.880] \\
    \bottomrule
\end{tabular}%
}
    {\raggedright\footnotesize\textit{Notes:} We report the bias (RMSE) [95\% coverage ratio].\par}
  \label{t2}%
\end{table}%

\clearpage
\begin{table}[htbp]
  \centering
  \caption{Estimation results under heteroskedasticity with $\gamma_{i}=2 \vartheta_i / (\sum_{j=1}^n \vartheta_j/n)$}
  \resizebox{\textwidth}{!}{
    \begin{tabular}{lrrrrr}
    \toprule
          & \multicolumn{1}{c}{QMLE} & \multicolumn{1}{c}{IGMME} & \multicolumn{1}{c}{RGMME} & \multicolumn{1}{c}{ME} & \multicolumn{1}{c}{RBE} \\
\cmidrule{2-6}          & \multicolumn{5}{c}{$\mathbf{W}_1$} \\
     \midrule
    $\alpha_0=-2$ & $-$0.0014(.057)[.950] & 0.0004(.074)[.945] & $-$0.0022(.057)[.953] & $-$0.0011(.059)[.948] & $-$0.0049(.049)[.960] \\
    $\tau_0=-1$ & 0.0001(.104)[.911] & 0.0043(.112)[.945] & $-$0.0017(.111)[.893] & $-$0.0001(.105)[.942] & $-$0.0012(.103)[.920] \\
    $\beta_{10}=1$ & 0.0044(.057)[.952] & 0.0046(.059)[.954] & 0.0037(.057)[.953] & 0.0044(.052)[.933] & 0.0013(.043)[.960] \\
    $\beta_{20}=1$ & $-$0.0016(.049)[.943] & $-$0.0017(.051)[.953] & $-$0.0028(.049)[.944] & 0.0014(.049)[.970] & 0.0118(.048)[.970] \\
    \midrule
          & \multicolumn{5}{c}{$\mathbf{W}_1$} \\
    \midrule
    $\alpha_0=-2$ & 0.0024(.052)[.947] & 0.0000(.057)[.957] & 0.0008(.052)[.948] & $-$0.0021(.064)[.913] & 0.0104(.059)[.820] \\
    $\tau_0=1$ & 0.0097(.096)[.938] & 0.0097(.096)[.960] & 0.0068(.101)[.923] & 0.0131(.100)[.962] & $-$0.0037(.097)[.900] \\
    $\beta_{10}=1$ & 0.0015(.054)[.951] & $-$0.0002(.056)[.950] & 0.0000(.054)[.947] & 0.0011(.047)[.992] & $-$0.0035(.042)[.930] \\
    $\beta_{20}=1$ & 0.0035(.052)[.945] & 0.0015(.057)[.961] & 0.0019(.052)[.944] & 0.0001(.013)[.999] & 0.0116(.060)[.840] \\
    \midrule
          & \multicolumn{5}{c}{$\mathbf{W}_1$} \\
    \midrule
    $\alpha_0=0.5$ & $-$0.0036(.056)[.965] & $-$0.0004(.071)[.955] & $-$0.0042(.056)[.963] & $-$0.0039(.058)[.958] & 0.0005(.050)[.980] \\
    $\tau_0=-1$ & 0.0071(.104)[.931] & 0.0074(.108)[.962] & 0.0073(.111)[.909] & 0.0075(.104)[.946] & $-$0.0026(.086)[.970] \\
    $\beta_{10}=1$ & 0.0010(.057)[.948] & 0.0010(.059)[.948] & 0.0003(.057)[.945] & 0.0009(.056)[.891] & 0.0027(.041)[.960] \\
    $\beta_{20}=1$ & $-$0.0019(.044)[.972] & $-$0.0016(.045)[.971] & $-$0.0033(.044)[.970] & 0.0011(.044)[.989] & $-$0.0061(.050)[.970] \\
    \midrule
          & \multicolumn{5}{c}{$\mathbf{W}_1$} \\
    \midrule
    $\alpha_0=0.5$ & 0.0016(.058)[.937] & $-$0.0006(.062)[.963] & 0.0002(.058)[.936] & $-$0.0111(.095)[.892] & 0.0056(.047)[.900] \\
    $\tau_0=1$ & 0.0075(.103)[.924] & 0.0060(.102)[.962] & 0.0058(.109)[.907] & 0.0174(.125)[.976] & $-$0.0143(.098)[.890] \\
    $\beta_{10}=1$ & 0.0002(.060)[.935] & $-$0.0004(.063)[.939] & $-$0.0012(.061)[.930] & 0.0002(.054)[.993] & 0.0016(.041)[.920] \\
    $\beta_{20}=1$ & 0.0030(.059)[.937] & 0.0012(.064)[.956] & 0.0016(.059)[.935] & 0.0000(.014)[.999] & 0.0072(.047)[.900] \\
    \midrule
          & \multicolumn{5}{c}{$\mathbf{W}_2$}  \\
    \midrule
\cmidrule{2-6}    $\alpha_0=-2$ & 0.0033(.068)[.953] & 0.0027(.091)[.941] & 0.0031(.067)[.954] & 0.0049(.069)[.883] & $-$0.0096(.061)[.999] \\
    $\tau_0=-1$ & 0.0018(.094)[.935] & 0.0093(.105)[.956] & $-$0.0015(.095)[.930] & 0.0008(.094)[.945] & $-$0.0129(.090)[.930] \\
    $\beta_{10}=1$ & 0.0016(.054)[.942] & 0.0019(.057)[.948] & 0.0011(.055)[.945] & 0.0023(.053)[.945] & 0.0051(.039)[.980] \\
    $\beta_{20}=1$ & $-$0.0020(.048)[.955] & $-$0.0022(.052)[.947] & $-$0.0034(.049)[.954] & 0.0000(.048)[.995] & $-$0.0063(.045)[.960] \\
    \midrule
          & \multicolumn{5}{c}{$\mathbf{W}_2$}  \\
    \midrule
    $\alpha_0=-2$ & 0.0001(.052)[.958] & $-$0.0028(.061)[.960] & $-$0.0029(.053)[.955] & $-$0.0062(.074)[.900] & $-$0.0017(.054)[.920] \\
    $\tau_0=1$ & 0.0117(.091)[.928] & 0.0100(.095)[.946] & 0.0078(.092)[.927] & 0.0147(.097)[.948] & $-$0.0012(.089)[.940] \\
    $\beta_{10}=1$ & 0.0023(.059)[.939] & 0.0013(.062)[.956] & 0.0010(.060)[.940] & 0.0028(.058)[.992] & 0.0001(.038)[.980] \\
    $\beta_{20}=1$ & 0.0012(.052)[.957] & $-$0.0011(.061)[.960] & $-$0.0018(.052)[.954] & $-$0.0001(.013)[.999] & 0.0009(.054)[.890] \\
    \midrule
          & \multicolumn{5}{c}{$\mathbf{W}_2$}  \\
    \midrule
    $\alpha_0=0.5$ & 0.0069(.074)[.951] & 0.0082(.094)[.940] & 0.0074(.074)[.947] & 0.0086(.076)[.843] & 0.0086(.059)[.980] \\
    $\tau_0=-1$ & $-$0.0007(.099)[.923] & 0.0043(.105)[.949] & $-$0.0036(.100)[.918] & $-$0.0016(.099)[.943] & 0.0074(.089)[.930] \\
    $\beta_{10}=1$ & $-$0.0024(.056)[.943] & $-$0.0017(.059)[.955] & $-$0.0028(.056)[.945] & $-$0.0011(.055)[.926] & 0.0055(.036)[.970] \\
    $\beta_{20}=1$ & 0.0004(.050)[.954] & 0.0014(.051)[.948] & $-$0.0008(.050)[.951] & 0.0021(.049)[.999] & $-$0.0061(.049)[.930] \\
    \midrule
          & \multicolumn{5}{c}{$\mathbf{W}_2$}  \\
    \midrule
    $\alpha_0=0.5$ & 0.0030(.056)[.942] & $-$0.0039(.065)[.954] & $-$0.0004(.057)[.945] & 0.0009(.073)[.871] & 0.0068(.059)[.880] \\
    $\tau_0=1$ & 0.0048(.088)[.947] & 0.0069(.092)[.955] & 0.0009(.089)[.947] & 0.0053(.092)[.963] & 0.0079(.084)[.930] \\
    $\beta_{10}=1$ & 0.0037(.055)[.944] & 0.0030(.058)[.944] & 0.0022(.055)[.942] & 0.0033(.052)[.992] & 0.0019(.034)[.990] \\
    $\beta_{20}=1$ & 0.0042(.056)[.935] & $-$0.0017(.062)[.951] & 0.0010(.056)[.936] & 0.0001(.013)[.999] & 0.0102(.054)[.880] \\
    \bottomrule
\end{tabular}%
}
    {\raggedright\footnotesize\textit{Notes:} We report the bias (RMSE) [95\% coverage ratio].\par}
  \label{t3}%
\end{table}%

\clearpage
\begin{table}[htbp]
  \centering
  \caption{Estimation results under heteroskedasticity with $\gamma_{i}=\exp(0.1+0.35X_{2i})$}
  \resizebox{\textwidth}{!}{
    \begin{tabular}{lrrrrr}
    \toprule
          & \multicolumn{1}{c}{QMLE} & \multicolumn{1}{c}{IGMME} & \multicolumn{1}{c}{RGMME} & \multicolumn{1}{c}{ME} & \multicolumn{1}{c}{RBE} \\
\cmidrule{2-6}          & \multicolumn{5}{c}{$\mathbf{W}_1$} \\
    \midrule
    $\alpha_0=-2$ & 0.0012(.045)[.933] & 0.0009(.054)[.946] & $-$0.0024(.036)[.945] & 0.0010(.045)[.906] & $-$0.0006(.040)[.950] \\
    $\tau_0=-1$ & 0.0005(.088)[.944] & 0.0037(.099)[.942] & 0.0036(.086)[.939] & 0.0009(.088)[.944] & 0.0037(.074)[.990] \\
    $\beta_{10}=1$ & $-$0.0016(.043)[.944] & $-$0.0017(.045)[.943] & $-$0.0004(.032)[.948] & $-$0.0008(.043)[.890] & 0.0018(.041)[.950] \\
    $\beta_{20}=1$ & $-$0.0001(.035)[.948] & 0.0003(.037)[.945] & 0.0006(.027)[.961] & 0.0018(.035)[.981] & 0.0011(.036)[.930] \\
    \midrule
          & \multicolumn{5}{c}{$\mathbf{W}_1$} \\
              \midrule
    $\alpha_0=-2$ & 0.0016(.038)[.946] & 0.0020(.043)[.943] & 0.0002(.029)[.948] & 0.0007(.046)[.910] & 0.0022(.041)[.930] \\
    $\tau_0=1$ & 0.0071(.088)[.936] & 0.0013(.094)[.940] & 0.0057(.086)[.927] & 0.0079(.091)[.943] & $-$0.0006(.081)[.960] \\
    $\beta_{10}=1$ & 0.0013(.040)[.941] & 0.0016(.042)[.937] & 0.0002(.028)[.957] & 0.0013(.041)[.968] & $-$0.0027(.043)[.910] \\
    $\beta_{20}=1$ & 0.0027(.038)[.953] & 0.0033(.042)[.951] & 0.0015(.029)[.946] & 0.0006(.044)[.951] & 0.0034(.041)[.940] \\
    \midrule
          & \multicolumn{5}{c}{$\mathbf{W}_1$} \\
              \midrule
    $\alpha_0=0.5$ & 0.0024(.049)[.947] & 0.0027(.061)[.941] & $-$0.0003(.041)[.946] & 0.0025(.050)[.891] & 0.0020(.044)[.970] \\
    $\tau_0=-1$ & 0.0027(.092)[.942] & 0.0061(.106)[.953] & 0.0047(.089)[.937] & 0.0027(.092)[.942] & $-$0.0051(.082)[.950] \\
    $\beta_{10}=1$ & 0.0017(.041)[.941] & 0.0017(.043)[.932] & 0.0005(.030)[.946] & 0.0026(.041)[.890] & $-$0.0010(.046)[.950] \\
    $\beta_{20}=1$ & 0.0011(.034)[.934] & 0.0010(.035)[.936] & 0.0017(.027)[.951] & 0.0031(.034)[.990] & $-$0.0046(.036)[.930] \\
    \midrule
          & \multicolumn{5}{c}{$\mathbf{W}_1$} \\
              \midrule
    $\alpha_0=0.5$ & 0.0009(.042)[.940] & 0.0016(.046)[.945] & $-$0.0009(.031)[.949] & $-$0.0022(.055)[.881] & $-$0.0039(.040)[.940] \\
    $\tau_0=1$ & 0.0121(.087)[.945] & 0.0050(.096)[.944] & 0.0114(.083)[.948] & 0.0149(.094)[.950] & 0.0010(.089)[.920] \\
    $\beta_{10}=1$ & $-$0.0003(.042)[.941] & 0.0000(.044)[.943] & $-$0.0008(.031)[.939] & 0.0003(.043)[.976] & $-$0.0051(.039)[.930] \\
    $\beta_{20}=1$ & 0.0019(.041)[.944] & 0.0027(.045)[.949] & 0.0002(.030)[.946] & 0.0003(.052)[.948] & $-$0.0038(.040)[.950] \\
    \midrule
          & \multicolumn{5}{c}{$\mathbf{W}_2$} \\
              \midrule
    $\alpha_0=-2$ & 0.0036(.051)[.949] & 0.0025(.060)[.954] & $-$0.0011(.045)[.946] & 0.0036(.052)[.888] & $-$0.0059(.061)[.900] \\
    $\tau_0=-1$ & 0.0006(.084)[.939] & 0.0069(.093)[.948] & 0.0028(.083)[.939] & 0.0008(.085)[.941] & $-$0.0052(.080)[.980] \\
    $\beta_{10}=1$ & $-$0.0017(.042)[.948] & $-$0.0017(.044)[.943] & $-$0.0018(.032)[.956] & $-$0.0011(.042)[.883] & 0.0014(.041)[.930] \\
    $\beta_{20}=1$ & $-$0.0013(.033)[.952] & $-$0.0015(.034)[.953] & $-$0.0016(.025)[.949] & $-$0.0002(.033)[.999] & $-$0.0041(.036)[.950] \\
    \midrule
          & \multicolumn{5}{c}{$\mathbf{W}_2$}  \\
              \midrule
    $\alpha_0=-2$ & 0.0009(.039)[.939] & 0.0015(.045)[.945] & $-$0.0001(.030)[.951] & $-$0.0008(.053)[.847] & $-$0.0072(.036)[.940] \\
    $\tau_0=1$ & 0.0102(.075)[.958] & 0.0050(.082)[.959] & 0.0047(.075)[.960] & 0.0102(.079)[.963] & 0.0029(.083)[.920] \\
    $\beta_{10}=1$ & $-$0.0011(.036)[.943] & $-$0.0016(.039)[.953] & $-$0.0010(.026)[.952] & $-$0.0011(.039)[.983] & $-$0.0025(.039)[.920] \\
    $\beta_{20}=1$ & 0.0014(.038)[.945] & 0.0023(.044)[.950] & 0.0007(.029)[.949] & 0.0000(.050)[.952] & $-$0.0057(.036)[.950] \\
    \midrule
          & \multicolumn{5}{c}{$\mathbf{W}_2$} \\
              \midrule
    $\alpha_0=0.5$ & 0.0019(.050)[.947] & 0.0012(.061)[.949] & $-$0.0007(.045)[.943] & 0.0023(.051)[.857] & $-$0.0002(.054)[.930] \\
    $\tau_0=-1$ & 0.0006(.085)[.946] & 0.0077(.094)[.941] & 0.0011(.085)[.936] & 0.0004(.085)[.945] & 0.0099(.091)[.940] \\
    $\beta_{10}=1$ & $-$0.0033(.037)[.950] & $-$0.0017(.040)[.945] & $-$0.0026(.029)[.941] & $-$0.0027(.038)[.922] & $-$0.0022(.043)[.940] \\
    $\beta_{20}=1$ & $-$0.0004(.035)[.959] & $-$0.0008(.036)[.951] & 0.0003(.026)[.957] & 0.0006(.035)[.997] & $-$0.0004(.032)[.950] \\
    \midrule
          & \multicolumn{5}{c}{$\mathbf{W}_2$} \\
              \midrule
    $\alpha_0=0.5$ & 0.0006(.039)[.946] & $-$0.0006(.045)[.945] & 0.0009(.029)[.949] & $-$0.0016(.050)[.863] & 0.0124(.048)[.850] \\
    $\tau_0=1$ & 0.0142(.082)[.940] & 0.0114(.088)[.940] & 0.0093(.081)[.940] & 0.0149(.084)[.943] & 0.0038(.089)[.930] \\
    $\beta_{10}=1$ & $-$0.0012(.035)[.953] & $-$0.0022(.038)[.952] & $-$0.0007(.026)[.958] & $-$0.0010(.037)[.986] & 0.0090(.041)[.870] \\
    $\beta_{20}=1$ & 0.0008(.039)[.948] & 0.0000(.045)[.947] & 0.0011(.028)[.957] & $-$0.0004(.049)[.952] & 0.0138(.048)[.880] \\
    \bottomrule
\end{tabular}%
}
    {\raggedright\footnotesize\textit{Notes:} We report the bias (RMSE) [95\% coverage ratio].\par}
  \label{t4}%
\end{table}%

\clearpage
\section{Conclusion and outlook}\label{sec:conclusion}
In this article, we provide an extensive review of cross-sectional MESS models. We mainly focus on a first-order MESS model to discuss specification, estimation, model selection, and interpretation issues. The primary characteristic of a MESS-type model lies in its use of matrix exponential terms to specify spatial dependence in both the dependent variable and the disturbance term. These models possess several distinctive properties: 
\begin{itemize}
\item The power series representation of a matrix exponential term indicates an exponential decay of spatial dependence in these models.
 \item The reduced forms of MESS-type models always exist and do not require any restrictions on the parameter space of spatial parameters.
 \item The likelihood functions of these types of models are free of any Jacobian terms that must be computed at each iteration during the estimation process.
 \item When the spatial weights matrices are commutative, the QMLEs of these types of models can be consistent under an unknown form of heteroskedasticity.
 \item When the spatial weights matrices are not commutative, the QMLE can be inconsistent under an unknown form of heteroskedasticity. In such cases, a heteroskedasticity-robust estimation is required.
 \end{itemize}
We have provided a comprehensive description of various estimation methods, namely, the QML approach, the M-estimation approach, the GMM approach, and the Bayesian estimation approach. This detailed overview enables practitioners to easily choose and adapt a method that aligns with their specific needs. Additionally, we address estimation in the presence of endogenous explanatory variables and the Durbin terms.

In future studies, it might be interesting to consider the MESS in a social interactions scenario, and compare its implications with the SAR-type social interaction models. As the QMLE of the MESS can still remain consistent under an unknown form of heteroskedasticity, allowing for such heteroskedasticity in a social interactions model would be a significant contribution. We also think that the literature on nonlinear spatial models, such as the spatial extensions of the limited dependent variable data models, still holds some open questions, and estimation strategies for the the MESS-type limited dependent variable data models must be studied carefully. Finally, although the matrix-vector product approach to compute the matrix exponential terms can reduce the computation time significantly, we think that a faster and more reliable computation approach would be a significant contribution to the literature. 

\clearpage
\appendix
\noindent{\textbf{\Large{Appendix}}}
\section{Useful Lemmas}\label{lemmas}

In this section, we collect some commonly used lemmas for the asymptotic analysis in the spatial econometric literature.
Lemma~\ref{l1} can be found in \citet{KP:1999} and \citet{Lee:2002}. The homoskedastic and heteroskedastic versions of Lemma \ref{l2} can be found in \citet{Lee:2007a} and \citet{Lin:2010}, respectively.  Lemma~\ref{l3} can be found in \citet{Jin:2015}, Lemma~\ref{l4} gives a CLT result from \citet{KP:2010}, and Lemma \ref{l5} is a modified version of Lemma A.4 in \citet{Zhenlin:2018}.  
\begin{lemma}\label{l1}
Let $\{\mf{A}\}$ and $\{\mf{B}\}$ be two sequences of $n\times n$ matrices that are uniformly bounded in both row sum and column sum matrix norms. Let $\{\mf{C}\}$ be a sequence of conformable matrices whose elements are uniformly $O(h_n^{-1})$. Then,
\begin{enumerate}
\item[(a)] the sequence $\{\mf{A}\mf{B}\}$ are uniformly bounded in both row sum and column sum matrix norms,

\item[(b)] the elements of $\{\mf{A}\}$ are uniformly bounded and $\tr(\mf{A})=O(n)$, and 

\item[(c)] the elements of $\{\mf{A}\mf{C}\}$   and $\{\mf{C}\mf{A}\}$ are uniformly $O(h_n^{-1})$.
\end{enumerate}
\end{lemma}

\begin{lemma} \label{l2}
Let $\{\V\}$ be a sequence of random $n\times1$ column vectors, $\boldsymbol{c}$ be the $n\times1$ vector of constants, and $\left\{\mf{A}\right\}$ and $\left\{\mf{B}\right\}$ be two sequences of $n\times n$ matrices of constants. Let  $\mathrm{vec}_{D}\left(\mf{A}\right)$ be a column vector formed by the diagonal elements of $\mf{A}$, and $\mf{A}^{s}=\mf{A}+\mf{A}^{\prime}$. 
\begin{enumerate}
\item\textbf{Homoskedastic case}: Suppose that the elements of $\V$ satisfy $v_{i}\sim i.i.d.(0,\sigma_0^2)$. Let $\E(v_{i}^3)=\mu_3$  and $\E(v_{i}^4)=\mu_4$. Then, we have the following results:
\begin{enumerate}
\item $\mathrm{E}(\mf{A}\V\cdot \V^{'} \mf{B} \V)=\mu_3\mf{A}\vec_D(\mf{B})$,
\item $\mathrm{E}( \V^{'} \mf{A} \V\cdot \V^{'}\mf{B})=\mu_3\vec_D^{'}(\mf{A})\mf{B}$,
\item $\mathrm{E}( \V^{'} \mf{A} \V\cdot \V^{'} \mf{B} \V)=\left(\mu_{4}-3 \sigma_{0}^{4}\right)  \operatorname{vec}_{D}^{'}\left(\mf{A}\right) \operatorname{vec}_{D}\left(\mf{B}\right)+\sigma_{0}^{4}\left(\mathrm{tr}(\mf{A})\mathrm{tr}(\mf{B})+ \operatorname{tr}\left(\mf{A} \mf{B}^{s}\right)\right)$.
\end{enumerate}
\item\textbf{Heteroskedastic case}: Suppose $\V$ has elements that are independently distributed (i.n.i.d.) with $v_{i}\sim i.n.i.d.(0,\sigma_{i}^2)$. Let $\boldsymbol{\Sigma}=\Diag(\sigma_{1}^2,\hdots,\sigma_{n}^2)$. Then, we have the following results:
\begin{enumerate}
\item $\mathrm{E}(\V^{'}\mf{A}\V)=\mathrm{tr}\left(\boldsymbol{\Sigma} \mf{A}\right)$,
\item $\mathrm{E}(\V^{\prime} \mf{A} \V\cdot  \boldsymbol{c}^{'}\V)=\sum_{i=1}^{n}\E(v_{i}^3)a_{ii}c_i$, where $a_{ii}$ is the $(i,i)$th element of $\mf{A}$ and $c_i$ is the $i$th element of $\boldsymbol{c}$.
\item $\mathrm{E}(\V^{\prime} \mf{A} \V\cdot \V^{'} \mf{B} \V)=\sum_{i=1}^{n}a_{ii}b_{ii}(\E(v_{i}^4) - 3\sigma^4_{i})+\mathrm{tr}(\boldsymbol{\Sigma} \mf{A})\mathrm{tr}(\boldsymbol{\Sigma} \mf{B})+\mathrm{tr}(\boldsymbol{\Sigma} \mf{A}\boldsymbol{\Sigma} \mf{B}^s)$.
\end{enumerate}
If the diagonal elements of $\mf{A}$ are zeros, then these results take the following form:
\begin{enumerate}
\item $\mathrm{E}(\V^{'}\mf{A}\V)=0$,
\item $\mathrm{E}( \V^{'} \mf{A} \V\cdot  \boldsymbol{c}^{'}\V)=0,$
\item $\mathrm{E}( \V^{'} \mf{A} \V\cdot \V^{'} \mf{B} \V)=\mathrm{tr}(\boldsymbol{\Sigma} \mf{A}\boldsymbol{\Sigma} \mf{B}^s)$.
\end{enumerate}
\end{enumerate}
\end{lemma}

\begin{lemma}\label{l3}
Let $\mf{A}$ be any $n\times n$ matrix that is uniformly bounded in row sum and column sum matrix norms, and let $a=o_p(1)$. Then, $\left\Vert e^{a\mf{A}}-\mf{I}_n\right\Vert_{\infty}=o_p(1)$ and $\left\Vert e^{a\mf{A}}-\mf{I}_n\right\Vert_{1}=o_p(1)$.
\end{lemma}

\begin{lemma}\label{l4}
Suppose that $\left\{\mf{A}\right\}$ is a sequence of $n \times n$  matrices uniformly bounded in both row sum and column sum matrix norms, $\left\{\boldsymbol{c}\right\}$ is a sequence of constant column vectors such that $\sup _{n} \frac{1}{n}  \sum_{i=1}^{n}\left|c_{i}\right|^{2+\eta_{1}}<\infty$ for some $\eta_{1}>0$, $v_{i}$ in $\V$ are independent random variables with mean zero and variance $\sigma_{i}^{2}$ and $\sup _{i} \E(\left|v_{i}\right|^{4+\eta_{2}})<\infty$ for some $\eta_{2}>0$. Denote $\sigma_{Z}^{2}=\var\left(Z\right)$, where $Z=\boldsymbol{c}^{\prime} \V+\V^{'} \mf{A} \V-\operatorname{tr}\left(\mf{A}\boldsymbol{\Sigma}\right)$ where $\boldsymbol{\Sigma}=\Diag(\sigma_1^2,\hdots,\sigma_n^2)$. Assume that $\frac{1}{n} \sigma_{Z}^{2}$ is bounded away from zero. Then $\frac{Z}{\sigma_{Z}^{2}} \stackrel{d}{\longrightarrow} N(0,1)$.
\end{lemma}

\begin{lemma}\label{l5}
Let $\{\mf{A}\}$ be a sequence of $n\times n$ matrices that are bounded in both row sum and column sum matrix norms. Suppose also that the elements of $\mf{A}$ are $O\left(h_{n}^{-1}\right)$, uniformly. Let $\boldsymbol{c}$ be an $n \times 1$ vector with elements of the uniform order $O(h_{n}^{-1/2})$. Assume that the elements of the $n\times1$ innovation vector $\V$ have zero mean and finite variance, and are mutually independent.  Then,
\begin{enumerate}
\item[(1)] $\mathrm{E}(\V^{'} \mf{A} \V)=O\left(n/h_{n}\right)$, $\var(\V^{'} \mf{A} \V)=O\left(n/h_{n}\right)$,
\item[(2)]  $\V^{'} \mf{A} \V=O_{p}\left(n/h_{n}\right)$, $\V^{\prime} \mf{A} \V - \mathrm{E}(\V^{'} \mf{A} \V)=O_{p}\bigl(\left(n/h_{n}\right)^{1/2}\bigr)$, $\boldsymbol{c}^{'} \mf{A} \V=O_{p}\bigl(\left(n/h_{n}\right)^{1/2}\big)$.
\end{enumerate}
\end{lemma}

 
\section{Proofs of Main Results}\label{appB}
In this section, we provide the proofs of the main theorems in Section~\ref{sec:m-est} of the paper. 

\subsection{Proof of Theorem \ref{thm1}}\label{appB1}
As discussed in the main paper, we only need to show the consistency of $\hat{\zet}$. To that end, given Assumption \ref{a5}, we need to show that  $\sup_{\zet\in\DDDelta}\frac{1}{n}\left\Vert S^{* c}(\zet)-\bar{S}^{* c}(\zet)\right\Vert\xrightarrow{\enskip p\enskip}0$. Note that 
\begin{align}\label{B1}
\hat{\V}(\zet)=\V(\hat{\bet}_M,\zet)=\erm(\elw \Y-\X\hat{\bet}_M)=\Q\erm\elw\Y,
\end{align} 
and 
\begin{align}\label{B2}
\bar{\V}(\zet)=\V(\bar{\bet}_M,\zet)=\erm(\elw \Y-\X\bar{\bet}_M)=\Q\erm\elw\Y+\PP\erm\elw(\Y-\E(\Y)),
\end{align} 
where $\PP$ is the projection matrix based on $\erm\X$ and $\Q=\mf{I}_n-\PP$. Denote $\G=\erm\elw$ and $\GZ=\ermz\elwz$. Substituting \eqref{B1} and \eqref{B2} into $S^{* c}(\zet)$ and $\bar{S}^{* c}(\zet)$, the proof of $\sup_{\zet\in\DDDelta}\frac{1}{n}\left\Vert S^{c*}(\zet)-\bar{S}^{c*}(\zet)\right\Vert\xrightarrow{\enskip p\enskip}0$  is equivalent to that of the following:
\begin{enumerate}
\item[(i)] $\sup_{\zet\in\boldsymbol{\Delta}}\frac{1}{n}\left(\Y^{'}\RR_i(\zet)\Y-\E\left(\Y^{'}\RR_i(\zet)\Y\right)\right)=o_p(1)$, for $i=1,2$,
\item[(ii)] $\sup_{\zet\in\boldsymbol{\Delta}}\frac{1}{n}\tr\left(\SSigma\GZPI\TT_i(\zet)\GZI\right)=o(1)$, for $i=1,2,3$,
\end{enumerate}
where the terms are defined as  $\RR_1(\zet)=\GP\WW_D(\rho)\Q\G$, $\RR_2(\zet)=\GP\Q\M\Q\G$, $\TT_1(\zet)=\GP\WW_D(\rho)\PP\G$, $\TT_2(\zet)=\GP\PP\G$ and $\TT_3(\zet)=\GP\Q\M^s\PP\G$.

\noindent
\textbf{Proof of (i).} Note that (i) follows from the point-wise convergence of $\frac{1}{n}\left(\Y^{'}\RR_i(\zet)\Y-\E(\Y^{'}\RR_i(\zet)\Y)\right)$ in each $\zet\in\DDDelta$ and stochastic equicontinuity of $\frac{1}{n}\Y^{'}\RR_i(\zet)\Y$ for $i=1,2$. To prove the point-wise convergence, we have  
\begin{align*}
&\frac{1}{n}\left(\Y^{'}\RR_i(\zet)\Y-\E(\Y^{'}\RR_i(\zet)\Y)\right)\\
&=\frac{2}{n}\beta_0^{'}\X^{'}\elwzp\RR_i(\zet)\GZI\V+\frac{1}{n}\left(\V^{'}\GZPI\RR_i(\zet)\GZI\V-\tr(\SSigma\GZPI\RR_i(\zet)\GZI)\right).
\end{align*}
Note that $\elwzp\RR_i(\zet)\GZI$ and $\GZPI\RR_i(\zet)\GZI$  are uniformly bounded in row and column sum norms for $i=1,2$ by Lemma \ref{l1}. Thus, the terms on the r.h.s.\ are pointwise convergent by Lemma \ref{l5}(2).

To prove the stochastic equicontinuity, by  the mean value theorem, for any two parameter vectors $\zet_1, \zet_2 \in \DDDelta$,  we have 
\begin{align*}
\frac{1}{n}\left(\Y^{'}\RR_i(\zet_1)\Y-\Y^{'}\RR_i(\zet_2)\Y\right)=\frac{1}{n}\Y^{'}\frac{\partial \RR_i(\bar{\zet})}{\partial \zet^{'}}\Y(\zet_1-\zet_2),
\end{align*}
where $\bar{\zet}$ is between $\zet_1$ and $\zet_2$ elementwise. Thus we need to prove that $\sup_{\zet\in\DDDelta}\frac{1}{n}\Y^{'}\frac{\partial \RR_i(\zet)}{\partial \lambda}\Y=O_p(1)$ and $\sup_{\zet\in\DDDelta}\frac{1}{n}\Y^{'}\frac{\partial \RR_i(\zet) }{\partial \rho}\Y=O_p(1)$ for $i=1,2$. Note that 
\begin{align*}
&\frac{\partial \RR_1(\zet) }{\partial \lambda}=\W^{'}\RR_1(\zet)+\RR_1(\zet)\W, \quad \frac{\partial \RR_2(\zet) }{\partial \lambda}=\W^{'}\RR_2(\zet)+\RR_2(\zet)\W,\\
&\frac{\partial \RR_1(\zet) }{\partial \rho}=\M^{'}\RR_1(\zet)+\GP\dot{\WW}_D(\rho)\Q\G+\GP\WW_D(\rho)\dot{\mathbf{Q}}(\rho)\G+\RR_1(\zet)\M,\\
&\frac{\partial \RR_2(\zet) }{\partial \rho}=\M^{'}\RR_2(\zet)+\GP\dot{\mathbf{Q}}(\rho)\M\Q\G+\GP\Q\M\dot{\mathbf{Q}}(\rho)\G+\RR_2(\zet)\M,
\end{align*}
where $\dot{\WW}_D(\rho)=\frac{\partial \WW_D(\rho)}{\partial \rho}=\M\WW_D(\rho)-\WW_D(\rho)\M-\Diag\left(\M\WW_D(\rho)-\WW_D(\rho)\M\right)$ and $\dot{\mathbf{Q}}(\rho)=\frac{\partial \Q}{\partial \rho}=-(\Q\M\PP+\PP\M^{'}\Q)$.
By substituting the reduced form $\Y=\elwzn(\X\bet_0+\ermzn\V)$ into $\Y^{'}\frac{\partial \RR_i(\zet)}{\partial \lambda}\Y$ and $\Y^{'}\frac{\partial \RR_i(\zet)}{\partial \rho}\Y$ for $i=1,2$, we get a group of nonstochastic terms and linear and quadratic forms in $\V$. By Lemma \ref{l5}, $\sup_{\zet\in\DDDelta}\frac{1}{n}\Y^{'}\frac{\partial \RR_i(\zet) }{\partial \lambda}\Y=O_p(1)$ and $\sup_{\zet\in\DDDelta}\frac{1}{n}\Y^{'}\frac{\partial \RR_i(\zet) }{\partial \rho}\Y=O_p(1)$ for $i=1,2$.

\noindent
\textbf{Proof of (ii).} Under Assumption~\ref{a5}, Lemma~\ref{l1} ensures that $\frac{1}{n}\tr\left(\SSigma\GZPI\TT_i(\zet)\GZI\right)=o(1)$, for $i=1,2,3$.

\subsection{Proof of Theorem \ref{thm2}}\label{appB2}
By the mean value theorem, $\sqrt{n}(\hat{\gam}_M-\gam_0)=-\left(\frac{1}{n}\frac{\partial S^{*}(\overline{\gam})}{\partial \gam^{'}}\right)^{-1}\frac{1}{\sqrt{n}}S^{*}(\gam_0)$, where $\overline{\gam}$ is between $\hat{\gam}_M$ and $\gam_0$  elementwise. Thus we need to prove:
\begin{enumerate}
\item[(i)]  $\frac{1}{\sqrt{n}}S^{*}(\gam_0)\xrightarrow{\enskip d\enskip}N\left(0,\lim_{N\rightarrow\infty}\OOmega(\gam_0)\right),$
\item[(ii)]  $\frac{1}{n}\left(\frac{\partial S^{*}(\overline{\gam})}{\partial \gam^{'}}-\frac{\partial S^{*}(\gam_0)}{\partial \gam^{'}}\right)=o_p(1),$
\item[(iii)]  $\frac{1}{n}\left(\frac{\partial S^{*}(\gam_0)}{\partial \gam^{'}}-\E\left(\frac{\partial S^{*}(\gam_0)}{\partial \gam^{'}}\right)\right)=o_p(1).$
\end{enumerate}
\noindent
\textbf{Proof of (i).} Note that the elements of $S^{*}(\gam_0)$ are linear-quadratic forms of $\V$ as shown in \eqref{4.10}. Then we can construct an $(k+2)\times 1$ vector $\mf{a}=(\mf{a}_1^{'},a_2,a_3)^{'}$ such that $\mf{a}^{'}S^{*}(\gam_0)=\mf{b}^{'}\V+\V^{'}\mf{B}\V$, where $\mf{b}^{'}=\mf{a}_1^{'}\X^{'}\ermzp-a_2\bet_0^{'}\X^{'}\ermzp\WW_D$ and $\mf{B}=-a_2\WW_D-a_3\M$. Then, $\frac{1}{n}\mf{a}^{'}S^{*}(\gam_0)$ is asymptotically normal by Lemma~\ref{l4}. Then, the Cramer-Wold device leads to (i).

\noindent
\textbf{Proof of (ii).} Let $\bs{\Pi}(\gam)=-\frac{1}{n}\frac{\partial S^*(\gam)}{\partial \gam^{'}}$.  Since $\Y=\elwzn(\X\bet_0+\ermzn\V)$, all terms in the Hessian matrix can be written in forms of functions in Lemma \ref{l5}. By Lemma \ref{l5}, we know that $\frac{1}{n}\boldsymbol{\Pi}(\gam_0)=O_p(1)$, which implies $\frac{1}{n}\boldsymbol{\Pi}(\bar{\gam})=O_p(1)$. We can write ${e^{\bar{\lambda}\mathbf{W}}}=({e^{\bar{\lambda}\mathbf{W}}}-\elwz)+\elwz$,  ${e^{\bar{\rho}\mathbf{M}}}=({e^{\bar{\rho}\mathbf{M}}}-\ermz)+\ermz$ and $\bar{\bet}=(\bar{\bet}-\bet_0)+\bet_0$, and then expand the terms in $\frac{1}{n}\boldsymbol{\Pi}(\gam)$. By Lemma \ref{l5} and the reduced form of $\Y$, $\frac{1}{n}\Y^{'}\A\Y=O_p(1)$ and $\frac{1}{n}\X^{'}\A\Y=O_p(1)$, where $A$ is an $n\times n$ matrix that is bounded in both row and column sum matrix norms. Also note $\left\Vert {e^{\bar{\lambda}\mathbf{W}}}-\elw \right\Vert_{\infty}=\left\Vert (e^{(\bar{\lambda}-\lambda_0)\mathbf{W}}-I_n)e^{\lambda_0\mathbf{W}}\right\Vert_{\infty}\le\left\Vert (e^{(\bar{\lambda}-\lambda_0)\mathbf{W}}-I_n\right\Vert_{\infty}\left\Vert e^{\lambda_0\mathbf{W}}\right\Vert_{\infty}=o_p(1)$  by Lemma \ref{l3}, and similarly $\left\Vert {e^{\bar{\rho}\mathbf{M}}}-\ermz \right\Vert_{\infty}=o_p(1)$. Then the expanded forms of $\frac{1}{n}\left(\boldsymbol{\Pi}(\bar{\gam})-\boldsymbol{\Pi}(\gam_0)\right)$ imply that it is $o_p(1)$. 

\noindent
\textbf{Proof of (iii).} Substituting the reduced form of $\Y$ into $\frac{1}{n}\left(\frac{\partial S^{*}(\gam_0)}{\partial \gam^{'}}-\E\left(\frac{\partial S^{*}(\gam_0)}{\partial \gam^{'}}\right)\right)$, we know that each element is a linear or quadratic function of $\V$. For example, for $\boldsymbol{\Pi}^{*}_{\lambda\rho}(\gam_0)$, 
\begin{align*}
&\boldsymbol{\Pi}^{*}_{\lambda\rho}(\gam_0)-\E\left(\boldsymbol{\Pi}^{*}_{\lambda\rho}(\gam_0)\right)=-\bet_0^{'}\X^{'}\ermzp\M^{'}\WW_D(\rho)\V-\V^{'}\M^{'}\WW_D(\rho)\V+\E\left(\V^{'}\M^{'}\WW_D(\rho)\V\right)\\
&-\bet_0^{'}\X^{'}\ermzp\dot{\WW}_D(\rho)\V-\V^{'}\dot{\WW}_D(\rho)\V+\E\left(\V^{'}\M^{'}\dot{\WW}_D(\rho)\V\right)-\bet_0^{'}\X^{'}\ermzp\WW_D(\rho)\M\V\\
&-\V^{'}\WW_D(\rho)\M\V+\E\left(\V^{'}\WW_D(\rho)\M\V\right).
\end{align*}
By Lemma \ref{l5}, $\frac{1}{n}\left(\boldsymbol{\Pi}^{*}_{\lambda\rho}(\gam_0)-\E\left(\boldsymbol{\Pi}^{*}_{\lambda\rho}(\gam_0)\right)\right)=o_p(1)$. The proof for the rest of the terms are similar to that for $\boldsymbol{\Pi}^{*}_{\lambda\rho}(\gam_0)$ and thus are omitted.

\subsection{Proof of Theorem \ref{thm3}}\label{appB3}
Since the terms in $\OOmega(\gam_0)$ are similar to those in Proposition 5 in \citet{Jin:2015}, the proof is similar to that of Proposition 5 and thus is omitted.

\section{Details of the Identification Conditions}\label{appC}
The identification condition in  Assumption 6 under the heteroskedastic error terms  is a high level assumption. In this section, we derive low level conditions for the identification of $\zet_0$. Note that the identification of $\zet_0$ requires that $\bar{\Ss}^{c*}(\zet)\ne0$ for $\zet\ne\zet_0$ under the exact identification case, similar to the method of moment approach. Also recall that the population counterpart of the concentrated adjusted score function is given by 
\begin{align*}
\bar{S}^{c*}(\zet)=
\left\{\begin{array}{rl}
\lambda: &-\E\left(\Y^{'}\elwp\ermp\WW_D(\rho)\bar{\V}(\zet)\right),\\
\rho: &-\E\left(\bar{\V}^{'}(\zet)\M\bar{\V}(\zet)\right),
\end{array}\right.
\end{align*}
where $\bar{\V}(\zet)=\V(\bar{\bet}_M(\zet),\zet)$. 
Also recall that \eqref{B2} implies that 
\begin{align*}
\bar{\V}(\zet)=\Q\G\Y+\PP\G(\Y-\E(\Y)),
\end{align*} 
where $\G=\erm\elw$. Denote $\mathbf{G}=\GZ$. From the reduced form $\Y=\elwzn\X\bet_0+\mathbf{G}^{-1}\V$, we know $\Y-\E(\Y)=\mathbf{G}^{-1}\V$. Then,
\begin{align*}
\bar{\V}(\zet)&=\Q\G\Y+\PP\G\mathbf{G}^{-1}\V\\
&=\Q\G(\elwzn\X\bet_0+\mathbf{G}^{-1}\V)+\PP\G\mathbf{G}^{-1}\V\\
&=\Q\G\elwzn\X\bet_0+\G\mathbf{G}^{-1}\V.
\end{align*} 
Then we have
\begin{align*}
&\E(\bar{\V}^{'}(\zet)\mathbb{W}_D(\rho)\bar{\V}(\zet))\\
&=\E\left((\Q\G\elwzn\X\bet_0+\G\mathbf{G}^{-1}\V)^{'}\mathbb{W}_D(\rho)(\Q\G\elwzn\X\bet_0+\G\mathbf{G}^{-1}\V)\right)\\
&=\bet_0^{'}\X^{'}\elwznp\GP\Q\mathbb{W}_D\Q\G\elwzn\X\bet_0+\E(\V^{'}\mathbf{G}^{-1'}\GP\mathbb{W}_D(\rho)\G\mathbf{G}^{-1}\V)\\
&=\bet_0^{'}\X^{'}\elwznp\GP\Q\mathbb{W}_D\Q\G\elwzn\X\bet_0+\tr(\SSigma\mathbf{G}^{-1'}\GP\mathbb{W}_D(\rho)\G\mathbf{G}^{-1}).
\end{align*} 
Similarly $\E(\bar{\V}^{'}(\zet)\M\bar{\V}(\zet))$ can be expressed as 
\begin{align*}
&\E(\bar{\V}^{'}(\zet)\M(\rho)\bar{\V}(\zet))\\
&=\bet_0^{'}\X^{'}\elwznp\GP\Q\M\Q\G\elwzn\X\bet_0+\tr(\SSigma\mathbf{G}^{-1'}\GP\M\G\mathbf{G}^{-1}).
\end{align*} 
Thus, the identification of $\zet_0$ follows, if $\zet\ne\zet_0$, one of the following conditions holds:
\begin{align*}
&\text{(i)} \quad \lim_{n\to\infty}\frac{1}{n}[\bet_0^{'}\X^{'}\elwznp\GP\Q\mathbb{W}_D\Q\G\elwzn\X\bet_0+\tr(\SSigma\mathbf{G}^{-1'}\GP\mathbb{W}_D(\rho)\G\mathbf{G}^{-1})]\ne0,\\
&\text{(ii)}\quad \lim_{n\to\infty}\frac{1}{n}[ \bet_0^{'}\X^{'}\elwznp\GP\Q\M\Q\G\elwzn\X\bet_0+\tr(\SSigma\mathbf{G}^{-1'}\GP\M\G\mathbf{G}^{-1}) ] \ne0.
\end{align*}

\newpage
\bibliography{references}\label{sec:bib}

\begin{thebibliography}{}

\bibitem[Anselin, 1984]{Anselin:1984}
Anselin, L. (1984).
\newblock Specification tests on the structure of interaction in spatial
  econometric models.
\newblock {\em Papers of the Regional Science Association}, 54(1):165--182.

\bibitem[Anselin, 1986]{Anselin:1986}
Anselin, L. (1986).
\newblock Non-nested tests on the weight structure in spatial autoregressive
  models: Some monte carlo results.
\newblock {\em Journal of Regional Science}, 26(2):267--284.

\bibitem[Anselin, 1988]{Anselin:1988}
Anselin, L. (1988).
\newblock {\em Spatial Econometrics: Methods and Models}.
\newblock Springer, New York.

\bibitem[Anselin, 2001]{Anselin:2001}
Anselin, L. (2001).
\newblock Rao's score test in spatial econometrics.
\newblock {\em Journal of Statistical Planning and Inference}, 97(1):113--139.

\bibitem[Anselin et~al., 1996]{Anselin:1996}
Anselin, L., Bera, A.~K., Florax, R., and Yoon, M.~J. (1996).
\newblock Simple diagnostic tests for spatial dependence.
\newblock {\em Regional Science and Urban Economics}, 26(1):77--104.

\bibitem[Arbia et~al., 2020]{Arbia:2020}
Arbia, G., Bera, A.~K., Doğan, O., and Taşpınar, S. (2020).
\newblock Testing impact measures in spatial autoregressive models.
\newblock {\em International Regional Science Review}, 43(1-2):40--75.

\bibitem[{Arraiz} et~al., 2010]{Arraiz:2010}
{Arraiz}, I., {Drukker}, D.~M., {Kelejian}, H.~H., and {Prucha}, I.~R. (2010).
\newblock A spatial {Cliff-Ord-Type} model with heteroskedastic innovations:
  Small and large sample results.
\newblock {\em Journal of Regional Science}, 50(2):592--614.

\bibitem[Burridge, 2012]{Burridge:2012}
Burridge, P. (2012).
\newblock Improving the {J} test in the {SARAR} model by likelihood-based
  estimation.
\newblock {\em Spatial Economic Analysis}, 7(1):75--107.

\bibitem[Chan and Grant, 2016]{Chan:2016}
Chan, J. C.~C. and Grant, A.~L. (2016).
\newblock On the observed-data deviance information criterion for volatility
  modeling.
\newblock {\em Journal of Financial Econometrics}, 14(4):772--802.

\bibitem[Chib and Jeliazkov, 2001]{Chib:2001}
Chib, S. and Jeliazkov, I. (2001).
\newblock Marginal likelihood from the {Metropolis}–{Hastings} output.
\newblock {\em Journal of the American Statistical Association},
  96(453):270--281.

\bibitem[Cliff and Ord, 1969]{Cliff:1969}
Cliff, A.~D. and Ord, J. (1969).
\newblock The problem of spatial autocorrelation.
\newblock In Scott, A.~J., editor, {\em London Papers in Regional Science 1
  Studies in Regional Science}, pages 22--55. Pion, London.

\bibitem[Cliff and Ord, 1973]{Cliff:1973}
Cliff, A.~D. and Ord, J. (1973).
\newblock {\em Spatial autocorrelation}.
\newblock Pion, London.

\bibitem[Davidson and MacKinnon, 1981]{Davidson:1981}
Davidson, R. and MacKinnon, J.~G. (1981).
\newblock Several tests for model specification in the presence of alternative
  hypotheses.
\newblock {\em Econometrica}, 49(3):781--793.

\bibitem[Debarsy et~al., 2015]{Jin:2015}
Debarsy, N., Jin, F., and Lee, L.-F. (2015).
\newblock Large sample properties of the matrix exponential spatial
  specification with an application to {FDI}.
\newblock {\em Journal of Econometrics}, 188(1):1--21.

\bibitem[Dogan et~al., 2018]{Dogan:2018}
Dogan, O., Taspinar, S., and Bera, A.~K. (2018).
\newblock Simple tests for social interaction models with network structures.
\newblock {\em Spatial Economic Analysis}, 13(2):212--246.

\bibitem[Doğan, 2023]{Dogan:2023b}
Doğan, O. (2023).
\newblock Modified harmonic mean method for spatial autoregressive models.
\newblock {\em Economics Letters}, 223:110978.

\bibitem[Doğan et~al., 2023]{Dogan:2023}
Doğan, O., Yang, Y., and Taşpınar, S. (2023).
\newblock Information criteria for matrix exponential spatial specifications.
\newblock {\em Spatial Statistics}, 57:100776.

\bibitem[Elhorst, 2014]{Elhorst:2014}
Elhorst, J.~P. (2014).
\newblock {\em Spatial Econometrics: From Cross-Sectional Data to Spatial
  Panels}.
\newblock Springer Briefs in Regional Science. Springer, New York.

\bibitem[Fan and Li, 2001]{Fan.Li2001}
Fan, J. and Li, R. (2001).
\newblock Variable selection via nonconcave penalized likelihood and its oracle
  properties.
\newblock {\em Journal of the American Statistical Association}, 96:1348--1360.

\bibitem[Gelfand and Dey, 1994]{Dey:1994}
Gelfand, A.~E. and Dey, D.~K. (1994).
\newblock Bayesian model choice: Asymptotics and exact calculations.
\newblock {\em Journal of the Royal Statistical Society. Series B
  (Methodological)}, 56(3):501--514.

\bibitem[Gelfand and Smith, 1990]{Smith:1990}
Gelfand, A.~E. and Smith, A. F.~M. (1990).
\newblock Sampling-based approaches to calculating marginal densities.
\newblock {\em Journal of the American Statistical Association},
  85(410):398--409.

\bibitem[Gelman et~al., 2003]{Gelman:2003}
Gelman, A., Carlin, J., Stern, H., Dunson, D., Vehtari, A., and Rubin, D.
  (2003).
\newblock {\em Bayesian Data Analysis}.
\newblock Chapman \& Hall/CRC Texts in Statistical Science. Taylor \& Francis,
  New York, third edition edition.

\bibitem[Geweke, 1999]{Geweke:1999}
Geweke, J. (1999).
\newblock Using simulation methods for bayesian econometric models: Inference,
  development,and communication.
\newblock {\em Econometric Reviews}, 18(1):1--73.

\bibitem[Han and Lee, 2013a]{Han:2013}
Han, X. and Lee, L.-f. (2013a).
\newblock Bayesian estimation and model selection for spatial {Durbin} error
  model with finite distributed lags.
\newblock {\em Regional Science and Urban Economics}, 43(5):816--837.

\bibitem[Han and Lee, 2013b]{HAN2013250}
Han, X. and Lee, L.-f. (2013b).
\newblock Model selection using {J}-test for the spatial autoregressive model
  vs. the matrix exponential spatial model.
\newblock {\em Regional Science and Urban Economics}, 43(2):250--271.

\bibitem[Hansen, 1982]{Hansen:1982}
Hansen, L.~P. (1982).
\newblock Large sample properties of generalized method of moments estimators.
\newblock {\em Econometrica}, 50(4):1029--1054.

\bibitem[Hepple, 1995]{Hepple:1995}
Hepple, W.~L. (1995).
\newblock Bayesian techniques in spatial and network econometrics: 1. model
  comparison and posterior odds.
\newblock {\em Environment and Planning A: Economy and Space}, 27(3):447--469.

\bibitem[Hsu and Shi, 2017]{Shi:2017}
Hsu, Y. and Shi, X. (2017).
\newblock Model‐selection tests for conditional moment restriction models.
\newblock {\em The Econometrics Journal}, 20(1):52--85.

\bibitem[Jin and Lee, 2013]{Jin:2013}
Jin, F. and Lee, L.-f. (2013).
\newblock Cox-type tests for competing spatial autoregressive models with
  spatial autoregressive disturbances.
\newblock {\em Regional Science and Urban Economics}, 43(4):590--616.

\bibitem[Jin and Lee, 2018]{Jin.Lee2018}
Jin, F. and Lee, L.-F. (2018).
\newblock Irregular {N2SLS} and {LASSO} estimation of the matrix exponential
  spatial specification model.
\newblock {\em Journal of Econometrics}, 206(2):336--358.

\bibitem[Jin and Wang, 2022]{Jin.Wang2022}
Jin, F. and Wang, Y. (2022).
\newblock {GMM} estimation of a spatial autoregressive model with
  autoregressive disturbances and endogenous regressors.
\newblock {\em Econometric Reviews}, 41(6):652--674.

\bibitem[Kass and Raftery, 1995]{Kass:1995}
Kass, R.~E. and Raftery, A.~E. (1995).
\newblock Bayes factors.
\newblock {\em Journal of the American Statistical Association},
  90(430):773--795.

\bibitem[Kelejian, 2008]{Kelejian:2008}
Kelejian, H.~H. (2008).
\newblock A spatial {J}-test for model specification against a single or a set
  of non-nested alternatives.
\newblock {\em Letters in Spatial and Resource Sciences}, 1:3--11.

\bibitem[Kelejian and Piras, 2011]{Kelejian:2011}
Kelejian, H.~H. and Piras, G. (2011).
\newblock An extension of kelejian's {J}-test for non-nested spatial models.
\newblock {\em Regional Science and Urban Economics}, 41(3):281--292.

\bibitem[Kelejian and Prucha, 2010]{KP:2010}
Kelejian, H.~H. and Prucha, I. (2010).
\newblock Specification and estimation of spatial autoregressive models with
  autoregressive and heteroskedastic disturbances.
\newblock {\em Journal of Econometrics}, 157(1):53--67.

\bibitem[{Kelejian} and {Prucha}, 1999]{KP:1999}
{Kelejian}, H.~H. and {Prucha}, I.~R. (1999).
\newblock A generalized moments estimator for the autoregressive parameter in a
  spatial model.
\newblock {\em International Economic Review}, 40(2):509--533.

\bibitem[Kelejian and Prucha, 2001]{KP:2001}
Kelejian, H.~H. and Prucha, I.~R. (2001).
\newblock On the asymptotic distribution of the moran {I} test statistic with
  applications.
\newblock {\em Journal of Econometrics}, 104(2):219--257.

\bibitem[Lee, 2002]{Lee:2002}
Lee, L.-f. (2002).
\newblock Consistency and efficiency of least squares estimation for mixed
  regressive, spatial autoregressive models.
\newblock {\em Econometric Theory}, 18(2):252--277.

\bibitem[Lee, 2007]{Lee:2007a}
Lee, L.-f. (2007).
\newblock {GMM} and {2SLS} estimation of mixed regressive, spatial
  autoregressive models.
\newblock {\em Journal of Econometrics}, 137(2):489--514.

\bibitem[LeSage and Chih, 2018]{Chih:2018}
LeSage, J. and Chih, Y.-Y. (2018).
\newblock A matrix exponential spatial panel model with heterogeneous
  coefficients.
\newblock {\em Geographical Analysis}, 50(4):422--453.

\bibitem[Lesage, 1997]{Lesage:1997}
Lesage, J.~P. (1997).
\newblock Bayesian estimation of spatial autoregressive models.
\newblock {\em International Regional Science Review}, 20(1-2):113--129.

\bibitem[LeSage and Pace, 2007]{Lesage:2007}
LeSage, J.~P. and Pace, R.~K. (2007).
\newblock A matrix exponential spatial specification.
\newblock {\em Journal of Econometrics}, 140(1):190--214.

\bibitem[LeSage and Pace, 2009]{Lesage:2009}
LeSage, J.~P. and Pace, R.~K. (2009).
\newblock {\em Introduction to Spatial Econometrics}.
\newblock Chapman and Hall/CRC, London.

\bibitem[LeSage and Parent, 2007]{LeSage:2007b}
LeSage, J.~P. and Parent, O. (2007).
\newblock Bayesian model averaging for spatial econometric models.
\newblock {\em Geographical Analysis}, 39(3):241--267.

\bibitem[Li et~al., 2020]{Li:2020}
Li, Y., Yu, J., and Zeng, T. (2020).
\newblock Deviance information criterion for latent variable models and
  misspecified models.
\newblock {\em Journal of Econometrics}, 216(2):450--493.

\bibitem[Lin and Lee, 2010]{Lin:2010}
Lin, X. and Lee, L.-f. (2010).
\newblock {GMM} estimation of spatial autoregressive models with unknown
  heteroskedasticity.
\newblock {\em Journal of Econometrics}, 157(1):34--52.

\bibitem[Liu and Lee, 2019]{LIU2019434}
Liu, T. and Lee, L.-f. (2019).
\newblock A likelihood ratio test for spatial model selection.
\newblock {\em Journal of Econometrics}, 213(2):434--458.

\bibitem[MacKinnon, 2009]{Mac:2009}
MacKinnon, J.~G. (2009).
\newblock Bootstrap hypothesis testing.
\newblock In Belsley, D.~A. and Kontoghiorghes, E.~J., editors, {\em Handbook
  of Computational Econometrics}, pages 183--210. John Wiley \& Sons Ltd, West
  Sussex.

\bibitem[Moler and Van~Loan, 1978]{Moler:1978}
Moler, C. and Van~Loan, C. (1978).
\newblock Nineteen dubious ways to compute the exponential of a matrix.
\newblock {\em SIAM Review}, 20(4):801--836.

\bibitem[Moler and Van~Loan, 2003]{Moler:2003}
Moler, C. and Van~Loan, C. (2003).
\newblock Nineteen dubious ways to compute the exponential of a matrix,
  twenty-five years later.
\newblock {\em SIAM Review}, 45(1):3--49.

\bibitem[Newey and West, 1987]{Newey:1987}
Newey, W.~K. and West, K.~D. (1987).
\newblock Hypothesis testing with efficient method of moments estimation.
\newblock {\em International Economic Review}, 28(3):777--787.

\bibitem[Schwarz, 1978]{Schwarz:1978}
Schwarz, G. (1978).
\newblock Estimating the dimension of a model.
\newblock {\em The Annals of Statistics}, 6(2):461 -- 464.

\bibitem[Spiegelhalter et~al., 2002]{Spiegelhalter:2002}
Spiegelhalter, D.~J., Best, N.~G., Carlin, B.~P., and Van Der~Linde, A. (2002).
\newblock Bayesian measures of model complexity and fit.
\newblock {\em Journal of the Royal Statistical Society: Series B (Statistical
  Methodology)}, 64(4):583--639.

\bibitem[Verdinelli and Wasserman, 1995]{VW:1995}
Verdinelli, I. and Wasserman, L. (1995).
\newblock Computing {Bayes} factors using a generalization of the
  {Savage-Dickey} density ratio.
\newblock {\em Journal of the American Statistical Association},
  90(430):614--618.

\bibitem[White, 1994]{white:1994}
White, H. (1994).
\newblock {\em Estimation, inference and specification analysis}.
\newblock Econometric Society Monographs.

\bibitem[{White}, 1980]{White:1980}
{White}, H.~G. (1980).
\newblock A heteroskedasticity-consistent covariance matrix estimator a direct
  test for heteroskedasticity.
\newblock {\em Econometrica}, 48(4):817--838.

\bibitem[Whittle, 1954]{Whittle:1954}
Whittle, P. (1954).
\newblock On stationary processes in the plane.
\newblock {\em Biometrika}, 41(3/4):434--449.

\bibitem[Yang et~al., 2021]{yangfast}
Yang, Y., Dogan, O., and Taspinar, S. (2021).
\newblock Fast estimation of matrix exponential spatial models.
\newblock {\em Journal of Spatial Econometrics}, 2(9).

\bibitem[Yang et~al., 2022]{yangms}
Yang, Y., Dogan, O., and Taspinar, S. (2022).
\newblock Model selection and model averaging for matrix exponential spatial
  models.
\newblock {\em Econometric Reviews}, 41(8):827--858.

\bibitem[Yang et~al., 2023]{yangunbalanced}
Yang, Y., Dogan, O., and Taspinar, S. (2023).
\newblock Estimation of matrix exponential unbalanced panel data models with
  fixed effects: an application to {US} outward {FDI} stock.
\newblock {\em Journal of Business \& Economic Statistics}, 0(0):1--16.

\bibitem[Yang, 2018]{Zhenlin:2018}
Yang, Z. (2018).
\newblock Unified {M}-estimation of fixed-effects spatial dynamic models with
  short panels.
\newblock {\em Journal of Econometrics}, 205(2):423--447.

\bibitem[Zhang et~al., 2019]{Zhang2019}
Zhang, Y., Feng, S., and Jin, F. (2019).
\newblock {QML} estimation of the matrix exponential spatial specification
  panel data model with fixed effects and heteroskedasticity.
\newblock {\em Economics Letters}, 180:1--5.

\end{thebibliography}

\end{document}